\documentclass[10pt,conference,letterpaper]{IEEEtran}
 \pdfoutput=1 
\makeatletter
\def\ps@headings{%
\def\@oddhead{\mbox{}\scriptsize\rightmark \hfil \thepage}%
\def\@evenhead{\scriptsize\thepage \hfil \leftmark\mbox{}}%
\def\@oddfoot{}%
\def\@evenfoot{}}
\makeatother

\makeatletter
\def\hlinew#1{%
  \noalign{\ifnum0=`}\fi\hrule \@height #1 \futurelet
   \reserved@a\@xhline}

\usepackage{graphicx}
\usepackage[noadjust]{cite}
\usepackage{tabularx}
\usepackage{times}
\usepackage{alltt}
\usepackage{verbatim}
\usepackage{moreverb}
\usepackage{amsmath}
\usepackage{amssymb}
\ifCLASSINFOpdf \else \fi
\usepackage{url}
\usepackage{epsfig}
\usepackage{subfigure}
\usepackage{multirow}
\usepackage[ruled,vlined]{algorithm2e}
\usepackage{mathrsfs}
\usepackage{amsthm}
\usepackage{subfig}
\usepackage{algorithmic}
\usepackage{dsfont}
\usepackage{cases}
\usepackage{setspace}
\usepackage{mathrsfs}
\usepackage{dsfont}
\usepackage{graphicx}
\usepackage{subfigure}
\usepackage{booktabs}
\usepackage[usenames]{color}
\newtheorem{theorem}{Theorem}
\newtheorem{definition}{Definition}
\newtheorem{lemma}{Lemma}
\newtheorem{Proposition}{Proposition}
\newtheorem{corollary}{Corollary}

\begin{document}
\title{De-anonymizing Social Networks with Overlapping Community Structure}
\date{}
\author{Luoyi Fu$^{1}$, Xinyu Wu$^{2}$, Zhongzhao Hu$^{1}$, Xinzhe Fu$^{1}$, Xinbing Wang$^{1,2}$\\
$^{1,2}$Dept. of \{Computer Science, Electronic Engineering\}, Shanghai Jiao Tong University, China. \\
Email:$^{1}$\{hzz5611577,fxz0114,yiluofu,xwang8\}@sjtu.edu.cn, $^{2}$wuxinyu@sjtu.edu.cn. \\
}
\thanks{This work was accepted by IEEE International Conference of Computer Communication (INFOCOM) 2018.}
\maketitle

\begin{abstract}
The advent of social networks poses severe threats on
user privacy as adversaries can de-anonymize users' identities
by mapping them to correlated cross-domain networks. Without ground-truth mapping, prior literature proposes various cost
functions in hope of measuring the quality of mappings. However,
there is generally a lacking of rationale behind the cost functions,
whose minimizer also remains algorithmically unknown.

We jointly tackle above concerns under a more practical social network model parameterized by \emph{overlapping communities}, which, neglected by prior art, can serve as side information for de-anonymization. Regarding the unavailability of ground-truth mapping to adversaries, by virtue of the Minimum Mean Square Error (MMSE), our first contribution is a well-justified cost function minimizing the expected number of mismatched users over all possible true mappings.
While proving the NP-hardness of minimizing MMSE, we validly transform it into the weighted-edge matching problem (WEMP), which, as disclosed theoretically, resolves the tension between optimality and complexity:
(i) WEMP asymptotically returns a negligible mapping error in large network size under mild conditions facilitated by higher overlapping strength; (ii) WEMP can be algorithmically characterized via the convex-concave based de-anonymization algorithm (CBDA), perfectly finding the optimum of WEMP. Extensive experiments further confirm the effectiveness of CBDA under overlapping communities, in terms of
averagely $90\%$ re-identified users in the rare true cross-domain co-author networks when communities overlap densely, and roughly $70\%$ enhanced re-identification ratio compared to
non-overlapping cases.

	\end{abstract}
	\section{Introduction}
		With the mounting popularity of social networks, the privacy of users has been under great concern, as information of users in social networks is often released to public for wide usage in academy or advertisement \cite{cite:arxiv-community,privacy}. Although users can be anonymized by removing personal identifiers such as names and family addresses, it is not sufficient for privacy protection since adversaries may re-identify these users by correlated side information, for example the cross domain networks where the identities of these users are unveiled \cite{cite:arxiv-community}.

Such user identification process in social networks resorting to auxiliary information is called \emph{Social Network De-anonymization}. Initially proposed by Narayanan and Shimatikov \cite{cite:de-anonymization}, this fundamental issue has then gained increasing attention, leading to a large body of subsequent works \cite{cite:seedless,cite:allerton,cite:shouling1,cite:shouling2,cite:improved-bound,cite:arxiv-community,palla2005uncovering}.  Particularly, this family of works embarked on de-anonymization under a common framework, as will also be the framework of interest in our setting. To elaborate, in the framework there is an underlying network $G$ which characterizes the relationship among users. Then there are two networks observed in reality, named as published network $G_1$ and auxiliary network $G_2$, whose node sets are identical and edges are independently sampled from $G$ with probability $s_1$ and $s_2$ respectively. \emph{The aim of de-anonymization is to discover the correct mapping between $V_1$ and $V_2$, which corresponds the same user in two networks, with the network structure as the only side information available to the adversaries.}

Regardless of the considerable efforts paid to de-anonymization, there is still a severe lacking of a comprehensive understanding about the conditions under which the adversaries can perfectly de-anonymize user identities. It can be accounted for from three aspects. \underline{(i) Analytically,} despite a variety of existing work \cite{cite:seedless,cite:allerton} that proposed several cost functions in measuring the quality of mappings, the theoretical devise of those costs functions
lacks sufficient rationale behind. \underline{(ii) Algorithmically,} previous works \cite{cite:seedless,cite:allerton} failed to provide any algorithm to demonstrate that the optimal solution of proposed cost functions can indeed be effectively obtained. \underline{(iii) Experimentally,} due to the destitution of real cross-domain datasets, state-of-the-art research \cite{cite:shouling1,cite:shouling2} simply evaluated the performance of proposed algorithms on synthetic datasets or real cross-domain networks formed by artificial sampling, falling short of reproducing the genuine social networks.

\textbf{The above limitations motivate us to shed light on de-anonymization problem by jointly incorporating analytical, algorithmic and experimental aspects under the common framework noted earlier.}
As far as we know, the only work that shares the closest correlation with us belongs to Fu et. al. \cite{Fu:arxiv,Fu:GC}, who investigated this problem on social networks with non-overlapping communities and derived their cost function from the Maximum A Posterior (MAP) manner. However, we remark that the assumption of disjoint communities fails to reflect the real situation where a user belongs to multiple communities, as observed in massive real situations. For example, in social networks of scientific collaborators \cite{palla2005uncovering}, actors and political blogospheres \cite{latouche2011overlapping}, users might belong to several research groups with different research topics, movies and political parties respectively. Furthermore, while MAP enables adversaries to find the correct mapping with the highest probability, it relies heavily on a prerequisite, i.e., a hypothetically true mapping between the given published and auxiliary networks. However, once the MAP estimation fails to exactly match this ``true'' mapping, then the mapping error becomes unpredictable, with the probability that the estimation deviates largely from the real ground-truth. For the first concern, by adopting the overlapping stochastic block model (OSBM), we allow the communities to overlap arbitrarily, which can well capture a majority of real social networks. For the second concern, we derive our cost function based on Minimum Mean Square Error (MMSE), which minimizes the expected number of mismatched users by incorporating all the possible true mappings between the given published and auxiliary networks. This incorporation, from an average perspective, keeps the estimation of MMSE from significant deviation from any possible hypothetic true mapping.

Hereinafter we unfold our main contributions in analytical, algorithmic and experimental aspects respectively as follows:

\underline{1. Analytically,} we are the first to derive cost function based on MMSE, which justifiably ensures the minimum expected mapping error between our estimation and the ground-truth mapping. Then we demonstrate the NP-hardness of solving MMSE, whose intractability stems mainly from the calculation of all $n!$ possible mappings ($n$ is the total number of users). To cope with the hardness, we simplify MMSE by transforming it into a weighted-edge matching problem (WEMP), with mapping error negatively related to weights.

\underline{2. Algorithmically,} in terms of solving WEMP, we theoretically reveal that WEMP alleviates the tension between optimality and complexity: Solving WEMP ensures optimality since its optimum, in large network size, negligibly deviates from the ground-truth mapping under mild conditions where on average a user belongs to asymptotically non-constant communities.
Meanwhile it reduces complexity since perfectly deriving its optimum only entails a convex-concave based de-anonymization algorithm (CBDA) with polynomial time. The proposed CBDA serves as one of the very few attempts to address the algorithmic characterization, that has long remained open, of de-anonymization without pre-identification.

\underline{3. Experimentally,}
we validate our theoretical findings that minimizing WEMP indeed incurs negligible mapping error in large social networks based on real datasets. Interestingly, we also observe significant benefits that community overlapping effect brings to the performance of CBDA: (i) in notable true cross-domain co-author networks with dense overlapping communities, CBDA can correctly re-identify $90\%$ nodes on average; (ii) the overlapping communities bring about an enhancement of around $70\%$ re-identification ratio compared with non-overlapping cases.

Unlike de-anonymization with pre-identified seed nodes, to which a family of work pays endeavor, no prior knowledge of such seeds complicates this problem, thus leaving many aspects largely unexplored. Meanwhile, theoretical results on such seedless cases in prior art is short of experimental verification. Our work is, as far as we are concerned, the initial devotion to theoretically dissecting seedless cases with overlapping communities, under real cross-domain networks with more than $3000$ nodes.
With novel exploitations of structural information, future design of more efficient mechanisms will be expected to further dilute the limitation of network size.

	\section{Related Works}\label{sec:relatedworks}
Social network de-anonymization problem has been in the dimelight in recent decades. Narayanan and Shimatikove \cite{cite:de-anonymization} formulated this problem initially. They presented its framework and
proposed a generic algorithm, which did not utilize any side information except the network structure
and worked based on some pre-identified nodes, called seed nodes.

Predicated on this seminal paper, a large amount of work emerges focusing on different facets of de-anonymization. One major division is whether the anonymized network is seeded or seedless, i.e., whether pre-identified nodes exist.
For seeded anonymized network, as the pioneering work \cite{cite:de-anonymization}, the common idea to solve the problem is to design algorithms based on \emph{percolation}, which means that the re-identification process starts from the seed nodes and identify their neighbor nodes iteratively until all the nodes are de-anonymized
\cite{cite:de-anonymization,cite:percolation-matching,cite:vldb1,cite:garetto1,cite:vldb}. Yartseva et al. \cite{cite:percolation-matching}, Kazemi et al. \cite{cite:vldb1} and Fabiana et al. \cite{cite:garetto1} studied seeded problem under Erdos-Renyi graph model, while Korula and Lattenzi \cite{cite:vldb} shed light on preferential attachment model.

However, in real situations it is often the case that adversaries are difficult to obtain seeded nodes before de-anonymizing \cite{Fu:arxiv,Fu:GC} due to the limited access to user profiles. For seedless networks, the major methodology is to propose cost functions and obtain an estimation of the correct mapping between two networks by optimizing these cost functions. Pedarsani and Grossglauer \cite{cite:seedless} are the precursors in de-anonymizing seedless networks. They studied this problem under Erdos-Renyi graph and their cost function was the number of mismatched edges. With the same cost function, Kazemi et al. \cite{cite:allerton} considered the situation where the nodes in two networks are overlapping partially, and Cullina and Kiyavash \cite{cite:improved-bound} further investigated the information-theoretic threshold for exact identification in \cite{cite:seedless}. However, the cost functions in \cite{cite:seedless,cite:allerton,cite:improved-bound} were not justified by rationale.
One cost function based on Maximum A Posterior (MAP) has been justified by \cite{cite:arxiv-community,Fu:arxiv,Fu:GC}. Onaran et al. \cite{cite:arxiv-community} theoretically proved the validity of MAP and Fu et al. \cite{Fu:arxiv,Fu:GC} provided two approximation algorithms to solve this problem.

Another facet for de-anonymization problem is the amount of side information adversaries have. A large amount of work \cite{cite:de-anonymization,cite:percolation-matching,cite:vldb1,cite:garetto1,cite:vldb,cite:seedless,cite:allerton,cite:improved-bound}, either in seeded or seedless situations, studied this problem without any side information except the topological structure of two networks, i.e., the edge sets in two networks. However, the clustering effect exists in real social networks, which has not been considered in work above. To incorporate clustering effect, Chiasserini et al. \cite{cite:garetto2} studied clustering under seeded problem and drew the conclusion that the impact of clustering is double-edged, which may dramatically reduce the required seed nodes but make the algorithm more fragile to errors. Onaran et al. \cite{cite:arxiv-community} and Fu et al. \cite{Fu:arxiv,Fu:GC} both studied clustering by modeling it as communities in two networks, and Fu et al. \cite{Fu:arxiv,Fu:GC} showed that the side information of communities makes for higher accuracy of the algorithms intended for seedless problem. However, as far as we know, no existing work has ever focused on overlapping communities, which is omnipresent in real situations, especially the large-scale social networks nowadays.

\section{Preliminaries}\label{Preliminaries}
In this section we introduce some basic definitions and lemmas which will be used in our later analysis.

\subsection{Definitions}
\begin{definition}(\textbf{Trace})
Given an $n \times n$ square matrix $\mathbf{Y}$, the trace of $\mathbf{Y}$ is $\mathbf{tr}\mathbf{Y}=\sum_{i=1}^n \mathbf{Y}_{ii}$, where $\mathbf{Y}_{ii}$ denotes the element at the $i_{th}$ row and $i_{th}$ column of $\mathbf{Y}$.
\end{definition}

\begin{definition}({\textbf{Expectation Over Matrix}})
Given a random matrix variable $\mathbf{A}$ and a function of $\mathbf{A}$, denoted as $f(\mathbf{A})$, then the expectation of $f(\mathbf{A})$ over matrix $\mathbf{A}$ is denoted as $\textbf{E}_{\mathbf{A}}(f(\mathbf{A}))$.
\end{definition}

\begin{definition}(\textbf{Frobenius Norm})
Given an $m \times n$ matrix $\mathbf{X}$, the Frobenius norm of $\mathbf{X}$ is
\begin{equation}
||\mathbf{X}||_F=\sqrt{\sum_{i=1}^m \sum_{j=1}^n(\mathbf{X}_{ij}^2)},
\nonumber
\end{equation}
where $\mathbf{X}_{ij}$ denotes the element at the $i_{th}$ row and $j_{th}$ column of $\mathbf{X}$.
\end{definition}

\begin{definition}(\textbf{Hadamard Product})
Given two $n \times n$ matrices $\mathbf{Y}$ and $\mathbf{Z}$, The Hadamard Product between $\mathbf{Y}$ and $\mathbf{Z}$ is defined as $\forall i,j \in \{1,2,...,n\}, (\mathbf{Y}\circ\mathbf{Z})_{ij}=\mathbf{Y}_{ij}\mathbf{Z}_{ij}$, where $\mathbf{Y}\circ\mathbf{Z}$ is still an $n \times n$ matrix.
\end{definition}

\begin{definition}(\textbf{Approximation Ratio})
Given a maximization problem $\mathcal{I}$ and its optimal value $OPT(\mathcal{I})$, if an algorithm $\mathcal{A}$ outputs a solution $\mathcal{S}$ such that $\mathcal{S}\geq \tau OPT(\mathcal{I})$, where $\alpha\in [0,1]$. Then the approximation ratio of this algorithm $\mathcal{A}$ for problem $\mathcal{I}$ is $\tau$.
\end{definition}

\subsection{Lemmas}
\begin{lemma}(\textbf{Sequence Inequality} \cite{hardy1988inequalities})\label{lemma1}
 For two nonnegative sequences $a_1\leq a_2 \leq a_3 \cdots \leq a_n$ and $b_1 \leq b_2 \leq b_3 \cdots \leq b_n$, let $\eta=\sum_{k=1}^n a_{i_k}b_{j_k}$ where $\{i_1,i_2,...,i_n\}$ and $\{j_1,j_2,...,j_n\}$ are both permutations of $\{1,2,...,n\}$. Then we can obtain the Sequence Inequality that yields to
\begin{equation}
\sum_{k=1}^n a_kb_k \geq \eta \geq \sum_{k=1}^n a_kb_{n+1-k}.\nonumber
\end{equation}
\end{lemma}


\begin{lemma}
\label{lemma2}
Let $A(n)$, $B(n)$, $C(n)$ and $D(n)$ denote four functions with variable $n$, such that $A(n)=o(B(n))$ and $C(n)=o(D(n))$, then when $n \rightarrow \infty$,
\begin{equation}
\frac{A(n)+B(n)}{C(n)+D(n)}=\frac{B(n)}{D(n)}.\nonumber
\end{equation}
\end{lemma}

\begin{lemma}{(\textbf{Stirling's Formula})}
Stirling's formula presents an approximation for the factorial, $n!$, when $n\rightarrow \infty$, as
\begin{equation}
n!\sim \sqrt{2\pi n}(\frac{n}{e})^n.\nonumber
\end{equation}
\end{lemma}

\begin{lemma}
Given an $n \times n$ matrix $\mathbf{R}$ and an $n \times n$ permutation matrix $\mathbf{\Pi}$, then $||\mathbf{\Pi R}||_F=||\mathbf{R\Pi}||_F=||\mathbf{R}||_F$, i.e., multiplying a permutation matrix keeps invariant of the Frobenius norm.
\end{lemma}
	
\section{Models and Definitions}\label{sec:model}
In this section, we firstly introduce the social network models, then give the definition of the social network de-anonymization problem.

\subsection{Social Network Models}

The social network model considered in this paper is composed of three parts, i.e., the underlying network $G$, the published network $G_1$ and the auxiliary network $G_2$. $G_1$ and $G_2$ can be viewed as the incomplete observations of $G$. For instance, in reality $G$ may characterize the invisible relationship among a group of people, while $G_1$ might represent the online network in Facebook of this group of people and $G_2$ might represent the communication records in the cell phones of them, both of which are observable.

\subsubsection{Underlying Social Network}
Let $G=(V,E,\mathbf{U})$ be the underlying graph, where $V$ is the node set, $E$ is the edge set and $\mathbf{U}$\footnote{$\mathbf{U}(i,j)=1$ if $(i,j) \in E$ and $\mathbf{U}(i,j)=0$ if $(i,j) \notin E$} is the adjacency matrix of $G$. We regard $G$ as an undirected network and assume that the total number of nodes is $|V|=n$. One of the most common models to characterize the communities in networks is the Stochastic Block Model \cite{privacy1,privacy2,privacy3}. In our work, to reflect the property of overlapping communities, we suppose $G$ is generated based on the overlapping stochastic block model \cite{latouche2011overlapping}, the idea of which can be interpreted as follows:

Suppose there are $Q$ communities in $G$, where each community $q \in Q$ contains a subset of nodes. For a generic node $i$, we introduce a latent $Q$-dimensional column vector $\boldsymbol{C_i}$, in which all its $Q$ elements are independent boolean variables
$C_{iq}\in \{0,1\}$, with $C_{iq}$ being the $q_{th}$ row (element) in $\boldsymbol{C_i}$. $C_{iq}=1$ means that node $i$ is in community $q$ and $C_{iq}=0$ otherwise. Thus $\boldsymbol{C_i}$ can be seen as drawn from the Bernoulli distribution:

\begin{equation}
\label{e1}
\boldsymbol{C_i} \sim \prod_{q=1}^Q (p_q)^{C_{iq}}(1-p_q)^{1-C_{iq}},
\end{equation}
where $p_q$ is the probability of any node in $G$ falling into community $q$. Hence we have
\begin{equation}
\label{Bernoulli2}
Pr(\boldsymbol{C_i}=\{C_{i1},C_{i2},...,C_{iQ}\}^T)=\prod_{q=1}^Q (p_q)^{C_{iq}}(1-p_q)^{1-C_{iq}}.
\end{equation}
Intuitively, Eqn. (\ref{Bernoulli2}) shows the probability of node $i$ belonging to communities $q_1,q_2,...q_\ell$ which make the boolean variable $C_{iq_k}=1, k=1,2,...,\ell$ while not belonging to other communities. We call $\boldsymbol{C_i}$ as the \emph{community representation} of node $i$, since $\boldsymbol{C_i}$ explicitly represents to which communities node $i$ belongs and does not belong. For instance, if node $i$ belongs to communities $1$, $2$ and $3$, then the community representation of node $i$ is $\boldsymbol{C_i}=\{1,1,1,0,0,...,0\}^T$.

Unlike the stochastic block model in \cite{cite:blockmodel} which can only represent disjoint communities, the overlapping stochastic block model can measure the property of communities overlapping, which allows one node to belong to multiple communities. For ease of understanding, let us consider an example where node $i$ belongs to both communities $1$ and $2$. Then we have $Pr( \boldsymbol{C_i}=\{C_{i1},C_{i2},...,C_{iQ}\}^T)=p_1p_2\prod_{p=3}^Q(1-p_q)$.  For an edge $(i,j) \in E$, it is natural that the probability of the existence of this edge is determined by $\boldsymbol{C_i}$ and $\boldsymbol{C_j}$. Therefore we can set $Pr\{(i,j) \in E\}=Pr\{\mathbf{U}(i,j)=1\}=p_{\boldsymbol{C_i}\boldsymbol{C_j}}$, where $p_{\boldsymbol{C_i}\boldsymbol{C_j}}$ is a pre-defined parameter representing the probability of edge existence between two nodes belonging to any community representation. It has been demonstrated in \cite{latouche2011overlapping} that the overlapping stochastic block model turns out to be more reasonable in reality since overlapping property exists in social networks widely, and the parameters in this model can be estimated efficiently.

\subsubsection{Published Network and Auxiliary Network}
Now we proceed to define the published and auxiliary networks.
Specifically, we let $G_1(V_1,E_1,\mathbf{A})$ denote the published network, which can be interpreted as a graph that shares the same node labeling as the underlying graph, with its edges independently sampled from $G$ with some probability $s_1$. In contrast, an auxiliary network, denoted by $G_2(V_2,E_2,\mathbf{B})$, does not necessarily have the same node labeling as the underlying network and the edges are independently sampled from $G$ with some probability $s_2$. Again, here $\mathbf{A}$ and $\mathbf{B}$ respectively represent the adjacency matrix of published and auxiliary networks.

In correspondence to real situations, $G_1$ characterizes the publicly available anonymized network where users' identities are unavailable for privacy concern. On the contrary, $G_2$ characterizes an un-anonymized network where users' identities are all available. The adversary (attacker) can leverage the information of $G_2$, and tries to identify the users in $G_1$ based on the edge relationship between and community representation of both $G_1$ and $G_2$. In terms of edge relationship, the node of high degree in $G_1$ should be of higher possibility to correspond to a node which is also of high degree in $G_2$. Therefore while de-anonymizing any node in $G_2$, the adversary can harness this \emph{degree similarity} in matched node pairs to predict its corresponding node in $G_1$. In terms of community representation, the nodes in $G_1$ and $G_2$ with the same community representation should be matched with higher probability. Then the adversary can make use of this \emph{community representation similarity} while judging whether a node in $G_1$ is matched with the node in $G_2$ to be de-anonymized with high probability.

For the edge set $E_k$ $(k\in \{1,2\})$ of either network, we have
	\[Pr\{(i,j)\in E_{k}\}=\left\{\begin{array}{ll}
	s_k &\mbox{ if }(i,j)\in E,\\
	0& \mbox{ if }(i,j) \notin E.\\
	\end{array}\right.\]

For the node sets $V_1$ and $V_2$, we assume that the number of nodes in $G$, $G_1$ and $G_2$ are the same, i.e., $|V|=|V_1|=|V_2|=n$. By this assumption, there exists bijective mapping between $G_1$ and $G_2$, as will be defined in Section \ref{ModelB}. Note that it is easy to extend to the situation where $|V_1|\neq |V_2|$. Although the mapping between $G_1$ and $G_2$ in such case is no longer bijective, we only need to modify the permutation matrix (defined in Section \ref{ModelB}) between $G_1$ and $G_2$ from a square matrix into a non-square one, which will not influence our theoretical analysis.

Furthermore, we should clarify that in our model we render each node pair $(i,j)$ a weight $w_{ij}$, which, as will be defined in Section \ref{ModelB}, is dependent on the parameter set for the node pair $(i,j)$, i.e.,  $\mathbf{\theta}_{ij}=\{p_{\boldsymbol{C_i}\boldsymbol{C_j}},s_1,s_2\}$. Different pairs of nodes may have different weights. As we will state in Section \ref{ModelB}, $w_{ij}$ reflects the probability of edge existence between nodes $i$ and $j$, and the weights facilitates the reduction of the average de-anonymization error, which makes our estimation of permutation matrix more accurate.

\textbf{Remark:} According to the description above, it can be seen that $G$, $G_1$ and $G_2$ are all random variables. For the convenience of representation, we directly use $G$, $G_1$, $G_2$ as notations for the realizations of these random variables with no loss of clearance. Moreover, we set $\boldsymbol{\theta}=\{\{p_{\boldsymbol{C_iC_j}}\boldsymbol{|} 1\leq i,j \leq n\},s_1,s_2\}$ as the parameter set incorporating all pre-defined parameters in the model together.

\subsection{Social Network De-anonymization}\label{ModelB}
Predicated on the side information provided by the published network $G_1$ and the auxiliary network $G_2$, the goal of social network de-anonymization problem is to find a bijective node mapping $\pi: V_1 \mapsto V_2$, which is the true matching of nodes in $G_1$ and $G_2$. We can equivalently express this bijective mapping by forming a permutation matrix $\mathbf{\Pi} \in \{0,1\}^{n \times n}$, where $\mathbf{\Pi}(i,j)=1$ if $\pi(i)=j$ and $\mathbf{\Pi}(i,j)=0$ otherwise.
We denote $\mathbf{\Pi_0}$ as the true permutation matrix between $G_1$ and $G_2$, with $\pi_0$ representing the corresponding true bijective mapping. Note that we do not have any prior knowledge of $\mathbf{\Pi_0}$, and we do not have access to the underlying graph $G$ of $G_1$ and $G_2$.
Now we can formally define the social network de-anonymization problem as follows.

\begin{definition}(\textbf{Social Network De-anonymization Problem})\
Given the published network $G_1$, the auxiliary network $G_2$, parameter set $\boldsymbol{\theta}$, social network de-anonymization problem aims to construct the true bijective mapping $\pi_0$ between $V_1$ and $V_2$ (the true permutation matrix $\mathbf{\Pi_0}$ equivalently).
\end{definition}
\begin{figure}[htbp]
     	\centering
		\includegraphics[width=1\linewidth]{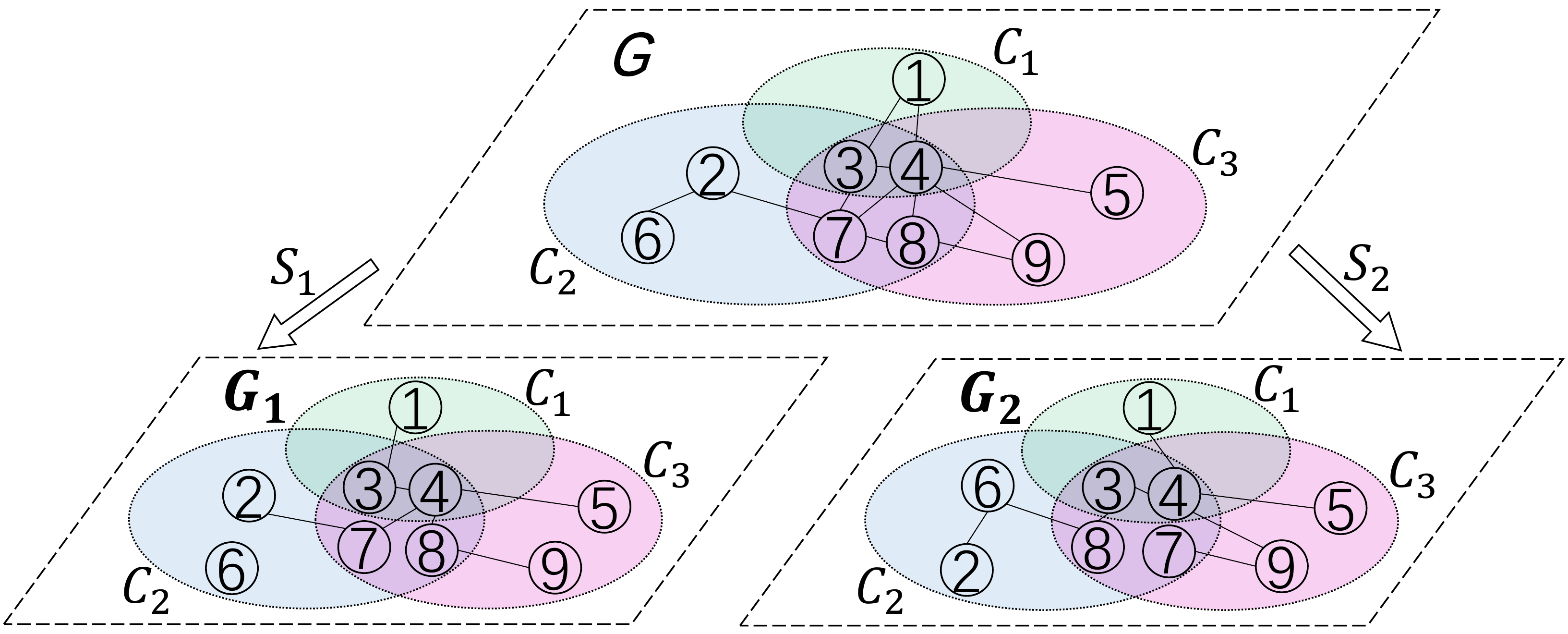}		
\caption{
\small
 An example of the underlying graph ($G$), published graph ($G_1$) and auxiliary graph ($G_2$). The edges of $G_{1(2)}$ are sampled independently from $G$ with probability $s_{1(2)}$. $C_1, C_2, C_3$ denote $3$ different communities in our overlapping stochastic block model. Nodes $7$ and $8$ belong to $2$ different communities. Nodes $3$ and $4$  belong to $3$ different communities. The true mapping $\pi_0=\{(1,1),(2,6),(3,3),(4,4),(5,5),(6,2),(7,8),(8,7),(9,9)\}$}
		\label{fig:1}
\end{figure}

Figure 1 illustrates an example of the proposed social de-anonymization problem that incorporates the feature of overlapping community. Here we note that our solution\footnote{Hereinafter our solution refers to the permutation matrix.} to the social network de-anonymization problem, denoted as $\hat{\mathbf{\Pi}}$, is not necessarily equal to the $\mathbf{\Pi_0}$. To quantify the difference between our solution and true permutation matrix, we introduce a metric called  ``node mapping error (NME)'', whose formal definition is provided as follows.

\begin{definition}{\textbf{(Node Mapping Error)}}
Given the estimated permutation matrix $\hat{\mathbf{\Pi}}$ and the true permutation matrix $\mathbf{\Pi_0}$, the node mapping error (NME) between $\hat{\mathbf{\Pi}}$ and $\mathbf{\Pi_0}$ is defined as
\begin{equation} d(\hat{\mathbf{\Pi}},\mathbf{\Pi_0})=\frac{1}{2}||\hat{\mathbf{\Pi}}-\mathbf{\Pi_0}||_F^2. \end{equation}
\end{definition}

Obviously $d(\hat{\mathbf{\Pi}},\mathbf{\Pi_0})$ equals to 0 if and only if two permutations are identical, and if there are $k$ nodes mapped mistakenly, then it equals to $k$. Therefore this metric reveals how much the estimated permutation of nodes deviates from the true one. Based on the definition of NME, the goal of the social network de-anonymization problem is thus to minimize NME.

As we have mentioned earlier, we have no prior knowledge of $\mathbf{\Pi_0}$, the true permutation matrix. 
Moreover, with the given $G_1$ and $G_2$, $\mathbf{\Pi_0}$ in fact can be viewed as a random variable whose probability distribution is conditioned on these two networks. Note that regarding $\mathbf{\Pi_0}$ as a random variable does not contradict the fact that there is only one determined true mapping between $G_1$ and $G_2$ in real situations, because this true mapping can be perceived as a realization of the random variable $\mathbf{\Pi_0}$. Therefore, we consider selecting $\mathbf{\hat{\Pi}}$, an estimation of the permutation matrix which minimizes the expected or mean value of the node mapping error (NME). We call this estimation as ``Minimum Mean Square Error (MMSE)'', since in the following Definition \ref{MMSE} we can discover that it is the minimizer of the node mapping error in the form of mean square. The formal definition of MMSE is as follows.

\begin{definition}(\textbf{The MMSE Estimator})
\label{MMSE}
Given the published network $G_1$, the auxiliary network $G_2$ and parameter set $\boldsymbol{\theta}$, the MMSE estimator is an estimation of permutation matrix which minimizes the number of mistakenly matched nodes in expectation, that is
\begin{equation}
\label{eh7}
\begin{aligned}
&\hat{\mathbf{\Pi}}=\arg \min_{\mathbf{\Pi} \in \mathbf{\Pi}^{n}} \mathbf{E}_{\mathbf{\Pi_0}}\{d(\mathbf{\Pi},\mathbf{\Pi_0})\}
\\&=\arg \min_{\mathbf{\Pi} \in \mathbf{\Pi}^{n}} \sum_{\mathbf{\Pi_0} \in \Pi^{n}}||\mathbf{\Pi}-\mathbf{\Pi_0}||_F^2Pr(\mathbf{\Pi_0}|G_1,G_2,\boldsymbol{\theta}),
\end{aligned}
\end{equation}
where $\textbf{E}_{\mathbf{\Pi_0}}$ means the expectation over all possible $\mathbf{\Pi_0}$.
The posterior probability $Pr(\mathbf{\Pi_0}|G_1,G_2,\boldsymbol{\theta})$ means the probability of a possible true permutation matrix $\mathbf{\Pi_0}$ given $G_1$, $G_2$ and $\boldsymbol{\theta}$.
\end{definition}


\textbf{Remark:} Recall that prior effort \cite{cite:arxiv-community} has leveraged Maximum A Posterior (MAP), which provides the solution with the highest probability being exactly identical to the true permutation.
MMSE and MAP characterize different aspects of minimizing NME. As far as we know, no previous work has learned de-anonymization under MMSE, which, however, is also of great significance as MAP in reducing NME.

The main notations used throughout the paper are summarized in Table 1.

				\begin{table}[!tb]
					\setlength{\extrarowheight}{3pt}
					
						\renewcommand\arraystretch{0.82}
						\caption{\bf Notions and Definitions}
						\centering
						\label{table:notation}
						\resizebox{1.0\columnwidth}{!}{
							
							\begin{tabular}{l|l}\hline\label{table:notation1}
								\textbf{Notation} & \textbf{Definition} \\\hline
								$G$ & Underlying social network \\ 
								$G_1,G_2$ & Published and auxiliary networks \\ 
								$V,V_1,V_2$ & Vertex sets of graphs $G$, $G_1$ and $G_2$ \\ 
								$E,E_{1},E_{2}$ & Edge sets of graphs $G$, $G_1$, $G_2$ \\  
								$s_1,s_2$ & Edge sampling probabilities of graphs $G_1$, $G_2$ \\  
								
                                $n$ & Total number of nodes \\
                                $Q$ & Total number of communities \\ 
								$q$ & One of the communities \\ 
                                $w_{ij}$ & The weight of node pair $(i,j)$ \\
                                								
                                $\boldsymbol{C_i}$ & Community representation of node $i$ \\ 
								$p_{\boldsymbol{C_i}\boldsymbol{C_j}}$ & Probability of edge existence between node $i$ and $j$ \\&with community representation $\boldsymbol{C_i}$ and $\boldsymbol{C_j}$ respectively\\ 
								$\boldsymbol{\theta}$ & Parameter set \\ 
                                $\mathbf{W}$ & The weight matrix \\
                                $\mathbf{U,A,B}$ & Adjacency matrices of $G$, $G_1$, $G_2$ \\  
                                $\mathbf{\Pi_0}(\pi_0)$ & True permutation matrix (True mapping) between $V_1$ and $V_2$ \\
                                $\mathbf{\Pi}(\pi)$ & A permutation matrix (A mapping) between $V_1$ and $V_2$ \\
                                $\mathbf{\hat{\Pi}}(\hat{\pi})$ & The MMSE estimator of de-anonymization problem\\& (the corresponding mapping)\\	
                                $\mathbf{\tilde{\Pi}}(\tilde{\pi})$ & The minimizer of weighted-edge matching problem\\& (the corresponding mapping)\\
                                $\mathbf{\Pi^n}$ & The set of $n \times n$ permutation matrices. \\
                                $g(\mathbf{\Pi})$ & The objective function of MMSE problem \\
                                \hline							
							\end{tabular}}
							
													
					\end{table}

\section{Analytical Aspect of De-anonymization Problem}\label{analytical}

In this section, we start to provide analysis of the social network de-anonymization problem that we have defined earlier. In doing so, we firstly prove that this problem is NP-hard. To facilitate the problem analysis, we then give an approximation to the original MMSE estimator and verify it under the expectation of different possible network structures. Furthermore, we validate this approximation by proving that the approximation ratio is not small for a single possible network structure.

\subsection{Transformation of MMSE Estimator}
As can be seen from the definition of MMSE (Eqn. (\ref{eh7}) in Section \ref{ModelB}), the posterior probability $Pr(\mathbf{\Pi_0}|G_1,G_2,\boldsymbol{\theta})$ still needs to be expressed more explicitly. Inspired by
the derivation in \cite{cite:arxiv-community}, we have the following theorem about the transformation of MMSE estimator.

\begin{theorem}
\label{th0}
Given the published graph $G_1$, the auxiliary graph $G_2$ and the parameter set $\boldsymbol{\theta}$, the MMSE estimator can be equivalently transformed into
\begin{equation}
\begin{aligned}
\label{eqn:th0}
\hat{\mathbf{\Pi}}&=\arg \max_{\mathbf{\Pi} \in \mathbf{\Pi^{n}}} \sum_{\mathbf{\Pi_0} \in \mathbf{\Pi^{n}}}||\mathbf{\Pi}-\mathbf{\Pi_0}||_F^2||\mathbf{W}\circ (\mathbf{\Pi_0} \mathbf{A}-\mathbf{B} \mathbf{\Pi_0})||_F^2
\\&=\arg \max_{\mathbf{\Pi} \in \mathbf{\Pi^{n}}} g(\mathbf{\Pi}),
\end{aligned}
\end{equation}
where $g(\mathbf{\Pi})=\sum_{\mathbf{\Pi_0} \in \mathbf{\Pi^{n}}}||\mathbf{\Pi}-\mathbf{\Pi_0}||_F^2||\mathbf{W}\circ (\mathbf{\Pi_0} \mathbf{A}-\mathbf{B} \mathbf{\Pi_0})||_F^2$ is the objective function of the MMSE problem, $\mathbf{W}$ is the weight matrix in which  $W(i,j)=\sqrt{w_{ij}}=W(j,i)$, $w_{ij}=\log \left(\frac{1-p_{\boldsymbol{C_i}\boldsymbol{C_j}}(s_1+s_2-s_1s_2)}{p_{\boldsymbol{C_i}\boldsymbol{C_j}}(1-s_1)(1-s_2)}\right)$ is weight between nodes $i$ and $j$,
and ``$\circ$'' denotes the Hadamard product.

\end{theorem}

\begin{proof}
Define $\mathcal{G_{\mathbf{\Pi}}}$ as the set of all realizations of the underlying network which is in consistency with the given $G_1$, $G_2$ and $\mathbf{\Pi}$. Then the MMSE estimator can be
written as
\begin{equation}
\label{def}
\hat{\mathbf{\Pi}}=\arg \min_{\mathbf{\Pi} \in \mathbf{\Pi^{n}}} \sum_{\mathbf{\Pi_0} \in \Pi^{n}}||\mathbf{\Pi}-\mathbf{\Pi_0}||_F^2\sum_{G \in \mathcal{G_{\mathbf{\Pi}}}} Pr(G,\mathbf{\Pi_0}|G_1,G_2,\boldsymbol{\theta}).\nonumber
\end{equation}

Let us focus on the conditional probability $Pr(G,\mathbf{\Pi_0}|G_1,G_2,\boldsymbol{\theta})$ in Eqn. (2). According to Bayesian's formula, along with the fact that $G_1$ and $G_2$ are sampled independently from each other, we obtain
\begin{equation}
\label{eq10}
\begin{aligned}
Pr(G,\mathbf{\Pi_0}|G_1,G_2,\boldsymbol{\theta})&=\frac{Pr(G,G_1,G_2,\mathbf{\Pi_0})}{Pr(G_1,G_2)} \\&\sim
Pr(G) Pr(G_1|G) Pr(G_2|G,\mathbf{\Pi_0}),
\end{aligned}
\end{equation}
where $a \sim b$ means that $a$ and $b$ are different only in parameters unrelated to $\mathbf{\Pi_0}$, which will not change the value of $\arg \max$ or $\arg \min$.\footnote{There is a notation abuse for $\sim$ between the one in Eqn. (\ref{e1}) and here.}
Note that the parameter set $\boldsymbol{\theta}$ remains invariant, so we need not add $\boldsymbol{C_i}$ and $\boldsymbol{\theta}$ into further consideration.

Set $E^{ij}$ as the indicator variable about whether an edge exists between nodes $i$ and $j$ in the edge set $E$. If an edge exists then $E^{ij}=1$, otherwise $E^{ij}=0$. The same rule also holds for indicators $E_1^{ij}$ and $E_2^{ij}$. Therefore Eqn. (\ref{eq10}) can be further written as
\begin{equation}
\label{e7}
\begin{aligned}
&\sum_{G \in \mathcal{G_{\mathbf{\Pi}}}}Pr(G) Pr(G_1|G) Pr(G_2|G,\Pi_0)
\\&=\sum_{G \in \mathcal{G_{\mathbf{\Pi}}}} \prod_{i<j}^n s_1^{E_1^{ij}} (1-s_1)^{E^{ij}-E_1^{ij}} s_2^{E_2^{\pi_0(i)\pi_0(j)}}
\\
&\quad \cdot (1-s_2)^{E^{ij}-E_2^{\pi_0(i)\pi_0(j)}} p_{\boldsymbol{C_i}\boldsymbol{C_j}}^{E^{ij}}(1-p_{\boldsymbol{C_i}\boldsymbol{C_j}})^{1-E^{ij}} \\
&= \prod_{i<j}\left(\frac{s_1}{1-s_1}\right)^{E_1^{ij}}\left(\frac{s_2}{1-s_2}\right)^{E_2^{\pi_0(i)\pi_0(j)}}
\\&\quad \cdot \sum_{G\in \mathcal{G_{\mathbf{\Pi}}}} \left((1-s_1)(1-s_2)
\frac{p_{\boldsymbol{C_i}\boldsymbol{C_j}}}{1-p_{\boldsymbol{C_i}\boldsymbol{C_j}}}\right)^{E^{ij}} \\
&\sim \sum_{G\in \mathcal{G_{\mathbf{\Pi}}}} \left((1-s_1)(1-s_2)\frac{p_{\boldsymbol{C_i}\boldsymbol{C_j}}}{1-p_{\boldsymbol{C_i}\boldsymbol{C_j}}}\right)^{E^{ij}}.
\end{aligned}
\end{equation}

Note that the last equivalence in Eqn. (\ref{e7}) holds since the term $\left(\frac{s_1}{1-s_1}\right)^{E_1^{ij}}$ does not depend on $\pi_0$ and the product $\prod_{i<j}\left(\frac{s_2}{1-s_2}\right)^{E_2^{\pi_0(i)\pi_0(j)}}$ is independent of $\pi_0$ due to the bijective property of $\pi_0$.

Then we define $G_{\pi_0}^*$ as the graph which has the smallest number of edges in $\mathcal{G_{\mathbf{\Pi}}}$. Equivalently $G_{\pi_0}^*=(V,E_1\cup \pi_0(E_1))$, where $\pi_0(E_1)=\{(\pi_0(i),\pi_0(j))\boldsymbol{|}(i,j)\in E_1\}$. Now we set $E_{\pi_0}^*$ as the edge set of $G_{\pi_0}^*$, and $E_{\pi_0}^{*ij}$ as the indicator variable between nodes $i$ and $j$, i.e., $E_{\pi_0}^{*ij}=1$ if $(i,j) \in E_{\pi_0}^*$ and $E_{\pi_0}^{*ij}=0$ otherwise. Then we sum up all the graphs in $\mathcal{G_{\mathbf{\Pi}}}$

\begin{equation}
\label{e8}
\begin{aligned}
&\sum_{G\in \mathcal{G_{\mathbf{\Pi}}}} \left((1-s_1)(1-s_2)\frac{p_{\boldsymbol{C_i}\boldsymbol{C_j}}}{1-p_{\boldsymbol{C_i}\boldsymbol{C_j}}}\right)^{E^{ij}}
\\&= \prod_{i<j}^n \left((1-s_1)(1-s_2)\frac{p_{\boldsymbol{C_i}\boldsymbol{C_j}}}{1-p_{\boldsymbol{C_i}\boldsymbol{C_j}}}\right)^{E_{\pi_0}^{*ij}} \\ &\quad \cdot \sum_{k=0}^{E_{ij}-E_{\pi_0}^{*ij}} C_{E_{ij}-E_{\pi_0}^{*ij}}^{k} \left((1-s_1)(1-s_2)\frac{p_{\boldsymbol{C_i}\boldsymbol{C_j}}}{1-p_{\boldsymbol{C_i}\boldsymbol{C_j}}}\right)^k. \end{aligned}
\end{equation}

Note that in Eqn. (\ref{e8}) last multiplicative factor ,
\begin{equation}
\sum_{k=0}^{E_{ij}-E_{\pi_0}^{*ij}} C_{E_{ij}-E_{\pi_0}^{*ij}}^{k} \left((1-s_1)(1-s_2)\frac{p_{\boldsymbol{C_i}\boldsymbol{C_j}}}{1-p_{\boldsymbol{C_i}\boldsymbol{C_j}}}\right)^k, \nonumber
\end{equation}
yields as a Bernoulli sum, therefore Eqn. (\ref{e8}) can be further written as

\begin{equation}
\label{e9}
\begin{aligned}
&\sum_{G\in \mathcal{G_{\mathbf{\Pi}}}} \left((1-s_1)(1-s_2)\frac{p_{\boldsymbol{C_i}\boldsymbol{C_j}}}{1-p_{\boldsymbol{C_i}\boldsymbol{C_j}}}\right)^{E^{ij}}
\\&=\prod_{i<j}^n \left((1-s_1)(1-s_2)\frac{p_{\boldsymbol{C_i}\boldsymbol{C_j}}}{1-p_{\boldsymbol{C_i}\boldsymbol{C_j}}}\right)^{E_{\pi_0}^{*ij}} \\ &\quad \cdot\left(1+(1-s_1)(1-s_2)\frac{p_{\boldsymbol{C_i}\boldsymbol{C_j}}}{1-p_{\boldsymbol{C_i}\boldsymbol{C_j}}}\right)^{1-E_{\pi_0}^{*ij}} \\
&\sim \prod_{i<j}^n \left(\frac{p_{\boldsymbol{C_i}\boldsymbol{C_j}}(1-s_1)(1-s_2)}{1-p_{\boldsymbol{C_i}\boldsymbol{C_j}}(s_1+s_2-s_1s_2)}\right)^{E_{\pi_0}^{*ij}}\\
&\sim \sum_{i<j}^n E_{\pi_0}^{*ij} \log\left(\frac{p_{\boldsymbol{C_i}\boldsymbol{C_j}}(1-s_1)(1-s_2)}{1-p_{\boldsymbol{C_i}\boldsymbol{C_j}}(s_1+s_2-s_1s_2)}\right).
\end{aligned}
\end{equation}

Here the last line in Eqn. (\ref{e9}) holds since the log operator keeps the minimum $\mathbf{\Pi_0}$ invariant.
Note that $G_{\pi_0}^*=(V,E_1\cup \pi_0(E_1))$. Then we can find that $E_{\mathbf{\Pi_0}}^{*ij}=0$ if and only if both $E_1^{ij}$ and $E_2^{ij}$ are equal to $0$, and $E_{\mathbf{\Pi_0}}^{*ij}=1$ occurs in the following three conditions:

\begin{itemize}
\item $(i,j) \in E_1$ but $(i,j) \notin E_2$. Note that this condition also ensures that $(\pi_0(i),\pi_0(j)) \in E_2$.
\item $(i,j) \in E_2$ but $(i,j) \notin E_1$. Note that this condition also ensures that
    $(\pi_0(i),\pi_0(j)) \notin E_2$.
\item $(i,j) \in E_1$ and $(i,j) \in E_2$. Note that this condition also ensures that
    $(\pi_0(i),\pi_0(j)) \in E_2$.
\end{itemize}

Synthesizing all the above conditions, we can express $E_{\pi_0}^{*ij}$ as
\begin{equation}
\label{e13}
E_{\pi_0}^{*ij}=\frac{1}{2}(E_1^{ij}+E_2^{ij}+|\mathds{1}\{(i,j)\in E_1\}-\mathds{1}\{(\pi_0(i),\pi_0(j))\in E_2\}|),
\end{equation}
where $\mathds{1}\{P\}=1$ if the random event $P$ happens and  $\mathds{1}\{P\}=0$ otherwise.
Substituting Eqn. (\ref{e13}) into the last line in Eqn. (\ref{e9}), we get
\begin{equation}
\label{e15}
\begin{aligned}
&\arg \min_{\mathbf{\Pi} \in \mathbf{\Pi^{n}}} \sum_{i<j}^n E_{\pi_0}^{*ij} \log\left(\frac{p_{\boldsymbol{C_i}\boldsymbol{C_j}(1-s_1)(1-s_2)}}{1-p_{\boldsymbol{C_i}\boldsymbol{C_j}}(s_1+s_2-s_1s_2)}\right) \\
&=\arg \max_{\mathbf{\Pi} \in \mathbf{\Pi^{n}}} \sum_{i<j}^n w_{ij} |\mathds{1}\{(i,j)\in E_1\}-\mathds{1}\{(\pi_0(i),\pi_0(j))\in E_2\}| \\
&=\arg \max_{\mathbf{\Pi} \in \mathbf{\Pi^{n}}} ||\mathbf{W}\circ (\mathbf{\Pi_0} \mathbf{A}-\mathbf{B} \mathbf{\Pi_0}) ||_{F}^2,
\end{aligned}
\end{equation}
where $w_{ij}=\log \left(\frac{1-p_{\boldsymbol{C_i}\boldsymbol{C_j}}(s_1+s_2-s_1s_2)}{p_{\boldsymbol{C_i}\boldsymbol{C_j}}(1-s_1)(1-s_2)}\right)$ is weight between nodes $i$ and $j$, $\mathbf{W}$ is the symmetric weight matrix where $\mathbf{W}(i,j)=\sqrt{w_{ij}}=\mathbf{W}(j,i)$, and ``$\circ$'' denotes the Hadamard product.

Substituting Eqn. (\ref{e15}) into Eqn. (\ref{def}), now we can formulate the MMSE estimator as
\begin{equation}
\label{eqn:19}
\hat{\mathbf{\Pi}}=\arg \max_{\mathbf{\Pi} \in \mathbf{\Pi^{n}}} \sum_{\mathbf{\Pi_0} \in \mathbf{\Pi^{n}}}||\mathbf{\Pi}-\mathbf{\Pi_0}||_F^2||\mathbf{W}\circ (\mathbf{\Pi_0} \mathbf{A}-\mathbf{B} \mathbf{\Pi_0})||_F^2.
\end{equation}
\end{proof}

\textbf{Remark:} Additionally, to simplify the form of $||\mathbf{W}\circ (\mathbf{\Pi_0} \mathbf{A}-\mathbf{B} \mathbf{\Pi_0}) ||_F^2$, we use $\mathbf{{\Pi_0\hat{A}}}$ to represent $
\mathbf{W}\circ \mathbf{\Pi_0} \mathbf{A}$, and $\mathbf{\hat{B}\Pi_0}$ to represent $\mathbf{W}\circ \mathbf{B} \mathbf{\Pi_0}$\footnote{We should clarify that we only provide a simpler form to represent $
\mathbf{W}\circ \mathbf{\Pi_0} \mathbf{A}$ and $\mathbf{W}\circ \mathbf{B} \mathbf{\Pi_0}$, and it does NOT imply that $\mathbf{W}\circ \mathbf{A}=\mathbf{\hat{A}}$ and $\mathbf{W}\circ \mathbf{B}=\mathbf{\hat{B}}$. But some operations under this new notation still hold, for example, multiplying a permutation matrix does not change the value of the Frobenius norm, i.e., $||\mathbf{\Pi_0} \mathbf{\hat{A}}-\mathbf{\hat{B}} \mathbf{\Pi_0}||_F^2=||\mathbf{W}\circ (\mathbf{\Pi_0} \mathbf{A}-\mathbf{B} \mathbf{\Pi_0})||_F^2=||\mathbf{W}\circ \mathbf{\Pi_0^T}(\mathbf{\Pi_0} \mathbf{A}-\mathbf{B} \mathbf{\Pi_0}))||_F^2=||\mathbf{W}\circ (\mathbf{A}-\mathbf{\Pi_0^T}\mathbf{B} \mathbf{\Pi_0})||_F^2 $ and $||\mathbf{\Pi_0} \mathbf{\hat{A}}-\mathbf{\hat{B}} \mathbf{\Pi_0}||_F^2=||\mathbf{\Pi_0}\mathbf{\hat{A}}\mathbf{\Pi_0^T}-\mathbf{\hat{B}} ||_F^2$.}.
Therefore we can rewrite the MMSE estimator in Eqn. (\ref{eqn:19}) as
\begin{equation}
\label{eqn:20}
\begin{aligned}
\hat{\mathbf{\Pi}}&=\arg \max_{\mathbf{\Pi} \in \mathbf{\Pi^{n}}} \sum_{\mathbf{\Pi_0} \in \mathbf{\Pi^{n}}}||\mathbf{\Pi}-\mathbf{\Pi_0}||_F^2||\mathbf{\Pi_0} \mathbf{\hat{A}}-\mathbf{\hat{B}} \mathbf{\Pi_0}||_F^2,
\end{aligned}
\end{equation}
and $g(\mathbf{\Pi})=\sum_{\mathbf{\Pi_0} \in \mathbf{\Pi^{n}}}||\mathbf{\Pi}-\mathbf{\Pi_0}||_F^2||\mathbf{\Pi_0} \mathbf{\hat{A}}-\mathbf{\hat{B}} \mathbf{\Pi_0}||_F^2$. In the following analysis, we use the form in Eqn. (\ref{eqn:20}). In Section \ref{algorithmicA}, we will discuss the condition under which $\mathbf{W}\circ\mathbf{A}=\mathbf{\hat{A}}$ and $\mathbf{W}\circ\mathbf{B}=\mathbf{\hat{B}}$.

\subsection{NP-hardness of Solving the MMSE Estimator}
Since we have derived a more explicit form of MMSE estimator, we are interested in whether there exists a polynomial-time algorithm that can solve the MMSE problem. However, as we will prove in the sequel, this problem is NP-hard, meaning that no polynomial time (pseudo-polynomial time) approximation algorithm exists for solving the MMSE estimator.

\begin{Proposition}
Solving the MMSE estimator is an NP-hard problem. There is no polynomial time or pseudo-polynomial time approximation algorithm for this problem with any multiplicative approximation guarantee unless P=NP.
\end{Proposition}

\begin{proof}
We derive the proof in two steps: $1$. modeling this problem as a clique with weighted nodes and edges, and $2$. reducing the $1$-median problem to our MMSE problem\footnote{Note that $1$-median itself is not NP-hard, but we demonstrate that the lower bound of $1$-median is of $O(n)$ and when applied in our problem it becomes larger than polynomial.}. Here the 1-median problem \cite{Kariv1979ALGORITHM} refers to that: Given a connected undirected graph $G=(V,E)$ in which no isolated vertices exist and each node $v$ is endowed with a nonnegative weight $\omega(v)$, find the vertex $v^*$ which minimizes weighted sum.
\begin{equation}
H(v^*)=\sum_{v\in V}\omega(v) \cdot D(v,v^*),\nonumber
\end{equation}
where $D(v,v^*)$ means the shortest path length (also nonnegative) between nodes $v$ and $v^*$.

\textbf{Reduction from $1$-median problem:} Our construction of the clique works as follows: Suppose there are $n$ nodes in $G_1$ and $G_2$. Then for any permutation matrix $\mathbf{\Pi} \in \mathbf{\Pi^n}$, we have
\begin{equation}
\begin{aligned}
&\hat{\mathbf{\Pi}}=\arg \max_{\mathbf{\Pi} \in \Pi^{n}} \sum_{\mathbf{\Pi_0} \in \mathbf{\Pi^{n}}}||\mathbf{\Pi}-\mathbf{\Pi_0}||_F^2|| \mathbf{\Pi_0} \mathbf{\hat{A}}-\mathbf{\hat{B}} \mathbf{\Pi_0}||_F^2
\\&=\arg \min_{\mathbf{\Pi} \in \Pi^{n}} \sum_{\mathbf{\Pi_0} \in \mathbf{\Pi^{n}}}(4n-||\mathbf{\Pi}-\mathbf{\Pi_0}||_F^2)||\mathbf{\Pi_0} \mathbf{\hat{A}}-\mathbf{\hat{B}} \mathbf{\Pi_0}||_F^2,\nonumber
\end{aligned}
\end{equation}
in which all the multiplicative factors are nonnegative. Since the number of elements in $\mathbf{\Pi^n}$ is $n!$, then we construct a clique with $n!$ nodes, with every node representing an $n \times n$ permutation matrix. We set the distance between two nodes $i$ and $j$ as $D(i,j)=4n-||\mathbf{\Pi}(i)-\mathbf{\Pi}(j)||_F^2$. Note that this distance satisfies the triangular equality $D(i, k)+D(k, j) \geq D(i, j)$, which assures that the edge directly connecting nodes $i$ and $j$ has the minimum distance among all possible paths between them. So the shortest path length between nodes $i$ and $j$ is just the distance $D(i,j)$. We define the weight of node $i$ as $\omega(i)=||\mathbf{\Pi_0\hat{A}-\hat{B}\Pi_0}||_F^2$ (Note that each $\mathbf{\Pi_0}$ is a node in the graph with $n!$ nodes). For ease of understanding, Fig. \ref{fig:3} illustrates the constructed clique with $5$ nodes.

	\begin{figure}[htbp]
		\centering	\includegraphics[width=1\linewidth]{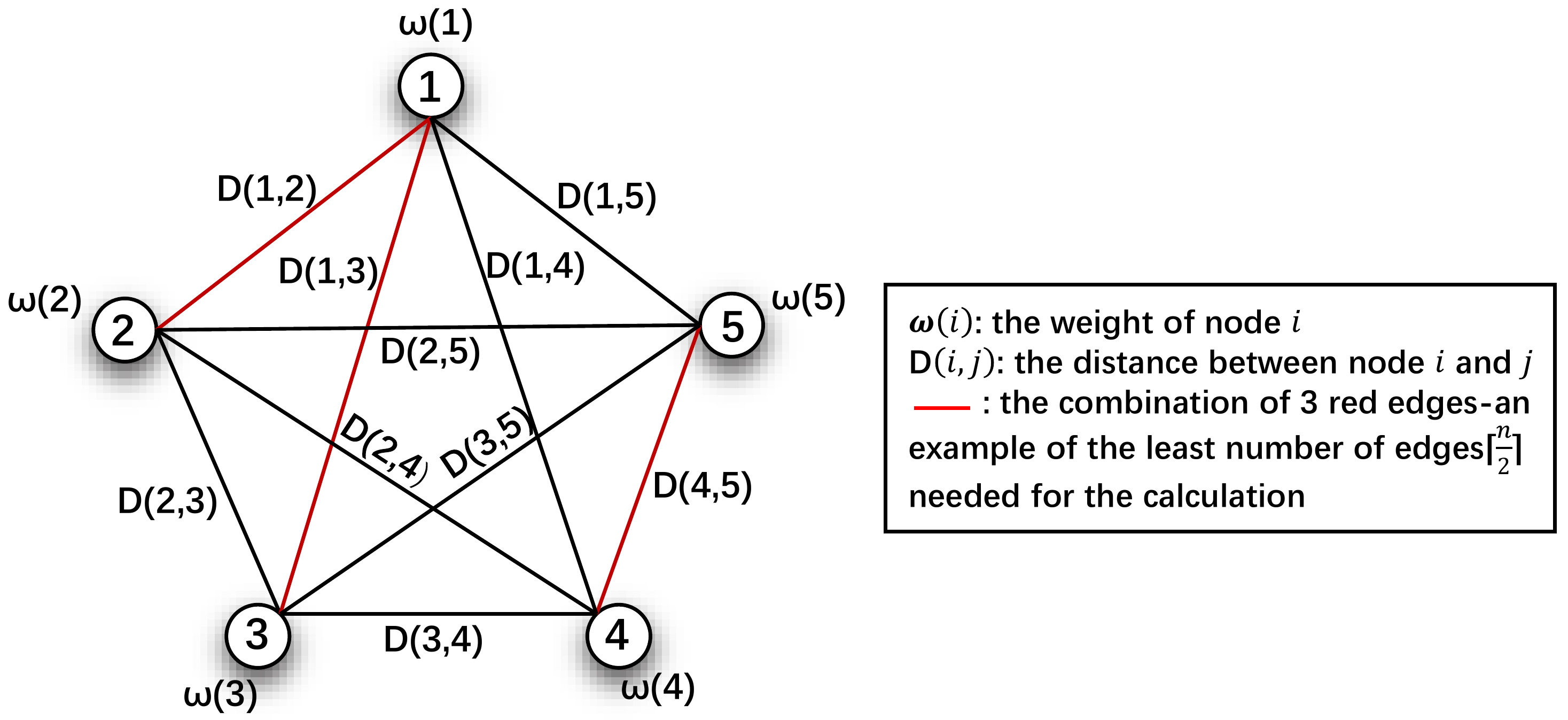}
		\caption{An Illustration of the Constructed Clique with $5$ Nodes}
		\label{fig:3}
	\end{figure}

\textbf{The Lower Bound for $1$-Median Problem:} Based on the above construction, we equivalently transform our problem into the form of a $1$-median problem:
\begin{equation}
i_0=\arg\min_{i^*\in V}\sum_{i\in V} \omega(i) \cdot D(i,i^*).\nonumber
\end{equation}

For a $1$-median problem with $n$ nodes, it is easy to discover that we need to calculate at least $\lceil n/2 \rceil$ times, since we need at least $\lceil n/2 \rceil$  edges to form an edge set such that the endpoints of all edges in this edge set cover all the vertices in the graph. Or else one node will not be calculated for any edge connecting it, thus no information about this node is revealed, and then we can not judge whether this node is the one we intend to find. The red lines in Fig. \ref{fig:3} illustrates an example that when there are $5$ nodes, the least number of edges needed to be calculated is $\lceil 5/2 \rceil=3$. For our MMSE estimator problem we have $n!$  nodes, thus the calculation times is at least $(n/2)!$, which means that we need to calculate $(n/2)!$ permutation matrices. Compared with the size of the problem, $n^2$, the complexity turns out to be $\Omega(((\sqrt{n})/2)!)=\Omega(\sqrt{n}!)$, which exceeds polynomial.

\end{proof}

The NP-hardness of MMSE estimator shows the impossibility to pursue an exact algorithm or any approximation algorithm with multiplicative guarantee. Thus we need to simplify this problem by conducting reasonable approximation to make it possible to solve this problem, with certain tolerance of mapping error. In the following we propose one way to approximate this problem, the analysis of which will indicate that the error arose by this approximation can be bounded.

\subsection{Approximation of the MMSE estimator}\label{estimateMMSE}
As we have just stated above, the NP-hardness of MMSE problem urges us to find proper approximation for the original problem. Recall that MMSE involves all the possible true mappings, the number of which is $n!$, thus leading to fairly prohibitive computational cost. To tackle the difficulty, we firstly transform the original MMSE problem into a weighted-edge matching problem (WEMP), which, as we will define and present more details later, simplifies the form of objective function of the original MMSE problem and makes it tractable. Then we demonstrate that this transformation is valid, meaning that the solution of WEMP will not deviate much from the solution of the original MMSE problem by proving its high approximation ratio. Definition \ref{def8} provides the formal statement of WEMP.

\begin{definition}{\textbf{(Weighted-Edge Matching Problem)}}
\label{def8}
Given the adjacent matrices of $G_1$ and $G_2$, denoted as $\mathbf{A}$ and $\mathbf{B}$ respectively, and the weight matrix $\mathbf{W}$, the weight-edge matching problem is to find
\begin{equation}
\begin{aligned}
\mathbf{\tilde{\Pi}}&=\arg \min_{\mathbf{\Pi}\in \mathbf{ \Pi^n}} ||\mathbf{\Pi \hat{A} - \hat{B} \Pi}||_F^2.\nonumber
\end{aligned}
\end{equation}
\end{definition}

Hereinafter we discuss the following two aspects of WEMP:
\begin{itemize}
\item How do we transform from the original MMSE problem into WEMP?
\item How is the validity of this transformation?
\end{itemize}

\subsubsection{\textbf{The Idea of Transformation}}
\label{idea}
We intend to transform the original problem of solving the MMSE estimator into WEMP. The idea of this transformation can be interpreted in the following sense: for any fixed $\mathbf{\Pi}$, define a set $S_k(\mathbf{\Pi}),0\leq k \leq n$, any element of which is an $n \times n$ permutation matrix $\mathbf{\Pi_0}$ such that $d(\mathbf{\Pi},\mathbf{\Pi_0})=2k$. It is obvious that $S_0(\mathbf{\Pi})=\{\mathbf{\Pi}\}$ and $S_1(\mathbf{\Pi})=\emptyset$. Then we can reform the original problem as
\begin{equation}
\label{e16}
\hat{\mathbf{\Pi}}=\arg \max_{\mathbf{\Pi} \in \mathbf{\Pi^{n}}} \sum_{k=0}^n k \left(\sum_{\mathbf{\Pi_0} \in S_k(\mathbf{\Pi})} ||\mathbf{\Pi_0 \hat{A} - \hat{B} \Pi_0}||_F^2\right).
\end{equation}

Zooming in on Eqn. (\ref{e16}), we propose our idea of transforming it into WEMP. To present our idea clearly, we divide our analysis into three parts; First we analyze a single term, $||\mathbf{\Pi_0 \hat{A} - \hat{B} \Pi_0}||_F^2$, where $\mathbf{\Pi_0}\in S_2(\mathbf{\Pi})$; Then we analyze $\mathbf{\Pi_0}\in S_k(\mathbf{\Pi})$ based on the analysis of $\mathbf{\Pi_0}\in S_2(\mathbf{\Pi})$; Finally we analyze the R.H.S of Eqn. (\ref{e16}) based on Lemma \ref{lemma1}. In the sequel we unfold the three parts in a more detailed manner:

\underline{1. Analysis of $||\mathbf{\Pi_0 \hat{A} - \hat{B} \Pi_0}||_F^2$ where $\mathbf{\Pi_0}\in S_2(\mathbf{\Pi})$}

Now we focus on the value of $||\mathbf{\Pi_0}\mathbf{\hat{A}}-\mathbf{\hat{B}}\mathbf{\Pi_0}||_F^2$, where $\mathbf{\Pi_0}\in S_2(\mathbf{\Pi})$. Note that any permutation in $S_2(\mathbf{\Pi})$ only causes matching error on one pair of nodes. Thus if we consider $\mathbf{\Pi}=\mathbf{\tilde{\Pi}}$ and set one specific $\mathbf{\Pi_0}\in S_2(\mathbf{\tilde{\Pi}})$, which differs from $\mathbf{\tilde{\Pi}}$ only in the $i_{th}$ and $j_{th}$ row, we can derive that
\begin{equation}
\label{eqn:28}
\begin{aligned}
&||\mathbf{\Pi_0}\mathbf{\hat{A}}-\mathbf{\hat{B}}\mathbf{\Pi_0}||_F^2-||\mathbf{\tilde{\Pi}}\mathbf{\hat{A}}-\mathbf{\hat{B}}\mathbf{\tilde{\Pi}}||_F^2 \\&=||\mathbf{W}\circ(\mathbf{\Pi_0}\mathbf{{A}}\mathbf{\Pi_0^T}-\mathbf{{B}})||_F^2-||\mathbf{W}\circ(\mathbf{\tilde{\Pi}}\mathbf{{A}}\mathbf{\tilde{\Pi}^T}-\mathbf{{B}})||_F^2
\\&=2\bigg(\sum_{k\neq i,j}^n [(\mathbf{W}\circ(\mathbf{\Pi_0{A}\Pi_0^T-B}))_{ik}^2-(\mathbf{W}\circ(\mathbf{\tilde{\Pi}{A}\tilde{\Pi}^T-B}))_{ik}^2]
\\&\quad+\sum_{k\neq i,j}^n [(\mathbf{W}\circ(\mathbf{\Pi_0{A}\Pi_0^T-B}))_{jk}^2-(\mathbf{W}\circ(\mathbf{\tilde{\Pi}{A}\tilde{\Pi}^T-B}))_{jk}^2]\bigg)
\\&=2\bigg(\sum_{k\neq i,j}^n w_{ik}[(\mathbf{\Pi_0A\Pi_0^T-B})_{ik}^2-(\mathbf{\tilde{\Pi}A\tilde{\Pi}^T-B})_{ik}^2]
\\&\quad+\sum_{k\neq i,j}^n w_{jk}[(\mathbf{\Pi_0A\Pi_0^T-B})_{jk}^2-(\mathbf{\tilde{\Pi}A\tilde{\Pi}^T-B})_{jk}^2]\bigg)
\\&=2\bigg(\sum_{k\neq i,j}^n w_{ik}[\mathbf{\Pi_0A\Pi_0^T}-\mathbf{\tilde{\Pi}A\tilde{\Pi}^T}]_{ik}\psi(\mathbf{B}_{ik})
\\&\quad+\sum_{k\neq i,j}^n w_{jk}[\mathbf{\Pi_0A\Pi_0^T}-\mathbf{\tilde{\Pi}A\tilde{\Pi}^T}]_{jk}\psi(\mathbf{B}_{jk})\bigg),
\end{aligned}
\end{equation}
where $\psi(x)=-1$ if $x=1$ and $\psi(x)=1$ if $x=0$. Fig. \ref{Fig:example} illustrates how Eqn. (\ref{eqn:28}) can be derived intuitively. Note that if $\mathbf{\Pi_0}$ and $\mathbf{\tilde{\Pi}}$ are different only in the $i_{th}$ and $j_{th}$ rows, then the difference between $||\mathbf{W}\circ(\mathbf{\Pi_0}\mathbf{{A}}\mathbf{\Pi_0^T}-\mathbf{{B}})||_F^2$ and $||\mathbf{W}\circ(\mathbf{\tilde{\Pi}}\mathbf{{A}}\mathbf{\tilde{\Pi}^T}-\mathbf{{B}})||_F^2 $ exists in the red circles in Fig. \ref{Fig:example}, which corresponds to the third line in Eqn. (\ref{eqn:28}). Note that the intersection part, i.e., the stars in Fig. \ref{Fig:example}, does not contribute to the $||\mathbf{W}\circ(\mathbf{\Pi_0}\mathbf{{A}}\mathbf{\Pi_0^T}-\mathbf{{B}})||_F^2$ and $||\mathbf{W}\circ(\mathbf{\tilde{\Pi}}\mathbf{{A}}\mathbf{\tilde{\Pi}^T}-\mathbf{{B}})||_F^2$.

\begin{figure}[htbp]

     	\centering
\includegraphics[width=0.48\textwidth]{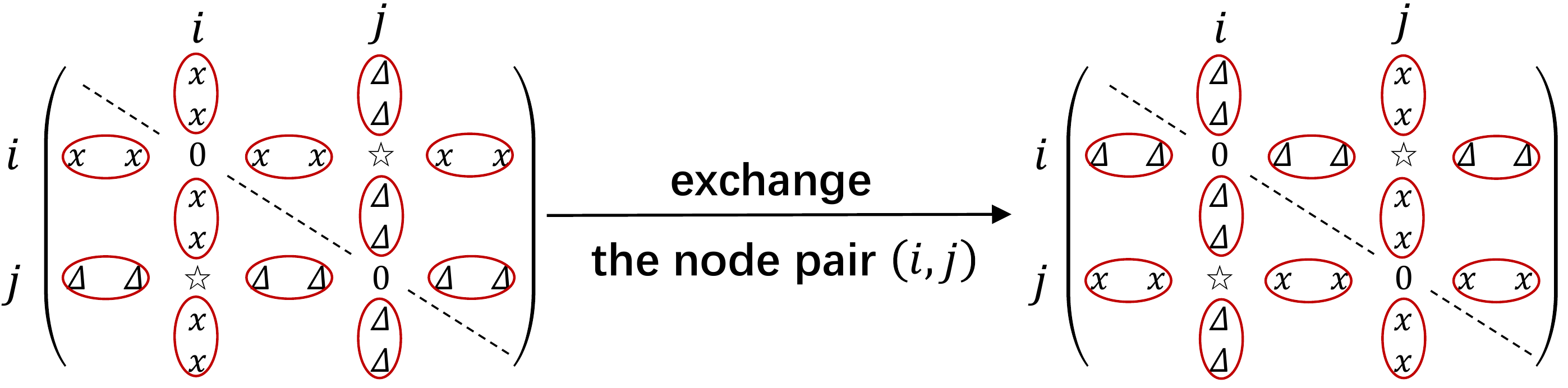}
		\caption{\small An example of the effect of $\mathbf{\Pi_0}$ which differs from $\mathbf{\tilde{\Pi_0}}$ only in the $i_{th}$ and $j_{th}$ row. The triangles denote the $j_{th}$ row and column the ``x''es denote the $i_{th}$ row and column of $\mathbf{W}\circ(\mathbf{\Pi_0}\mathbf{{A}}\mathbf{\Pi_0^T}-\mathbf{{B}})$.
And the triangles denote the $i_{th}$ row and column the ``x''es denote the $j_{th}$ row and column of $\mathbf{W}\circ(\mathbf{\tilde{\Pi}}\mathbf{{A}}\mathbf{\tilde{\Pi}^T}-\mathbf{{B}})$.
Note that the difference between $\mathbf{W}\circ(\mathbf{\Pi_0}\mathbf{{A}}\mathbf{\Pi_0^T}-\mathbf{{B}})$ and $\mathbf{W}\circ(\mathbf{\tilde{\Pi}}\mathbf{{A}}\mathbf{\tilde{\Pi}^T}-\mathbf{{B}})$
exists in the $i_{th}$ and $j_{th}$ row and column  except the intersections (those 0s and stars).
}
		\label{Fig:example}
\end{figure}

Note that since $\mathbf{\Pi_0}$ and $\mathbf{\tilde{\Pi}}$ are different in the $i_{th}$ and $j_{th}$ rows, then $(\mathbf{\tilde{\Pi}A\tilde{\Pi}^T})_{ik}=(\mathbf{{\Pi_0}A{\Pi_0^T}})_{jk}$. Therefore
\begin{equation}
\label{equation33}
\begin{aligned}
&||\mathbf{\Pi_0}\mathbf{\hat{A}}-\mathbf{\hat{B}}\mathbf{\Pi_0}||_F^2-||\mathbf{\tilde{\Pi}}\mathbf{\hat{A}}-\mathbf{\hat{B}}\mathbf{\tilde{\Pi}}||_F^2
\\&=2\bigg(\sum_{k\neq i,j}^n w_{ik}\psi(\mathbf{B}_{ik})([\mathbf{\Pi_0A\Pi_0^T}]_{ik}-[\mathbf{\Pi_0A\Pi_0^T}]_{jk})
\\&\quad+
\sum_{k\neq i,j}^n w_{jk}\psi(\mathbf{B}_{jk})([\mathbf{\Pi_0A\Pi_0^T}]_{jk}-[\mathbf{{\Pi_0}A{\Pi_0^T}}]_{ik})\bigg)
\\&=2\bigg(\sum_{k\neq i,j}^{n}(w_{ik}\psi(\mathbf{B}_{ik})-w_{jk}\psi(\mathbf{B}_{jk}))
\\& \quad \cdot [(\mathbf{\Pi_0} \mathbf{A}\mathbf{\Pi_0^T})_{ki}-(\mathbf{\Pi_0} \mathbf{A}\mathbf{\Pi_0^T})_{kj}]\bigg),
\end{aligned}
\end{equation}

\begin{figure}[htbp]
     	\centering
		\includegraphics[width=1\linewidth,height=0.3\linewidth]{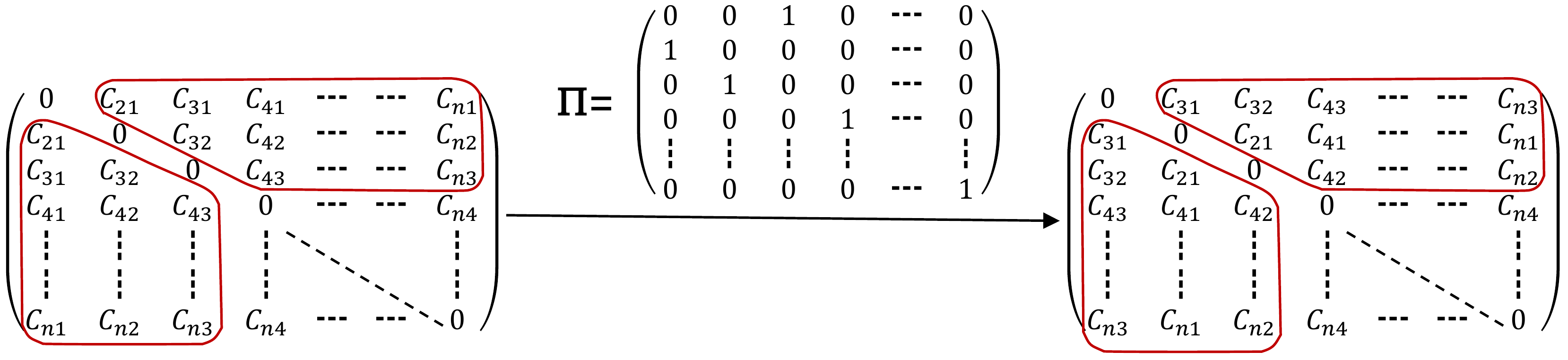}
		\caption{\small An example of the effect of $\mathbf{\Pi} \in S_3(\mathbf{\tilde{\Pi}})$, where we set $\mathbf{\tilde{\Pi}}=\mathbf{I}$. $\mathbf{I}$ is the identity matrix. Note that under the $\mathbf{\Pi}$ above the arrow, which differs from $\mathbf{I}$ only in the first three rows (columns). Thus the possible difference between two matrices only exists in the red circles, with $6n-6$ elements in the matrix involved. }
		\label{fig:loctwofig}
\end{figure}

Since $G_1$ and $G_2$ are independently sampled from $G$, then $\mathbf{A}$ and $\mathbf{B}$ are independent. Thus we can first take the expectation of $\mathbf{B}$ on both sides of Eqn. (\ref{equation33}). Note that the probability for the edge existence between nodes $i$ and $j$ in $\mathbf{B}$ is $p_{\boldsymbol{C_iC_j}}s_2$, therefore
$\textbf{E}[\psi(B_{ij})]=(-1)p_{\boldsymbol{C_iC_j}}s_2+(1-p_{\boldsymbol{C_iC_j}}s_2)
=1-2p_{\boldsymbol{C_iC_j}}s_2$. Hence, taking the expectation of $\mathbf{B}$ on both sides of Eqn. (\ref{equation33}) and we obtain

\begin{equation}
\begin{aligned}
&\textbf{E}_{\mathbf{B}}(||\mathbf{\Pi_0}\mathbf{\hat{A}}-\mathbf{\hat{B}}\mathbf{\Pi_0}||_F^2-||\mathbf{\tilde{\Pi}}\mathbf{\hat{A}}-\mathbf{\hat{B}}\mathbf{\tilde{\Pi}}||_F^2)
\\&=2\sum_{k\neq i,j}^{n}[w_{ik}[(1-2p_{\boldsymbol{C_iC_k}}s_2)-w_{jk}(1-2p_{\boldsymbol{C_jC_k}}s_2)) \\& \quad \cdot((\mathbf{\Pi_0} \mathbf{A}\mathbf{\Pi_0^T})_{ki}-(\mathbf{\Pi_0} \mathbf{A}\mathbf{\Pi_0^T})_{kj})].\nonumber
\end{aligned}
\end{equation}
Similarly, taking the expectation of $\mathbf{A}$ on both sides, we have
\begin{equation}
\label{e19}
\begin{aligned}
&\textbf{E}_{\mathbf{A},\mathbf{B}}(||\mathbf{\Pi_0}\mathbf{\hat{A}}-\mathbf{\hat{B}}\mathbf{\Pi_0}||_F^2-||\mathbf{\tilde{\Pi}}\mathbf{\hat{A}}-\mathbf{\hat{B}}\mathbf{\tilde{\Pi}}||_F^2)
\\&=2\sum_{k\neq i,j}^{n}(w_{ik}(1-2p_{\boldsymbol{C_iC_k}}s_2)-w_{jk}(1-2p_{\boldsymbol{C_jC_k}}s_2)) \\& \quad \cdot( p_{\boldsymbol{C_{\pi_0(i)}C_{{\pi_0}(k)}}}-p_{\boldsymbol{C_{{\pi_0}(j)}C_{{\pi_0}(k)}}})s_1
\\& =2\sum_{k\neq i,j}^n \Delta_{i,j,k,\pi_0},\nonumber
\end{aligned}
\end{equation}
where
\begin{equation}
\begin{aligned}
\Delta_{i,j,k,\pi_0}&=(w_{ik}(1-2p_{\boldsymbol{C_iC_k}}s_2)-w_{jk}(1-2p_{\boldsymbol{C_jC_k}}s_2)) \\& \quad \cdot( p_{\boldsymbol{C_{\pi_0(i)}C_{{\pi_0}(k)}}}-p_{\boldsymbol{C_{{\pi_0}(j)}C_{{\pi_0}(k)}}})s_1.
\nonumber
\end{aligned}
\end{equation}


$\Delta_{i,j,k,\pi_0}$ reflects a part of the difference $||\mathbf{\Pi_0}\mathbf{\hat{A}}-\mathbf{\hat{B}}\mathbf{\Pi_0}||_F^2-||\mathbf{\tilde{\Pi}}\mathbf{\hat{A}}-\mathbf{\hat{B}}\mathbf{\tilde{\Pi}}||_F^2$ caused by the difference of a single element in matrices $\mathbf{\Pi_0}\mathbf{\hat{A}}-\mathbf{\hat{B}}\mathbf{\Pi_0}$ and $\mathbf{\tilde{\Pi}}\mathbf{\hat{A}}-\mathbf{\hat{B}}\mathbf{\tilde{\Pi}}$.\footnote{For example, the difference of the corresponding element (with the same notation, e.g., (i,k) in the left matrix and (j,k) in the right matrix, both of which are triangles.)  in two matrices in Fig. \ref{Fig:example} inside one of the red circles}
Since we consider the average case of all possible $\mathbf{\Pi_0}$, we also consider the average value of $\Delta_{i,j,\pi_0}$,
which we set as $\hat{\Delta}=\textbf{E}_{i,j,\pi_0}(\Delta_{i,j,\pi_0})$.
Note that $\textbf{E}_{\mathbf{A},\mathbf{B}}(||\mathbf{\Pi_0}\mathbf{\hat{A}}-\mathbf{\hat{B}}\mathbf{\Pi_0}||_F^2-||\mathbf{\tilde{\Pi}}\mathbf{\hat{A}}-\mathbf{\hat{B}}\mathbf{\tilde{\Pi}}||_F^2)>0$
since $\mathbf{\tilde{\Pi}}$ is the minimizer of $||\mathbf{\Pi_0}\mathbf{\hat{A}}-\mathbf{\hat{B}}\mathbf{\Pi_0}||_F^2$. Therefore $\hat{\Delta}=\textbf{E}_{i,j,\pi_0}(\Delta_{i,j,\pi_0})>0$.

\underline{2. Analysis of $\sum_{\mathbf{\Pi_0} \in S_k(\mathbf{\Pi})} ||\mathbf{\Pi_0 \hat{A} - \hat{B} \Pi_0}||_F^2$}

Now we move to the second part involved in our idea.  We first focus on $S_k(\mathbf{\Pi_0})$, and count the number of elements in $S_k(\mathbf{\Pi_0})$, denoted as $|S_k|$. Note that if there are $k$ mismatched nodes in a graph with $n$ nodes, there are $C_n^k$ possible sets of mismatched nodes. We define $|T_k|$ as the number of elements in each possible set, and can get $|S_k|=C_n^k|T_k|$. For $|T_k|$, we can find that it satisfies
\begin{equation}
\label{eqn:25}
\begin{aligned}
|T_k|&=(k-1)(|T_{k-2}|+(k-2)(|T_{k-3}|+(k-3)(|T_{k-4}|+...)))
\\&=\sum_{t=1}^{k-1} (\prod_{i=1}^{t} (k-i))|T_{k-t-1}|.
\end{aligned}
\end{equation}

Consider $|T_{k}|$ and $|T_{k-1}|$ in Eqn. (\ref{eqn:25}), we can discover that
\begin{equation}
|T_k|=(k-1)(|T_{k-2}|+|T_{k-1}|)\geq (k-1)|T_{k-1}|, k\geq 2.\nonumber
\end{equation}

Therefore we obtain the relationship between $|S_{k}|$ and $|S_{k-1}|$ as
\begin{equation}
\label{eqn:27}
|S_k|=C_n^k|T_k|\geq (k-1) \frac{C_n^k}{C_n^{k-1}} |S_{k-1}|=(1-\frac{1}{k})(n-k+1)|S_{k-1}|,
\end{equation}
where $k \geq 2$. Eqn. (\ref{eqn:27}) shows that when $k$ is much smaller than $n$, then $\frac{|S_k|}{|S_{k-1}|}=(1-\frac{1}{k})(n-k+1)$ is large; when $k$ gets close to $n$, then $\frac{|S_k|}{|S_{k-1}|}$ approaches $1$, which means that $|S_k|$ and $|S_{k-1}|$ are almost the same.


Now we consider $\mathbf{\Pi_0}\in S_k(\mathbf{\Pi})$. Note that for any $\mathbf{\Pi_0}\in S_k(\mathbf{\Pi})$, there are $k$ rows and columns that may cause the difference between $||\mathbf{W}\circ(\mathbf{\Pi_0}\mathbf{{A}}\mathbf{\Pi_0^T}-\mathbf{{B}})||_F^2$
and $||\mathbf{W}\circ(\mathbf{\tilde{\Pi}}\mathbf{{A}}\mathbf{\tilde{\Pi}^T}-\mathbf{{B}})||_F^2$.
Fig. \ref{fig:loctwofig} illustrates an example of $\mathbf{\Pi_0} \in S_3(\mathbf{\Pi})$. Therefore we can discover for any $\mathbf{\Pi_0}\in S_k(\mathbf{\Pi})$, the number of node pairs $(i,j)$ which may influence the difference between $||\mathbf{W}\circ(\mathbf{\Pi_0}\mathbf{{A}}\mathbf{\Pi_0^T}-\mathbf{{B}})||_F^2$
and $||\mathbf{W}\circ(\mathbf{\tilde{\Pi}}\mathbf{{A}}\mathbf{\tilde{\Pi}^T}-\mathbf{{B}})||_F^2$ is approximately $\sum_{i=1}(n-i)=\frac{(2n-k-1)k}{2}$.\footnote{For example, in Fig. \ref{fig:loctwofig} when $k=3$ the number is $6n-6$. Although there may be some elements which do not cause error, such as the two stars in Fig. \ref{Fig:example}, the number of this kinds of node pairs can be neglected when $n$ is large enough.} Thus, denoting $N_k$ as this number of node pair, we can obtain
\begin{equation}
\begin{aligned}
N_k&=\frac{(2n-k-1)k}{2}|S_k|
\\&\geq \frac{(2n-k-1)k}{2}
(1-\frac{1}{k})(n-k+1)|S_{k-1}|
\\&=(1-\frac{1}{k})(n-k+1)\frac{(2n-k-1)k}{(2n-k)(k-1)}N_{k-1}
\\&=(n-k+1)\frac{2n-k-1}{2n-k}N_{k-1}.\nonumber
\end{aligned}
\end{equation}

Therefore in average, we have
\begin{equation}
\begin{aligned}
\label{i21}
&\sum_{\mathbf{\Pi_0} \in S_k} ||\mathbf{\Pi_0 \hat{A} - \hat{B} \Pi_0}||_F^2
= N_k \hat{\Delta}
\\& \geq (n-k+1)\frac{2n-k-1}{2n-k}N_{k-1} \hat{\Delta}
\\&\geq (n-k+1)\frac{2n-k-1}{2n-k}\sum_{\mathbf{\Pi_0} \in S_{k-1}} ||\mathbf{\Pi_0 \hat{A} - \hat{B} \Pi_0}||_F^2
\\&\approx (n-k+1)\sum_{\mathbf{\Pi_0} \in S_{k-1}} ||\mathbf{\Pi_0 \hat{A} - \hat{B} \Pi_0}||_F^2,
\end{aligned}
\end{equation}
where the last approximation holds since $k\leq n$ and when $n\rightarrow \infty$, $\frac{2n-k-1}{2n-k}\rightarrow 1$.

Therefore, we can claim that in average, if $k_1>k_2$, then
\begin{equation}
\sum_{\mathbf{\Pi_0} \in S_{k_1}} ||\mathbf{\Pi_0 \hat{A} - \hat{B} \Pi_0}||_F^2>
\sum_{\mathbf{\Pi_0} \in S_{k_2}} ||\mathbf{\Pi_0 \hat{A} - \hat{B} \Pi_0}||_F^2.
\end{equation}

\underline{3. Maximum Value Under Sequence Inequality}

Based on the analysis above, if we set $\mathbf{\Pi}=\mathbf{\tilde{\Pi}}$, then we find that when $k=0$, the minimum value in the set $\{0,2,3,\cdots,n\}$. Thus $\sum_{\mathbf{\Pi_0} \in S_0(\mathbf{\tilde{\Pi}})} ||\mathbf{\Pi_0 \hat{A} - \hat{B} \Pi_0}||_F^2=||\mathbf{\tilde{\Pi} \hat{A} - \hat{B} \tilde{\Pi}}||_F^2$ is also the minimum value in the set
\begin{equation}
\footnotesize
\left\{
\begin{aligned}
 &\sum_{\mathbf{\Pi_0} \in S_{0}(\mathbf{\tilde{\Pi}})} ||(\mathbf{\Pi_0 \hat{A} - \hat{B} \Pi_0})||_F^2, \sum_{\mathbf{\Pi_0} \in S_{2}(\mathbf{\tilde{\Pi}})} ||(\mathbf{\Pi_0 \hat{A} - \hat{B} \Pi_0})||_F^2,
 \\&\sum_{\mathbf{\Pi_0} \in S_{3}(\mathbf{\tilde{\Pi}})} ||(\mathbf{\Pi_0 \hat{A} - \hat{B} \Pi_0})||_F^2,...,\sum_{\mathbf{\Pi_0} \in S_{n}(\mathbf{\tilde{\Pi}})} ||(\mathbf{\Pi_0 \hat{A} - \hat{B} \Pi_0})||_F^2\nonumber
 \end{aligned}
 \right\}.
 \end{equation}

Thus according to Lemma \ref{lemma1}, we know that in average case, by setting $\mathbf{\Pi}$ in the original MMSE objective function
\begin{equation}
\sum_{\mathbf{\Pi_0} \in \Pi^{n}}||\mathbf{\Pi}-\mathbf{\Pi_0}||_F^2||\mathbf{W}\circ (\mathbf{\Pi_0} \mathbf{A}-\mathbf{B} \mathbf{\Pi_0})||_F^2 \nonumber
\end{equation}
equal to $\mathbf{\tilde{\Pi}}$, the minimizer of WEMP, then this original MMSE objective function reaches its largest value under Sequence Inequality.

Moreover, note that if we do not set $\mathbf{\tilde{\Pi}}=\mathbf{\hat{\Pi}}$, for example set $\mathbf{\tilde{\Pi}}=\mathbf{\Pi} \in S_k(\mathbf{\Pi})$, we can verify that $\mathbf{\Pi}$ does not make the objective function in Eqn. (\ref{e16}) larger than $\mathbf{\Pi_0}$ since
\begin{equation}
\begin{aligned}
&0\mathbf{||\mathbf{\Pi}\mathbf{\hat{A}}\mathbf{\Pi^T}-\mathbf{\hat{B}}||_F^2}+2k\mathbf{||\mathbf{\Pi}\mathbf{\hat{A}}\mathbf{\Pi^T}-\mathbf{\hat{B}}||_F^2}
\\&\geq 2k\mathbf{||\mathbf{\hat{\Pi}}\mathbf{\hat{A}}\mathbf{\hat{\Pi}^T}-\mathbf{\hat{B}}||_F^2}+0\mathbf{||\mathbf{\Pi}\mathbf{\hat{A}}\mathbf{\Pi^T}-\mathbf{\hat{B}}||_F^2}
\nonumber,
\end{aligned}
\end{equation}
which means that the Sequency Inequality preserves that when $||\mathbf{\Pi_0} \mathbf{\hat{A}} \mathbf{\Pi_0^T}-\mathbf{\hat{B}}||_F^2$ achieves its minimum, then $||\Pi-\Pi_0||_F^2$ also achieves its minimum. Therefore by setting $\mathbf{\tilde{\Pi}}=\mathbf{\hat{\Pi}}$ we can achieve the largest value of the original MMSE problem under this sequence inequality.

However, as we noted earlier, we can only transform the original MMSE problem into WEMP in an average case of network structures. This implies that the transformation is not necessarily the best approximation of a single network structure. In the following we further analyze the validity of this transformation in a possible network structure by showing the approximation ratio of our transformation is large (at least larger than $0.5$).

\subsubsection{The Validity of Transformation}

As we have stated above, $\mathbf{\tilde{\Pi}}$ does not necessarily achieve the maximum of the original MMSE problem for a specific network structure. That is to say there may exist error in $g(\mathbf{\tilde{\Pi}})$ and $g(\mathbf{\hat{\Pi}})$, where $g(\mathbf{\hat{\Pi}})$ is the maximum value of the original MMSE objective function and $g(\mathbf{\tilde{\Pi}})$ is the value of MMSE objective function when $\mathbf{\Pi}$ equals to the minimizer of WEMP.
If we demonstrate that this error can be bounded within a small range, then we can say that this approximation is \emph{valid}. Theorem \ref{th1} shows that under the mild condition indicated by Inequality (\ref{i21}), we can get approximation ratio $g(\mathbf{\tilde{\Pi}})/g(\mathbf{\hat{\Pi}})$ larger than $0.5$, which, to some extent, makes our estimation reasonable.

\begin{theorem}
\label{th1}
Given the published graph $G_1$, the auxiliary graph $G_2$, the parameter set $\boldsymbol{\theta}$ and the weight matrix $\mathbf{W}$, in average case
we have the approximation ratio $g(\mathbf{\tilde{\Pi}})/g(\mathbf{\hat{\Pi}})$ larger than $0.5$.
\end{theorem}

\begin{proof}
First we have
\begin{equation}
\begin{aligned}
\label{e28}
&g(\mathbf{\hat{\Pi}})-g(\mathbf{\tilde{\Pi}})
\\&=\sum_{\mathbf{\Pi_0} \in \mathbf{\Pi^n}}(||\mathbf{\hat{\Pi}}-\mathbf{\Pi_0}||_F^2-||\mathbf{\tilde{\Pi}}-\mathbf{\Pi_0}||_F^2)||\mathbf{\Pi_0}\mathbf{\hat{A}}-\mathbf{\hat{B}}\mathbf{\Pi_0}||_F^2.
\end{aligned}
\end{equation}

Then we divide the set $\mathbf{\Pi^n}$ into two subsets:
\begin{equation}
\begin{aligned}
&\mathbf{\Pi_1^n}=\{\mathbf{\Pi}\in \mathbf{\Pi^n}\boldsymbol{|}||\mathbf{\hat{\Pi}}-\mathbf{\Pi_0}||_F^2> ||\mathbf{\tilde{\Pi}}-\mathbf{\Pi_0}||_F^2 \};
\\ &\mathbf{\Pi_2^n}=\{\mathbf{\Pi}\in \mathbf{\Pi^n}\boldsymbol{|}||\mathbf{\hat{\Pi}}-\mathbf{\Pi_0}||_F^2< ||\mathbf{\tilde{\Pi}}-\mathbf{\Pi_0}||_F^2 \}.\nonumber
\end{aligned}
\end{equation}

Following that we divide the Eqn. (\ref{e28})
into two sets, $\mathbf{\Pi_1^n}$ and $\mathbf{\Pi_2^n}$:
\begin{equation}
\label{e30}
\begin{aligned}
&g(\mathbf{\hat{\Pi}})-g(\mathbf{\tilde{\Pi}})\\&=\sum_{\mathbf{\Pi_0}\in \mathbf{\Pi_1^n}}(||\mathbf{\hat{\Pi}}-\mathbf{\Pi_0}||_F^2-||\mathbf{\tilde{\Pi}}-\mathbf{\Pi_0}||_F^2)||\mathbf{\Pi_0}\mathbf{\hat{A}}-\mathbf{\hat{B}}\mathbf{\Pi_0}||_F^2
\\&\quad -\sum_{\mathbf{\Pi_0}\in \mathbf{\Pi_2^n}}(||\mathbf{\tilde{\Pi}}-\mathbf{\Pi_0}||_F^2-||\mathbf{\hat{\Pi}}-\mathbf{\Pi_0}||_F^2)||\mathbf{\Pi_0}\mathbf{\hat{A}}-\mathbf{\hat{B}}\mathbf{\Pi_0}||_F^2
\\&\leq ||\mathbf{\tilde{\Pi}}-\mathbf{\hat{\Pi}}||_F^2\sum_{\mathbf{\Pi_0}\in \mathbf{\Pi_1^n}}||\mathbf{\Pi_0}\mathbf{\hat{A}}-\mathbf{\hat{B}}\mathbf{\Pi_0}||_F^2.
\end{aligned}
\end{equation}
where the last inequality holds due to the triangular inequality $||\mathbf{\hat{\Pi}}-\mathbf{\Pi_0}||_F^2-||\mathbf{\tilde{\Pi}}-\mathbf{\Pi_0}||_F^2\leq ||\mathbf{\tilde{\Pi}}-\mathbf{\hat{\Pi}}||_F^2$ and the term $\sum_{\mathbf{\Pi_0}\in \mathbf{\Pi_2^n}}(||\mathbf{\tilde{\Pi}}-\mathbf{\Pi_0}||_F^2-||\mathbf{\hat{\Pi}}-\mathbf{\Pi_0}||_F^2)||\mathbf{\Pi_0}\mathbf{\hat{A}}-\mathbf{\hat{B}}\mathbf{\Pi_0}||_F^2$
is positive. Then we have
\begin{equation}
\label{i31}
\begin{aligned}
\frac{g(\mathbf{\hat{\Pi}})-g(\mathbf{\tilde{\Pi}})}{g(\mathbf{\tilde{\Pi}})}&=\frac{(||\mathbf{\tilde{\Pi}}-\mathbf{\hat{\Pi}}||_F^2)\sum_{\mathbf{\Pi_0}\in \mathbf{\Pi_1^n}}||\mathbf{\Pi_0}\mathbf{\hat{A}}-\mathbf{\hat{B}}\mathbf{\Pi_0}||_F^2}{\sum_{\mathbf{\Pi_0} \in \mathbf{\Pi^{n}}}||\mathbf{\tilde{\Pi}}-\mathbf{\Pi_0}||_F^2||\mathbf{\Pi_0} \mathbf{\hat{A}}-\mathbf{\hat{B}} \mathbf{\Pi_0}||_F^2}
\\&\leq \frac{2\beta n\sum_{\mathbf{\Pi_0}\in \mathbf{\Pi^n}}||\mathbf{\Pi_0}\mathbf{\tilde{A}}-\mathbf{\tilde{B}}\mathbf{\Pi_0}||_F^2}{\sum_{\mathbf{\Pi_0} \in \mathbf{\Pi^{n}}}||\mathbf{\tilde{\Pi}}-\mathbf{\Pi_0}||_F^2||\mathbf{\Pi_0} \tilde{\mathbf{A}}-\tilde{\mathbf{B}} \mathbf{\Pi_0}||_F^2}.
\end{aligned}
\end{equation}
where $||\mathbf{\tilde{\Pi}}-\mathbf{\hat{\Pi}}||_F^2=2\beta n$ and $\beta \in [0,1]$ is the ratio between the number of mistakenly matched nodes and that of all the nodes. The last inequality in ($\ref{i31}$) holds because $\mathbf{\Pi_1^n} \subset \mathbf{\Pi^n}$.


Now we divide the sum $\sum_{\mathbf{\Pi_0}\in \mathbf{\Pi^n}}||\mathbf{\Pi_0}\mathbf{\hat{A}}-\mathbf{\hat{B}}\mathbf{\Pi_0}||_F^2$
into two parts:
\begin{equation}
D_1=\sum_{k\leq\rho n} \sum_{\mathbf{\Pi_0}\in \mathbf{\Pi^n}}||\mathbf{\Pi_0}\mathbf{\hat{A}}-\mathbf{\hat{B}}\mathbf{\Pi_0}||_F^2;
\nonumber
\end{equation}
\begin{equation}
D_2=\sum_{\rho n< k\leq n}\sum_{\mathbf{\Pi_0}\in \mathbf{\Pi^n}}||\mathbf{\Pi_0}\mathbf{\hat{A}}-\mathbf{\hat{B}}\mathbf{\Pi_0}||_F^2.
\nonumber
\end{equation}
where $\rho$ is any real number in $[0,1]$ and we assume that $\rho n$ is an integer\footnote{If it is not an integer, we can easily modify it by rounding.}.

For $D_1$, in average case we can obtain
\begin{equation}
\small
\begin{aligned}
D_1&\leq \sum_{i=1}^{\rho n}
\sum_{\mathbf{\Pi_0}\in \mathbf{\Pi^n}}||\mathbf{\Pi_0}\mathbf{\hat{A}}-\mathbf{\hat{B}}\mathbf{\Pi_0}||_F^2
\leq \sum_{i=1}^{\rho n} \prod_{j=1}^{i} 2(n-j+1)
\\&\leq \sum_{i=1}^{\rho n} (2n)^{i}
=2n\frac{(2n)^{\rho n}-1}{2n-1}
\approx (2n)^{\rho n}.\nonumber
\end{aligned}
\end{equation}

\vspace{-3mm}
For $D_2$, according to Inequality (\ref{i21}), in average case we can get
\begin{equation}
\small
\begin{aligned}
D_2&\geq\sum_{k=\rho n+1}^{n} \prod_{j=1}^k (n-j+1)=
\sum_{k=\rho n+1}^{n} \frac{n!}{(n-k)!}
\\&\geq \sum_{k=\rho n+1}^{n} \frac{n!}{((1-\rho)n)!}
=(1-\rho)n \frac{n!}{((1-\rho)n)!}.\nonumber
\end{aligned}
\end{equation}
Note that if we set $\rho=\Omega(1)=c_0$, where $c_0\rightarrow 1$, then
$\rho\rightarrow 1$  and
\begin{equation}
D_2\geq c_0\frac{n!}{c_0!}
=cn!
\sim c\sqrt{2\pi n}(\frac{n}{e})^n,\nonumber
\end{equation}
where $c$ is a constant and the last step holds due to the Stirling's formula. Therefore we can upper bound $\frac{D_2}{D_1}$ as
\begin{equation}
\frac{D_2}{D_1}\geq c\frac{\sqrt{2\pi n}(\frac{n}{e})^n}{(2n)^{\rho n}}
=c\sqrt{2\pi n}\bigg(\frac{n^{1-\rho}}{2^{\rho}e}\bigg)^n.\nonumber
\end{equation}
Then if $\rho$ is a constant which approaches $1$ but does not equal to $1$, then we find that when $n\rightarrow \infty$, $D_2$ is of higher order of $n$ than $D_1$. Therefore we can easily verify that in the denominator of the last term in Inequality (\ref{i31}), $\sum_{\rho n <k\leq n}\sum_{\mathbf{\Pi_0} \in \mathbf{\Pi^{n}}}||\mathbf{\tilde{\Pi}}-\mathbf{\Pi_0}||_F^2||\mathbf{\Pi_0} \mathbf{\hat{A}}-\mathbf{\hat{B}} \mathbf{\Pi_0}||_F^2$
is of higher order of $n$ than
$\sum_{k\leq \rho n}\sum_{\mathbf{\Pi_0} \in \mathbf{\Pi^{n}}}||\mathbf{\tilde{\Pi}}-\mathbf{\Pi_0}||_F^2||\mathbf{\Pi_0} {\mathbf{\hat{A}}}-{\mathbf{\hat{B}}} \mathbf{\Pi_0}||_F^2$, since for $k_1>\rho n$ and $k_2<\rho n$, $\mathbf{\Pi_1^{\prime}} \in S_{k_1}(\mathbf{\tilde{\Pi}})$ and $\mathbf{\Pi_2^{\prime}} \in S_{k_2}(\mathbf{\tilde{\Pi}})$, we have $||\mathbf{\Pi_1^{\prime}}-\mathbf{\tilde{\Pi}}||_F^2\geq
||\mathbf{\Pi_2^{\prime}}-\mathbf{\tilde{\Pi}}||_F^2$. Therefore according to Lemma \ref{lemma2}, we can leave the term with highest order of $n$ in the denominator and numerator in the last term in Inequality (\ref{i31}) when $n\rightarrow \infty$ and thus we can obtain
\begin{equation}
\begin{aligned}
&\frac{2\beta n\sum_{\mathbf{\Pi_0}\in \mathbf{\Pi^n}}||\mathbf{\Pi_0}\mathbf{\hat{A}}-\mathbf{\hat{B}}\mathbf{\Pi_0}||_F^2}{\sum_{\mathbf{\Pi_0} \in \mathbf{\Pi^{n}}}||\mathbf{\tilde{\Pi}}-\mathbf{\Pi_0}||_F^2||\mathbf{\Pi_0} {\mathbf{\hat{A}}}-{\mathbf{\hat{B}}} \mathbf{\Pi_0}||_F^2}
\\&\approx \frac{2\beta n\sum_{\rho n< k\leq n}\sum_{\mathbf{\Pi_0}\in S_k\mathbf{(\mathbf{\tilde{\Pi}})}}||\mathbf{\Pi_0}\mathbf{\hat{A}}-\mathbf{\hat{B}}\mathbf{\Pi_0}||_F^2}
{\sum_{\rho n< k\leq n}\sum_{\mathbf{\Pi_0}\in S_k\mathbf{(\mathbf{\tilde{\Pi}})}}||\mathbf{\Pi_0-\tilde{\Pi}}||_F^2||\mathbf{\Pi_0}\mathbf{\hat{A}}-\mathbf{\hat{B}}\mathbf{\Pi_0}||_F^2}
\\&\leq \frac{2\beta n}{2\rho n}\frac{\sum_{\rho n< k\leq n}\sum_{\mathbf{\Pi_0}\in S_k\mathbf{(\mathbf{\tilde{\Pi}})}}||\mathbf{\Pi_0}\mathbf{\hat{A}}-\mathbf{\hat{B}}\mathbf{\Pi_0}||_F^2}{\sum_{\rho n< k\leq n}\sum_{\mathbf{\Pi_0}\in S_k\mathbf{(\mathbf{\tilde{\Pi}})}}||\mathbf{\Pi_0}\mathbf{\hat{A}}-\mathbf{\hat{B}}\mathbf{\Pi_0}||_F^2}
=\frac{\beta}{\rho}.\nonumber
\end{aligned}
\end{equation}


%

Thus we have the approximation ratio
\begin{equation}
\frac{g(\mathbf{\tilde{\Pi}})}{g(\mathbf{\hat{\Pi}})}
\geq \frac{1}{1+\frac{\beta}{\rho}}
\approx \frac{1}{1+\beta}
\geq \frac{1}{2}.\nonumber
\end{equation}
\end{proof}

Note that in the proof of Theorem \ref{th1}, we use several times of inequality scaling method to derive the lower bound of approximation ratio, which is $0.5$. These inequality scaling may cause this lower bound to be smaller than the real approximation ratio. That is to say, the approximation ratio $0.5$ may be even worse than the approximation ratio in the worst case in real situations. For example, in Inequality (\ref{i31}) we directly use
\begin{equation}
\sum_{\mathbf{\Pi_0}\in \mathbf{\Pi_1^n}}||\mathbf{\Pi_0}\mathbf{\hat{A}}-\mathbf{\hat{B}}\mathbf{\Pi_0}||_F^2\leq \sum_{\mathbf{\Pi_0}\in \mathbf{\Pi^n}}||\mathbf{\Pi_0}\mathbf{\hat{A}}-\mathbf{\hat{B}}\mathbf{\Pi_0}||_F^2 ,\nonumber
\end{equation}
which may cause a big gap. Therefore, for a more general situation we have the following corollary.

\begin{corollary}
Given the published graph $G_1$, the auxiliary graph $G_2$, the parameter set $\boldsymbol{\theta}$ and the weight matrix $\mathbf{W}$, and we let
\begin{equation}
\chi=\bigg(\sum_{\mathbf{\Pi_0}\in \mathbf{\Pi_1^n}}||\mathbf{\Pi_0}\mathbf{\hat{A}}-\mathbf{\hat{B}}\mathbf{\Pi_0}||_F^2\bigg)\bigg/\bigg(\sum_{\mathbf{\Pi_0}\in \mathbf{\Pi^n}}||\mathbf{\Pi_0}\mathbf{\hat{A}}-\mathbf{\hat{B}}\mathbf{\Pi_0}||_F^2\bigg),
\nonumber
\end{equation}
then in average case, the approximation $g(\mathbf{\tilde{\Pi}})/g(\mathbf{\hat{\Pi}})$ ratio is larger than $\frac{1}{1+\beta\chi}$.
\end{corollary}

This corollary can be easily proved by slightly changing the form of Eqn. (\ref{i31}). To take an example to illustrate the gap of approximation ratio caused by $\chi$ more intuitively, we assume that $\sum_{\mathbf{\Pi}\in\mathbf{\Pi_1^n}}||\mathbf{\Pi_0}\mathbf{\hat{A}}-\mathbf{\hat{B}}\mathbf{\Pi_0}||_F^2
=\sum_{\mathbf{\Pi}\in\mathbf{\Pi_2^n}}||\mathbf{\Pi_0}\mathbf{\hat{A}}-\mathbf{\hat{B}}\mathbf{\Pi_0}||_F^2$
\footnote{This is only a very special situation, which we use it to make an intuitive example to explain how $\chi$ causes the gap of approximation ratio. It is not necessarily the same as real situations}.
Then $\chi=\frac{1}{2}$ and $(\frac{1}{1+\beta\chi})>\frac{2}{3}$, which causes the gap of the lower bound of approximation ratio to be $\frac{2}{3}-\frac{1}{2}=\frac{1}{6}$.

Note that we still claim that the approximation ratio is \emph{larger} than $(\frac{1}{1+\beta\chi})$. This is because we eliminate the sum $\sum_{\mathbf{\Pi_0}\in \mathbf{\Pi_2^n}}(||\mathbf{\hat{\Pi}}-\mathbf{\Pi_0}||_F^2-||\mathbf{\tilde{\Pi}}-\mathbf{\Pi_0}||_F^2)||\mathbf{\Pi}\mathbf{\hat{A}}-\mathbf{\hat{B}}\mathbf{\Pi}||_F^2$ in Eqn. (\ref{e30}), which also generates a gap between the lower bound $\frac{1}{1+\beta\chi}$ and the real approximation ratio. We leave it a future direction to find a proper estimation of this gap. However, the current gap still ensures the real approximation ratio strictly larger than $\frac{1}{1+\beta\chi}$, which further strengthens our claim at the beginning of Section \ref{estimateMMSE} that the transformation of the original MMSE problem is valid.

\section{Algorithmic Aspect of De-anonymization Problem}\label{Algorithmic}

In this section, we show that WEMP is of significant advantages in seedless de-anonymization since it resolves the tension between \emph{optimality} and \emph{complexity}. For optimality, We prove the good performance of solving WEMP that the result makes the node mapping error (NME) negligible in large social networks under mild conditions, facilitated by higher overlapping strength; For complexity, the optimal mapping of WEMP, $\mathbf{\tilde{\Pi}}$, can be perfectly sought algorithmically by our convex-concave based de-anonymization algorithm (CBDA).

\subsection{The Influence of Transformation to WEMP on NME}
\label{algorithmicA}
Recall that our aim is to minimize NME in expectation, thus a natural question arises: \emph{how much NME $\mathbf{\tilde{\Pi}}$ may cause for any probable real permutation matrix $\mathbf{\Pi_0}$?} The answer reflects the ability of solving WEMP in enhancing mapping accuracy. To answer it, we demonstrate that under mild conditions, the \emph{relative NME}, defined as $\frac{||\mathbf{\tilde{\Pi}}-\mathbf{\Pi_0}||_F^2}{||\mathbf{\Pi_0}||_F^2}$, vanishes to $0$ as $n \rightarrow \infty$. This implies that under large network size, NME caus
ed by $\mathbf{\tilde{\Pi}}$ is negligible compared with $|V|=n$. Furthermore, we surprisingly find that the conditions are facilitated under higher overlapping strength, explicitly delineating benefits brought by overlapping communities in NME reduction.
Theorem \ref{th3} formally presents our result mentioned above. Before that, we give Lemma \ref{lemma5}, a prerequisite in proving Theorem \ref{th3}.

\begin{lemma}
\label{lemma5}
Suppose the permutation matrix $\mathbf{\Pi}$ keeps invariant of the community representation of all the nodes, i.e., $\forall \mathbf{\Pi}$ such that $\mathbf{\Pi}(i,j)=1$, $\boldsymbol{C_i}=\boldsymbol{C_j}$, then $\mathbf{\hat{A}}=\mathbf{W}\circ\mathbf{A}$, $\mathbf{\hat{B}}=\mathbf{W}\circ\mathbf{B}$ and
\begin{equation}
\label{aaa}
||\mathbf{\Pi\hat{A}-\hat{B}\Pi}||_F=||\mathbf{W}\circ(\mathbf{\Pi{A}\Pi^T-{B}})||_F
=||\mathbf{\Pi\hat{A}\Pi^T-\hat{B}}||_F.
\end{equation}
\end{lemma}

\begin{proof}
We know $||\mathbf{\Pi\hat{A}-\hat{B}\Pi}||_F=||\mathbf{W}\circ(\mathbf{\Pi{A}-{B}{\Pi}})||_F$, thus we only need to prove that $\mathbf{W}\circ\mathbf{\Pi{A}}=\mathbf{\Pi W}\circ \mathbf{A}$. Note that $w_{ij}$ only depends on $p_{\boldsymbol{C_iC_j}}$, $s_1$ and $s_2$, therefore
for some nodes $i,j,s,t$, if $\boldsymbol{C_i}=\boldsymbol{C_s}$ and $\boldsymbol{C_j}=
\boldsymbol{C_t}$, then $\mathbf{W}(i,j)=\mathbf{W}(s,t)$. This fact tells that the weight is invariant within communities. Therefore, since $\mathbf{\Pi}$ keeps invariant of the community representation of all the nodes, it is easy to verify that $\mathbf{W}\circ\mathbf{\Pi{A}}=\mathbf{\Pi W}\circ \mathbf{A}$. Thus we have $\mathbf{\hat{A}}=\mathbf{W}\circ\mathbf{A}$ and similarly, $\mathbf{\hat{B}}=\mathbf{W}\circ\mathbf{B}$.
Then Eqn. (\ref{aaa}) holds naturally.
\end{proof}

\textbf{Remark:} According to Lemma \ref{lemma5}, we can similarly show that $||\mathbf{W}\circ(\mathbf{{A}-\Pi{B}\Pi^T})||_F
=||\mathbf{\hat{A}-\Pi\hat{B}\Pi^T}||_F$, and there are no differences in form between $||\mathbf{\Pi_1\hat{A}\Pi_1^T-\hat{B}}||_F$ and $||\mathbf{\hat{A}-\Pi_2\hat{B}\Pi_2^T}||_F$ since the mappings are bijections and we can simply set $\mathbf{\Pi_2}=\mathbf{\Pi_1^T}$. Therefore, in the following we do not distinguish the forms $||\mathbf{\Pi\hat{A}\Pi^T-\hat{B}}||_F$ and
$||\mathbf{\hat{A}-\Pi\hat{B}\Pi^T}||_F$.


\begin{theorem}
\label{th3}
Given the published network $G_1$, the auxiliary network $G_2$, the parameter set $\boldsymbol{\theta}$, the weight matrix $\mathbf{W}$. Set $\mathbf{A}$ as the adjacent matrix of $G_1$, and $\mathbf{B}$ as the adjacent matrix of $G_2$. Set
$\tilde{p}_{\boldsymbol{C_i C_j}}=w_{ij}p_{\boldsymbol{C_i C_j}}$ and
\begin{equation}
\small
K=\min_{s,t,j}\{ (\tilde{p}_{\boldsymbol{C_s C_j}}+\tilde{p}_{\boldsymbol{C_t C_j}})\min\{s_1, s_2\}\},\nonumber
\end{equation}
\begin{equation}
\small
L=\max_{s,t,j} \{[(\tilde{p}_{\boldsymbol{C_s C_j}}+\tilde{p}_{\boldsymbol{C_t C_j}})\max\{s_1, s_2\} ]^2\}.\nonumber
\end{equation}
If the following four conditions:
\begin{itemize}
\item $\frac{L}{K}=o(1)$;
\item the minimizer of WEMP, $\mathbf{\tilde{\Pi}}$, satisfies that ${||{\mathbf{\hat{A}}}-\mathbf{\Pi_0}{\mathbf{\hat{B}}}\mathbf{\Pi_0^T}||_F^2}/{||{\mathbf{\hat{A}}}-\mathbf{\tilde{\Pi}}{\mathbf{\hat{B}}}\mathbf{\tilde{\Pi}^T}||_F^2}=
\Omega(1)$;
\item ${||{\mathbf{\hat{A}}}-\mathbf{\Pi_0}{\mathbf{\hat{B}}}\mathbf{\Pi_0^T}||_F^2}=o(Kn^2)$;
\item $\mathbf{\Pi_0}$ and $\mathbf{\tilde{\Pi}}$ keep invariant of the community representation of all the nodes,
\end{itemize}
hold, then the relative NME, $\frac{||\mathbf{\tilde{\Pi}}-\mathbf{\Pi_0}||_F^2}{||\mathbf{\Pi_0}||_F^2}$,  can be upper bounded by the minimum value of WEMP, i.e., $||{\mathbf{\hat{A}}}-\mathbf{\tilde{\Pi}}{\mathbf{\hat{B}}}\mathbf{\tilde{\Pi}^T}||_F^2$, and as $n\rightarrow \infty$, $\frac{||\mathbf{\tilde{\Pi}}-\mathbf{\Pi_0}||_F^2}{||\mathbf{\Pi_0}||_F^2} \rightarrow 0$.

\end{theorem}

\begin{proof}
We divide our proof into four main parts. Firstly, we start from $||\mathbf{\tilde{\Pi}}-\mathbf{\Pi_0}||_F$ and upper bound it using $||(\mathbf{\tilde{\Pi}}-\mathbf{\Pi_0})\mathbf{\hat{B}}||_F$ (or equivalently $||(\mathbf{\tilde{\Pi}}-\mathbf{\Pi_0})\mathbf{\hat{A}}||_F$). Secondly, we find the relationship between $||(\mathbf{\tilde{\Pi}}-\mathbf{\Pi_0})\mathbf{\hat{B}}||_F$
and $\textbf{tr}((\mathbf{\tilde{\Pi}-\Pi_0})\mathbf{\hat{B}}(\mathbf{(\tilde{\Pi}-\Pi_0)^T})\mathbf{\hat{A}})$.
Thirdly we upper bound the $\textbf{tr}((\mathbf{\tilde{\Pi}-\Pi_0})\mathbf{\hat{B}}(\mathbf{(\tilde{\Pi}-\Pi_0)^T})\mathbf{\hat{A}})$
and finally we upper bound $\frac{||(\mathbf{\Pi_0}-\mathbf{\tilde{\Pi}})||_F^2}{||\mathbf{\Pi_0}||_F^2}$, the relative NME, based on the first three steps.

\underline{1. Relationship Between $||\mathbf{\Pi_0}-\mathbf{\tilde{\Pi}}||_F$
and
$||(\mathbf{\Pi_0}-\mathbf{\tilde{\Pi}})\mathbf{\hat{B}}||_F$}

We start with the first part and focus on $||(\mathbf{\Pi_0}-\mathbf{\tilde{\Pi}})\mathbf{\hat{B}}||_F$.
For the $i_{th}$ row of $(\mathbf{\Pi_0}-\mathbf{\tilde{\Pi}})$, there are two possibilities: (i)
If $\mathbf{\Pi_0}$ and $\mathbf{\tilde{\Pi}}$ map node $i$ in $G_2$ to the same node in $G_1$, then the $i_{th}$ row of $(\mathbf{\Pi_0}-\mathbf{\tilde{\Pi}})\mathbf{\hat{B}}$ is a zero row vector; (ii) If $\mathbf{\Pi_0}$ and $\mathbf{\tilde{\Pi}}$ map node $i$ to node $s$ and $t$ respectively ($s\neq t$), then the $i_{th}$ row of $(\mathbf{\Pi_0}-\mathbf{\tilde{\Pi}})\mathbf{\hat{B}}$ is
$(\mathbf{\hat{B}}_{s1}-\mathbf{\hat{B}}_{t1}, \mathbf{\hat{B}}_{s2}-\mathbf{\hat{B}}_{t2}, \cdots, \mathbf{\hat{B}}_{sn}-\mathbf{\hat{B}}_{tn})$.
For an element, $([(\mathbf{\Pi_0}-\mathbf{\tilde{\Pi}})\mathbf{\hat{B}}]_{ij})^2=( \sqrt{w_{sj}}\mathbf{B}_{sj}-\sqrt{w_{tj}}\mathbf{B}_{tj})^2$.
Taking the expectation on both sides, we can derive that
\begin{equation}
\begin{aligned}
&\textbf{E}[(\mathbf{\Pi_0}-\mathbf{\tilde{\Pi}})\mathbf{\hat{B}}]_{ij}^2=\textbf{E}( \mathbf{\hat{B}}_{sj}-\mathbf{\hat{B}}_{tj})^2
\\&=(\tilde{p}_{\boldsymbol{C_s C_j}}+\tilde{p}_{\boldsymbol{C_t C_j}}-2\sqrt{w_{sj}w_{tj}}p_{\boldsymbol{C_s C_j}}p_{\boldsymbol{C_t C_j}}s_2)s_2
\\&\sim (\tilde{p}_{\boldsymbol{C_s C_j}}+\tilde{p}_{\boldsymbol{C_t C_j}})s_2.\nonumber
\end{aligned}
\end{equation}

So by summing up all the columns, we have
\begin{equation}
\begin{aligned}
\textbf{E}\sum_{j=1}^n[(\mathbf{\Pi_0}-\mathbf{\tilde{\Pi}})\mathbf{\hat{B}}]_{ij}^2
=\sum_{j=1}^n (\tilde{p}_{\boldsymbol{C_s C_j}}+\tilde{p}_{\boldsymbol{C_t C_j}})s_2.\nonumber
\end{aligned}
\end{equation}

Then summing up all the rows, we can obtain
\begin{equation}
\begin{aligned}
&||(\mathbf{\Pi_0}-\mathbf{\tilde{\Pi}})\mathbf{\hat{B}}||_F^2\\&=\textbf{E}\sum_{i=1}^n \sum_{j=1}^n[(\mathbf{\Pi_0}-\mathbf{\tilde{\Pi}})\mathbf{\hat{B}}]_{ij}^2
\\&=\sum_{i=1}^n  \mathbf{1}\{\pi_0(i)\neq \tilde{\pi}(i) \} \sum_{j=1}^n (\tilde{p}_{\boldsymbol{C_s C_j}}+\tilde{p}_{\boldsymbol{C_t C_j}})s_2
\\&\geq \sum_{i=1}^n  n \mathbf{1}\{\pi_0(i)\neq \tilde{\pi}(i) \} \min_{j} (\tilde{p}_{\boldsymbol{C_s C_j}}+\tilde{p}_{\boldsymbol{C_t C_j}})s_2,\nonumber
\end{aligned}
\end{equation}
where $\mathbf{1}\{\pi_0(i)\neq \tilde{\pi}(i)\}=1$ if $\pi_0$ and $\tilde{\pi}$ map node $i$ in $G_2$ to the same node in $G_1$ and $\mathbf{1}\{\pi_0(i)\neq \tilde{\pi}(i)\}=0$ otherwise. Thus it eliminates rows with all zero elements.

Note that $||(\mathbf{\Pi_0}-\mathbf{\tilde{\Pi}})||_F^2=2\sum_{i=1}^n \mathbf{1}\{\pi_0(i)\neq \tilde{\pi}(i)\}$. Setting
$\hspace{-0.5mm}K=\min_{s,t,j} (\tilde{p}_{\boldsymbol{C_s C_j}}+\tilde{p}_{\boldsymbol{C_t C_j}})s_2 $, we have
\begin{equation}
\label{Inequality:60}
||\mathbf{\Pi_0}-\mathbf{\tilde{\Pi}}||_F^2 \leq \frac{2}{nK} ||(\mathbf{\Pi_0}-\mathbf{\tilde{\Pi}})\mathbf{\hat{B}}||_F^2.
\vspace{-2mm}
\end{equation}
Similarly we can replace $\mathbf{\hat{B}}$ by $\mathbf{\hat{A}}$ and change $s_2$ to $s_1$ in $K$.


\underline{2. Relationship Between $||(\mathbf{\Pi_0}-\mathbf{\tilde{\Pi}})\mathbf{\hat{B}}||_F$
and} \underline{$\textbf{tr}((\mathbf{\tilde{\Pi}-\Pi_0})\mathbf{\hat{B}}(\mathbf{(\tilde{\Pi}-\Pi_0)^T})\mathbf{\hat{A}})$}

In the second part, note that
\begin{equation}
\label{in1}
\begin{aligned}
&||(\mathbf{\Pi_0}-\mathbf{\tilde{\Pi}})\mathbf{\hat{B}}||_F
=||\mathbf{\hat{B}}\mathbf{(\Pi_0-\tilde{\Pi})^T}||_F
=||\mathbf{\tilde{\Pi}}\mathbf{\hat{B}}\mathbf{(\Pi_0-\tilde{\Pi})^T}||_F
\\&\quad \leq ||(\mathbf{\tilde{\Pi}}\mathbf{\hat{B}}\mathbf{\Pi_0}-\mathbf{\hat{A}})-(\mathbf{\tilde{\Pi}}\mathbf{\hat{B}}\mathbf{\tilde{\Pi}}-\mathbf{\hat{A}})||_F
\\&\quad\leq ||\mathbf{\tilde{\Pi}} \mathbf{\hat{B}} \mathbf{\tilde{\Pi}^T}- \mathbf{\hat{A}}||_F+||\mathbf{\tilde{\Pi}} \mathbf{\hat{B}} \mathbf{\Pi_0^T}-\mathbf{\hat{A}}||_F,\nonumber
\end{aligned}
\end{equation}
where the second equation holds since the permutation matrix $\mathbf{\tilde{\Pi}}$ keeps invariant of Frobenius norm, and the second inequality holds due to the triangular inequality of Frobenius norm. Then we obtain
\begin{equation}
||(\mathbf{\Pi_0}-\mathbf{\tilde{\Pi}})\mathbf{\hat{B}}||_F^2\leq 2(||\mathbf{\tilde{\Pi}} \mathbf{\hat{B}} \mathbf{\tilde{\Pi}^T}- \mathbf{\hat{A}}||_F^2+||\mathbf{\tilde{\Pi}} \mathbf{\hat{B}} \mathbf{\Pi_0^T}-\mathbf{\hat{A}}||_F^2).\nonumber
\end{equation}

For the term $||\mathbf{\tilde{\Pi}} \mathbf{\hat{B}} \mathbf{\Pi_0^T}-\mathbf{\hat{A}}||_F^2$,
\begin{equation}
\label{Inequality:65}
\begin{aligned}
&||\mathbf{\tilde{\Pi}} \mathbf{\hat{B}} \mathbf{\Pi_0^T}-\mathbf{\hat{A}}||_F^2=
\textbf{tr}(\mathbf{(\mathbf{\tilde{\Pi}} \mathbf{\hat{B}} \mathbf{\Pi_0^T}-\mathbf{\hat{A}})^T}(\mathbf{\tilde{\Pi}} \mathbf{\hat{B}} \mathbf{\Pi_0^T}-\mathbf{\hat{A}}))
\\&=\textbf{tr}(\mathbf{\hat{A}^T}\mathbf{\hat{A}})+\textbf{tr}(\mathbf{\hat{B}^T}\mathbf{\hat{B}})
-2\textbf{tr}(\mathbf{\Pi_0}\mathbf{\hat{B}}\mathbf{\tilde{\Pi}^T}\mathbf{\hat{A}})
\\&=
||\mathbf{\hat{A}}||_F^2+||\mathbf{\hat{B}}||_F^2
-2\textbf{tr}(\mathbf{\Pi_0}\mathbf{\hat{B}}\mathbf{\tilde{\Pi}^T}\mathbf{\hat{A}})
\\&=\frac{1}{2}(||\mathbf{\tilde{\Pi}} \mathbf{\hat{B}} \mathbf{\tilde{\Pi}^T}-\mathbf{\hat{A}}||_F^2+||\mathbf{\Pi_0} \mathbf{\hat{B}} \mathbf{\Pi_0^T}-\mathbf{\hat{A}}||_F^2)
\\&\quad +
\textbf{tr}(\mathbf{\Pi_0}\mathbf{\hat{B}}\mathbf{\Pi_0^T}\mathbf{\hat{A}})+
\textbf{tr}(\mathbf{\tilde{\Pi}}\mathbf{\hat{B}}\mathbf{\tilde{\Pi}^T}\mathbf{\hat{A}})-
2\textbf{tr}(\mathbf{\Pi_0}\mathbf{\hat{B}}\mathbf{\tilde{\Pi}^T}\mathbf{\hat{A}})
\\&\leq ||\mathbf{\Pi_0} \mathbf{\hat{B}} \mathbf{\Pi_0}^T-\mathbf{\hat{A}}||_F^2+
\textbf{tr}((\mathbf{\tilde{\Pi}-\Pi_0})\mathbf{\hat{B}}(\mathbf{(\tilde{\Pi}-\Pi_0)^T})\mathbf{\hat{A}}),
\end{aligned}
\end{equation}
where the last equation can be verified by the first three equations, and the last inequality holds since $||\mathbf{\tilde{\Pi}} \mathbf{\hat{B}} \mathbf{\tilde{\Pi}^T}-\mathbf{\hat{A}}||_F^2\leq ||\mathbf{\Pi_0} \mathbf{\hat{B}} \mathbf{\Pi_0^T}-\mathbf{\hat{A}}||_F^2$.


\underline{3. Upper Bound of $\textbf{tr}((\mathbf{\tilde{\Pi}-\Pi_0})\mathbf{\hat{B}}\mathbf{(\tilde{\Pi}-\Pi_0)^T}\mathbf{\hat{A}})$}

Set $\mathbf{Z}=(\mathbf{\tilde{\Pi}-\Pi_0})\mathbf{\hat{B}}\mathbf{(\tilde{\Pi}-\Pi_0)^T}\mathbf{\hat{A}}$. Now we focus on $\textbf{tr}(\mathbf{Z})$. Note that the $i_{th}$ row of $\mathbf{\tilde{\Pi}-\Pi_0}$ is composed
of either zeros if $\mathbf{\tilde{\Pi}}$ and $\mathbf{\Pi_0}$ map node $i$ in $G_2$ to the same node in $G_1$, or zeros expect one $1$ and one $-1$ if $\mathbf{\tilde{\Pi}}$ and $\mathbf{\Pi_0}$ map node $i$ in $G_2$ to different nodes in $G_1$.
 It is easy to verify that
for any node $i$, when $\mathbf{\tilde{\Pi}}$ and $\mathbf{\Pi_0}$ map it to the same node, then $\mathbf{Z}_{ii}=0$. If not, for node $i$ we assume that
$\mathbf{\tilde{\Pi}}$ maps it to $s$ and $\mathbf{\Pi_0}$ maps it to $t$, where $s \neq t$.
For simplicity, we define $\mathbf{Y}=(\mathbf{\tilde{\Pi}-\Pi_0})\mathbf{\hat{B}}$
and $\mathbf{X}=(\mathbf{(\tilde{\Pi}-\Pi_0)^T})\mathbf{\hat{A}}$, thus $\mathbf{Z}=\mathbf{YX}$.
Then we can obtain the $i_{th}$ row of $\mathbf{Y}$
as
\begin{equation}
\mathbf{Y_{i\cdot}}=(\mathbf{\hat{B}}_{s1}-\mathbf{\hat{B}}_{t1}, \mathbf{\hat{B}}_{s2}-\mathbf{\hat{B}}_{t2}, \cdots, \mathbf{\hat{B}}_{sn}-\mathbf{\hat{B}}_{tn}).\nonumber
\end{equation}

 Similarly, we can obtain the $i_{th}$ column of $\mathbf{X}$ as
\begin{equation}
\mathbf{X_{\cdot i}}=(\mathbf{\hat{A}}_{p_1 1}-\mathbf{\hat{A}}_{q_1 1}, \mathbf{\hat{A}}_{p_2 2}-\mathbf{\hat{A}}_{q_2 2}, \cdots, \mathbf{\hat{A}}_{p_n n}-\mathbf{\hat{A}}_{q_n n} )^\mathbf{T},\nonumber
\end{equation}
where $p_i(q_i)$ means the row number of the $1(-1)$ in the $i_{th}$ column of $\mathbf{\tilde{\Pi}-\Pi_0}$, when $\mathbf{\tilde{\Pi}}$ and $\mathbf{\Pi_0}$ map node $i$ in $G_2$ into different nodes in $G_1$.
If they map node $j$ in $G_2$ to the same node in $G_1$, then we set $\mathbf{X}_{ji}=0$. Therefore for a single value on the diagonal of $\mathbf{Z}$,  i.e., $\mathbf{Z}_{ii}$, we can bound its absolute value as
\begin{equation}
\label{i42}
\begin{aligned}
|\mathbf{Z}_{ii}|&=| \langle \mathbf{Y_{i\cdot}} \mathbf{X_{\cdot i}} \rangle | \leq ||\mathbf{Y_{i\cdot}}||_F||\mathbf{X_{\cdot i}}||_F
\\& \leq n \max_{k} |\mathbf{\hat{B}}_{sk}-\mathbf{\hat{B}}_{tk}| \max_{\ell} |\mathbf{\hat{A}}_{p_{\ell} \ell}-\mathbf{\hat{A}}_{q_{\ell} \ell}|.
\end{aligned}
\end{equation}
Taking the expectation of $\mathbf{A}$ and $\mathbf{B}$ on both sides of Inequality (\ref{i42}), we can obtain that

\begin{equation}
\begin{aligned}
\textbf{E}_{\mathbf{A,B}}(|\mathbf{Z_{ii}}|)&=
\textbf{E}(\max_{s,t,k} |\mathbf{\hat{B}}_{sk}-\mathbf{\hat{B}}_{tk}| \max_{p,q,\ell} |\mathbf{\hat{A}}_{p_{\ell} \ell}-\mathbf{\hat{A}}_{q_{\ell} \ell}|)
\\&\leq \max_{s,t,j} \{[(p_{\boldsymbol{C_s C_j}}+p_{\boldsymbol{C_t C_j}})\max\{s_1, s_2\} ]^2\}
= L,\nonumber
\end{aligned}
\end{equation}
based on the Jensen's Inequality. Hence
\begin{equation}
\label{Inequality:70}
|\textbf{tr}((\mathbf{\tilde{\Pi}-\Pi_0})\mathbf{\hat{B}}(\mathbf{(\tilde{\Pi}-\Pi_0)^T})\mathbf{\hat{A}})| \leq n \max_{i} | \langle \mathbf{Y_{i\cdot}} \mathbf{X_{\cdot i}} \rangle | \leq n^2L.
\end{equation}

\underline{4. Upper Bound of $\frac{||(\mathbf{\Pi_0}-\mathbf{\tilde{\Pi}})||_F^2}{||\mathbf{\Pi_0}||_F^2}$}

Upon completion of the former three parts, now we can move to the final part. Specifically, from Inequalities (\ref{Inequality:60}), (\ref{Inequality:65}) and (\ref{Inequality:70}), we can obtain
\begin{equation}
\begin{aligned}
&||(\mathbf{\Pi_0}-\mathbf{\tilde{\Pi}})||_F^2 \leq
\frac{2}{nK} ||(\mathbf{\Pi_0}-\mathbf{\tilde{\Pi}})\mathbf{\hat{B}}||_F^2
\\&\leq \frac{8}{nK} ||\mathbf{\Pi_0} \mathbf{\hat{B}} \mathbf{\Pi_0^T}-\mathbf{\hat{A}}||_F^2+
2\textbf{tr}((\mathbf{\tilde{\Pi}-\Pi_0})\mathbf{\hat{B}}(\mathbf{(\tilde{\Pi}-\Pi_0)^T})\mathbf{\hat{A}})
 \\&\leq \frac{8}{nK} ||\mathbf{\Pi_0} \mathbf{\hat{B}} \mathbf{\Pi_0^T}-\mathbf{\hat{A}}||_F^2 +\frac{4nL}{K}.\nonumber
\end{aligned}
\end{equation}

Since $\mathbf{\tilde{\Pi}}$ is the minimizer of $||\mathbf{\hat{A}}-\mathbf{\Pi} \mathbf{\hat{B}} \mathbf{\Pi^T}||_F^2$ and the second condition, ${||{\mathbf{\hat{A}}}-\mathbf{\Pi_0}{\mathbf{\hat{B}}}\mathbf{\Pi_0^T}||_F^2}/{||{\mathbf{\hat{A}}}-\mathbf{\tilde{\Pi}}{\mathbf{\hat{B}}}\mathbf{\tilde{\Pi}^T}||_F^2}=
\Omega(1)$ holds, there exists a constant $\tilde{c}\geq 1$ such that $||\mathbf{\hat{A}}-\mathbf{{\Pi_0}} \mathbf{\hat{B}} \mathbf{{\Pi_0}^T}||_F \leq \tilde{c}||\mathbf{\hat{A}}-\mathbf{\tilde{\Pi}} \mathbf{\hat{B}} \mathbf{\tilde{\Pi}^T}||_F$.
Therefore since $||\mathbf{\Pi_0}||_F^2=2n$ and the first and third condition, we can bound the relative NME when $n \rightarrow \infty$ as:
\begin{equation}
\begin{aligned}
\frac{||(\mathbf{\Pi_0}-\mathbf{\tilde{\Pi}})||_F^2}{||\mathbf{\Pi_0}||_F^2}
&\leq \frac{4}{n^2K} ||\mathbf{\Pi_0} \mathbf{\hat{B}} \mathbf{\Pi_0^T}-\mathbf{\hat{A}}||_F^2+\frac{2L}{K}
\\&=\frac{4\tilde{c}}{n^2K} ||\mathbf{\tilde{\Pi}} \mathbf{\hat{B}} \mathbf{\tilde{\Pi}^T}-\mathbf{\hat{A}}||_F^2+\frac{2L}{K}\rightarrow0.
\nonumber
\end{aligned}
\end{equation}


This completes our proof.
\end{proof}

Theorem \ref{th3} demonstrates that under certain conditions, the relative NME goes to $0$ when the size of network tends to be infinity. Although this result does not show that the NME, expressed as $||\mathbf{\tilde{\Pi}}-\mathbf{\Pi_0}||_F^2$, vanishes under the conditions, it shows that compared with the number of nodes in the network, the NME can be neglected when the size of network is very large. This phenomenon makes sense in de-anonymization since it demonstrates that by minimizing the weighted-edge matching problem (WEMP), we can neglect the NME in large social networks and map most of the nodes correctly.


\textbf{The Positive Impact of Overlapping Communities on Theorem \ref{th3}:} Now we demonstrate that the overlapping communities exert a positive impact on diminishing the relative NME through making the conditions in Theorem \ref{th3} more prone to be satisfied. Specifically, when the overlapping strength in the networks becomes stronger, then the condition $3$ is easier to be met. We claim that condition $3$ is a decisive prerequisite for the vanish of relative NME, since conditions $2$ and $4$ are easy to meet by the common assumption that true mapping keeps invariant of the community representations and the additive constraint about communities, which we will discuss in Section \ref{algorithmdesign}. Therefore the overlapping strength holds a balance in vanishing the relative NME.

For convenience, we assume that $s=s_1=s_2$ in the following setting. Note that when the correct mapping $\pi_0$ keeps invariant of community representations, then on average condition $3$ can be written as
\begin{equation}
2\sum_{1\leq i<j\leq n} \log \left( \frac{1-p_{\boldsymbol{C_iC_j}}(2s-s^2)}{p_{\boldsymbol{C_iC_j}}(1-s)^2} \right) p_{\boldsymbol{C_iC_j}}s=o(Kn^2).
\end{equation}
To characterize the global situation in the networks, we define an average probability $\hat{p}$ such that
\begin{equation}
\begin{aligned}
&\sum_{1\leq i<j\leq n} \log \left( \frac{1-p_{\boldsymbol{C_iC_j}}(2s-s^2)}{p_{\boldsymbol{C_iC_j}}(1-s)^2} \right) p_{\boldsymbol{C_iC_j}}s\\&=\frac{n(n-1)}{2}\log \left( \frac{1-\hat{p}(2s-s^2)}{\hat{p}(1-s)^2} \right) \hat{p}s,
\end{aligned}
\end{equation}
where $\hat{p}$ is positively correlated to the overlapping strength of the whole networks. Taking the derivative of $\hat{p}$ over $\log \left( \frac{1-\hat{p}(2s-s^2)}{\hat{p}(1-s)^2} \right) \hat{p}s$, we find that
\begin{equation}
\frac{d( \log \left( \frac{1-\hat{p}(2s-s^2)}{\hat{p}(1-s)^2} \right) \hat{p}s)}{d\hat{p}}=\log \left( \frac{1-\hat{p}(2s-s^2)}{\hat{p}(1-s)^2} \right) s-\frac{1}{1-\hat{p}(2s-s^2)},
\end{equation}
and it is easy to verify that $\frac{d\left( \log \left( \frac{1-\hat{p}(2s-s^2)}{\hat{p}(1-s)^2} \right) \hat{p}s\right)}{d\hat{p}}$ is a decreasing function in terms of $\hat{p}$. Now focus on $\frac{d\left( \log \left( \frac{1-\hat{p}(2s-s^2)}{\hat{p}(1-s)^2} \right) \hat{p}s\right)}{d\hat{p}}$. If we consider dense communities such that $\hat{p}=1-o(1)$, which means that $\hat{p}$ asymptotically approaches $1$ (shown to be right under the Overlapping Stochastic Block Model(OSBM) below), then we can derive
\begin{equation}
\begin{aligned}
\log \left( \frac{1-\hat{p}(2s-s^2)}{\hat{p}(1-s)^2} \right) \hat{p}s&=\log \left( 1+\frac{1-\hat{p}}{\hat{p}(1-s)^2} \right) \hat{p}s
\\&\sim \frac{1-\hat{p}}{(1-s)^2}s = o(1),
\end{aligned}
\end{equation}
where $s=\Omega(1)$. Therefore if $\hat{p}$ is asymptotically close to $1$ as the overlapping strength enhances, then the order of ${||{\mathbf{\hat{A}}}-\mathbf{\Pi_0}{\mathbf{\hat{B}}}\mathbf{\Pi_0^T}||_F^2}$ turns smaller, which is more prone to satisfy ${||{\mathbf{\hat{A}}}-\mathbf{\Pi_0}{\mathbf{\hat{B}}}\mathbf{\Pi_0^T}||_F^2}=o(Kn^2)$.

Taking a vivid example of the overlapping stochastic block model (OSBM) in which
\begin{equation}
\label{exp39}
p_{\boldsymbol{C_iC_j}}=\frac{1}{1+ae^{-x}},
\end{equation}
where $a$ is an adjustable parameter and $x$ is the number of overlapping communities. We find that $\min_{i,j} p_{\boldsymbol{C_iC_j}}=\frac{1}{1+a}$ is a constant if $a=\Omega(1)$, and can be arbitrarily close to $1$ when $x$ is large enough. So if $s=o(1)$ and $\hat{p}=1-o(1)$, which means that the overlapping strength is very large, then
\begin{equation}
\begin{aligned}
\log \frac{1-\hat{p}(2s-s^2)}{p(1-s)^2}p&=\log (1+\frac{1-\hat{p}}{\hat{p}(1-s)^2})p
\\&\approx \frac{1-\hat{p}}{(1-s)^2}=o(\min_{i,j} p_{\boldsymbol{C_iC_j}})=o(1),
\end{aligned}
\end{equation}
thus condition $3$ holds. Meanwhile note that $s=o(1)$ makes condition $1$ hold as well. Therefore all the four conditions in Theorem \ref{th3} hold, thus the relative NME vanishes to $0$.

\subsection{Algorithm Design and Convergence Analysis}
\label{algorithmdesign}

In Sections \ref{estimateMMSE} and \ref{algorithmicA} we have verified the validity of the transformation from MMSE estimator to the weighted-edge matching problem (WEMP). In this section, we will propose an algorithm to solve WEMP and analyze its convergence.

\subsubsection{Formulation of WEMP in Constrained Optimization Form}\label{AlgorithmB1}

Before designing the algorithm, we first restate WEMP in the form of the following constrained optimization problem:

\begin{align}
\text{minimize}  \|(\mathbf{\hat{A}}-&\mathbf{\Pi}\mathbf{\hat{B}}\mathbf{\Pi^T})\|_\mathrm{F}^2 \nonumber\\[-2pt]
\text{\textbf{s.t. }}  \forall i\in V_1,\ \textstyle\sum_{i}\mathbf{\Pi}_{ij}&=1\label{P3:constraint1}\\[-2pt]
\forall j\in V_2,\ \textstyle\sum_{j}\mathbf{\Pi}_{ij}&=1 \label{P3:constraint2}\\[-2pt]
\forall i,j,\ \mathbf{\Pi}_{ij}\in&\{0,1\},\label{P3:constraint3}
\end{align}

Additionally, note that in previous sections we have assume that the true mapping between $G_1$ and $G_2$ should keep invariant of the community representation of every node before and after mapping. That is to say, the same user in $G_1$ and $G_2$ belongs to the same subset of communities, which is in line with real situations where there is no difference in communities in $G_1$ and $G_2$. To elaborate, let us recall Fig. \ref{fig:1}, where the communities in $G_1$ and $G_2$ are with no differences since the number of communities are the same and the corresponding communities in two networks contain the same subset of users. Here we point out that we keep this assumption in our algorithm design. Therefore, in order to obtain the correct mapping $\pi_0$, another constraint about community representation should be added, which is
\begin{equation}
\label{con1}
\forall i \in V_1, \boldsymbol{C_{i}}=\boldsymbol{C_{\pi(i)}}.
\end{equation}

Eqn. \ref{con1} means that our estimated mapping $\pi$ should keep the community representation of all the nodes in $V_1$ unchanged before and after mapping. Note that it is hard to implement this constraint directly in the optimization problem since it is not in the form of permutation matrix. However, we can easily convert it into a suitable one by defining a new matrix to characterize the community representation of all the nodes, which we call as ``Community Representation Matrix'', denoted as $\mathbf{M}$. Its formal definition is as follows.

 \begin{definition}{(\textbf{Community Representation Matrix})}
 \label{def11}
 Given a graph $G$ with $n$ nodes and $m$ communities, the community representation matrix of $G$ is an $n \times m$ matrix $\mathbf{M}$ which is composed of $0$s and $1$s, and $\forall i \in \{1,2,\cdots,n\}$, the $i_{th}$ row of $\mathbf{M}$ is the community representation of node $i$ in $G$.
 \end{definition}

 Take Fig. \ref{fig:1} as an instance again. The community representation matrix of $G$, denoted as $\mathbf{M_G}$, satisfies
 \begin{equation}
 \mathbf{M_G^T}=
 \begin{bmatrix}
 1 & 0 & 1 & 1 & 0 & 0 & 0 & 0 & 0 \\
 0 & 1 & 1 & 1 & 0 & 1 & 1 & 1 & 0 \\
 0 & 0 & 1 & 1 & 1 & 0 & 1 & 1 & 1
 \end{bmatrix}
 .\nonumber
 \end{equation}

 Note that the community representation matrices for $G$, $G_1$ and $G_2$ are identical. So we set all of them to be $\mathbf{M}$. Hence the constraint (\ref{con1}) can be rewritten as $||\mathbf{\Pi}\mathbf{M}-\mathbf{M}||_F^2=0$. According to optimization theory, we can form this constraint into the objective function by regarding it as the penalty term and obtain a new objective function
\begin{equation}
F_0(\mathbf{\Pi})=||\mathbf{\hat{A}}-\mathbf{\Pi}\mathbf{\hat{B}}\mathbf{\Pi^T}||_\mathrm{F}^2+\mu ||\mathbf{\Pi}\mathbf{M}-\mathbf{M}||_F^2,\nonumber
\end{equation}
where $\mu$ is an adjustable penalty parameter, which is large enough such that when the objective function reaches its minimum value, $||\mathbf{\Pi}\mathbf{M}-\mathbf{M}||_F^2$ is exactly or very close to $0$. Note that this transformation of objective function does not affect the previous analytical results of WEMP since we have the assumption that the true mapping keeps invariant of the community representation of every single node before and after mapping. Then with the aim of finding the true permutation matrix $\mathbf{\Pi_0}$, we must have $||\mathbf{\Pi_0}\mathbf{M}-\mathbf{M}||_F^2=0$, thus the objective function is the same as that of WEMP.

\subsubsection{Problem Relaxation and Idea of Algorithm Design}
Hereinafter, we focus on how we design our algorithm targeting the WEMP.

\textbf{Problem Relaxation:} WEMP is an integer program problem which cannot be solved efficiently. We relax the original feasible region of WEMP $\Omega_0$ into $\Omega$, which are respectively
\vspace{-1mm}
\begin{equation}
\begin{aligned}
\centering
&\Omega_0=\{\mathbf{\Pi}_{ij}\in\{0,1\}\boldsymbol{|}\forall i,j , \textstyle\sum_{i}\mathbf{\Pi}_{ij}=1\
, \textstyle\sum_{j}\mathbf{\Pi}_{ij}=1 \};
\\& \Omega=\{\mathbf{\Pi}_{ij}\in[0,1] \boldsymbol{|}\forall i,j , \textstyle\sum_{i}\mathbf{\Pi}_{ij}=1\
, \textstyle\sum_{j}\mathbf{\Pi}_{ij}=1 \}.\nonumber
\end{aligned}
\vspace{-2mm}
\end{equation}

After this relaxation the problem becomes tractable. However, a natural question arises: \emph{How to obtain the solution of the original unrelaxed problem from that of the relaxed problem}?

\textbf{Idea of Convex-Concave Optimization Method:}
Note that the minimizer of a concave function must be at the boundary of the feasible region,  coinciding that $\Omega_0$, the original feasible set, is just the boundary of $\Omega$. Therefore, a natural idea emerges: \emph{We can modify the convex relaxed problem into a concave problem gradually.} Thus we apply the convex-concave optimization method (CCOM), whose concept is pioneeringly proposed in \cite{zaslavskiy2009path} to solve graph matching problems: For $F_0(\mathbf{\Pi})$, 
we find its convex and concave relaxed version respectively $F_1(\mathbf{\Pi})$ and $F_2(\mathbf{\Pi})$.
Then we obtain a new objective function as
$F(\mathbf{\Pi})=(1-\alpha) F_1(\mathbf{\Pi})+\alpha F_2(\mathbf{\Pi})$. We modify $\alpha$ gradually from $0$ to $1$ with interval $\Delta\alpha$, each time solving the new $F(\mathbf{\Pi})$ initialized by the optimizer last time. $F(\mathbf{\Pi})$ becomes more concave, with its optimum closer to $\Omega_0$ where $\mathbf{\tilde{\Pi}}$ lies.

\subsubsection{\textbf{Implementation of CCOM and Algorithm Design}}

Although \cite{zaslavskiy2009path} has proposed the general framework of CCOM, the way it presents to obtain $F_1(\mathbf{\Pi})$ and $F_2(\mathbf{\Pi})$ is rather complex, as it involves Kronecker product and the Laplacian matrix of graphs. Here we provide a simple way, as defined in Lemma \ref{lemma3}, to get the convex relaxation and concave relaxation, for simplifying the objective function compared with that in \cite{zaslavskiy2009path}.

\begin{lemma}
\label{lemma3}
A proper way to get the convex relaxation and concave relaxation is
\begin{equation}
F_1(\mathbf{\Pi})=F_0(\mathbf{\Pi})+\frac{\lambda_{min}}{2} (n-||\mathbf{\Pi}||_F^2);\nonumber
\end{equation}
\begin{equation}
F_2(\mathbf{\Pi})=F_0(\mathbf{\Pi})+\frac{\lambda_{max}}{2} (n-||\mathbf{\Pi}||_F^2).\nonumber
\end{equation}
Therefore we form our new objective function in CCOM as
\begin{equation}
F(\mathbf{\Pi})=(1-\alpha)F_1(\mathbf{\Pi})+\alpha F_2(\mathbf{\Pi}) = F_0(\mathbf{\Pi})+2\xi(n-||\mathbf{\Pi}||_F^2),\nonumber
\end{equation}
where $\xi=(1-\alpha)\lambda_{min}+\alpha\lambda_{max}$,  $\xi \in [\lambda_{min},\lambda_{max}]$.

\end{lemma}

\begin{proof}
First we verify that $F_1(\mathbf{\Pi})$ is a convex function. One of the sufficient and necessary condition for a function whose variable is matrix  is convex is that the Hessian matrix of this function is positive semi-definite. The Hessian matrix of $F(\mathbf{\Pi})$ can be obtained by taking the second derivative over $\mathbf{\Pi}$ on $F(\mathbf{\Pi})$, we denote it as $\nabla^2F(\mathbf{\Pi})$. Therefore we can obtain the Hessian matrix of $F_1(\mathbf{\Pi})$ by
\begin{equation}
\nabla^2F_1(\mathbf{\Pi})=\nabla^2F_0(\mathbf{\Pi})-\lambda_{min}\mathbf{I}
.\nonumber
\end{equation}
where $\mathbf{I}$ is the identity matrix\footnote{The identity matrix $I$ means all the elements on the diagonal of $I$ are all $1$s while others are all $0$s. Note that here $I$ is an $n^2 \times n^2$ matrix since the first order derivative of a function whose variable is a matrix is a $n\times n$ matrix, thus the second derivative of $F_0$ ($F_1$) is $n^2 \times n^2$ matrix.}. Note that $\lambda_{min}$ is the minimum eigenvalue of $\nabla^2F_0(\mathbf{\Pi})$, therefore all the eigenvalues of $\nabla^2F_0(\mathbf{\Pi})-\lambda_{min}\mathbf{I}$ are equal to or larger than $0$. Hence $\nabla^2F_1(\mathbf{\Pi})$ is a nonnegative definite matrix and $F_1(\mathbf{\Pi})$ is a convex function.

Meanwhile, one of the sufficient and necessary conditions for a function whose variable is matrix  is concave is that the Hessian matrix of this function is negative semi-definite. Similar to the analysis of $F_1(\mathbf{\Pi})$, we can verify that $F_2(\mathbf{\Pi})$ is a concave function. Thus we complete the proof.
\end{proof}

Lemma \ref{lemma3} presents a simple way to implement CCOM algorithmically, since $F_0(\mathbf{\Pi})$ is just our objective function in Section \ref{AlgorithmB1} and $||\mathbf{\Pi}||_F^2$ can be computed easily. We can modify $F(\mathbf{\Pi})$ step by step from a convex function to a concave function by modifying the value of $\xi$ or $\alpha$. In the following analysis, we set $F_{\xi}(\mathbf{\Pi})$ equivalent to $F(\mathbf{\Pi})$ since $\xi$ is an adjustable parameter in $F(\mathbf{\Pi})$.

A vivid example of the CCOM under the formulation of $F_{\xi}(\mathbf{\Pi})$ by Lemma \ref{lemma3} is illustrated in Fig. \ref{Fig:convexconcave}. As can be seen in the figure, when $\xi$ starts at $\lambda_{min}$, $F_{\xi}(\mathbf{\Pi})$ is a convex function, thus we can obtain the minimizer of this objective function. After we find the minimizer, we modify $\alpha$ to be $0.2$, thus $\xi=0.8\lambda_{min}+0.2\lambda_{max}$, which makes the objective function become less convex. To obtain the minimizer of this new objective function, we have the prior knowledge of the previous minimizer, and since we only slightly modify the objective function, the optimal solution of new objective function should not deviate much from the previous one intuitively. Therefore we can start from the previous minimizer to find the new minimizer. Gradually, as $\alpha$ becomes increasingly larger, the objective function tends to be concave while the minimizer of it tends to get close to the boundary,  on which the optimal solution of the original WEMP exists. The trail for the minimizer can be referred to the red line with arrows in Fig. \ref{Fig:convexconcave}.

\begin{figure}[htbp]
     	\centering
		\includegraphics[width=0.46\textwidth]{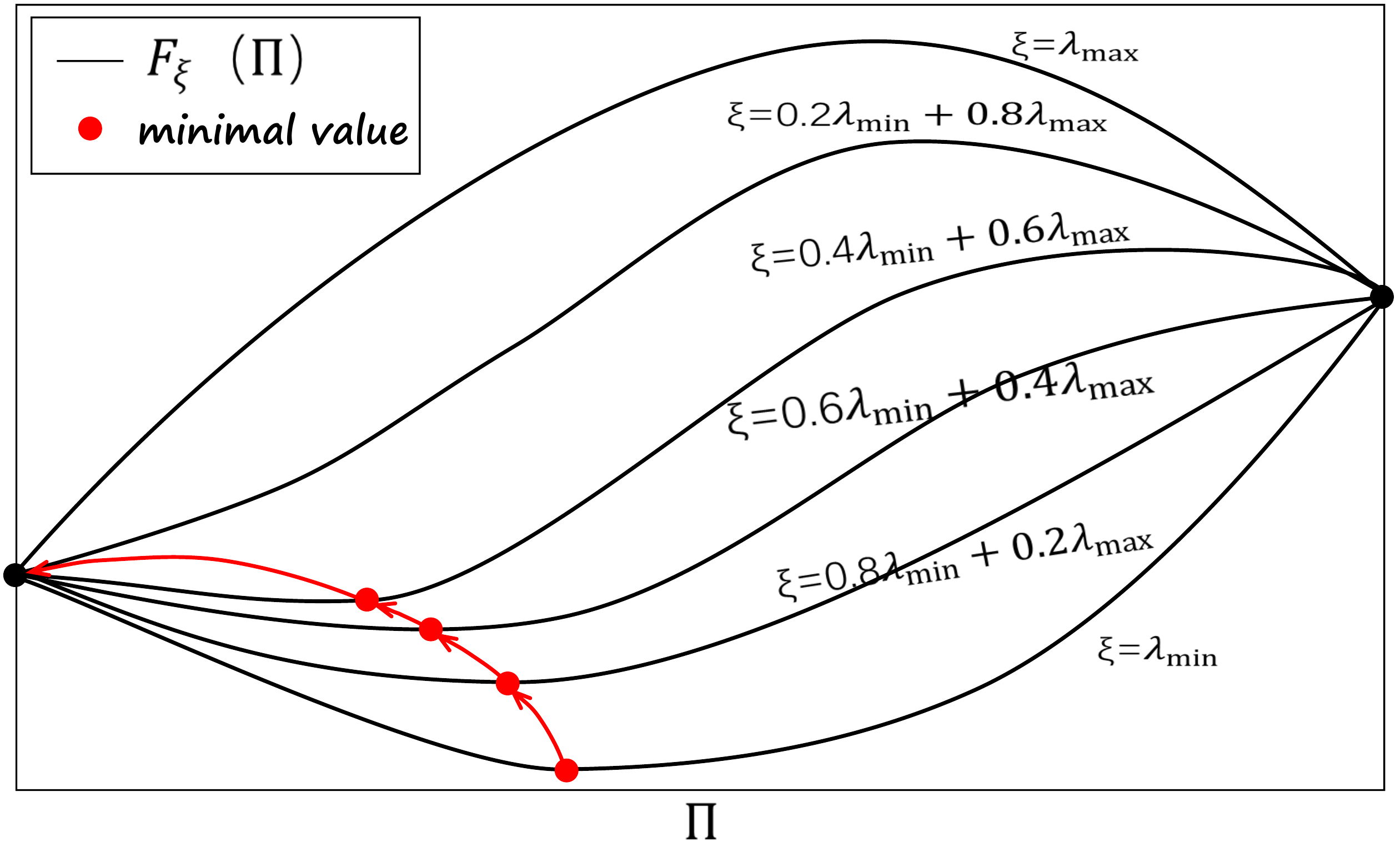}
		\caption{An Illustration of the Implementation of CCOM by Lemma \ref{lemma3}.}
		\label{Fig:convexconcave}
\end{figure}

Based on the above analysis, we propose Algorithm \ref{alg:1} as our main algorithm for the weighted-edge matching problem (WEMP) under CCOM. We call Algorithm \ref{alg:1} \emph{Convex-concave Based De-anonymization Algorithm (CBDA)}. Note that $F_0(\mathbf{\Pi})$ itself is convex in our problem, thus we can set $\xi$ from $0$ to an arbitrarily large number, which obviates the great complexity to calculate eigenvalues of Hessian matrices.

CBDA consists of an outer loop (lines $3$ to $10$) and an inner loop (lines $4$ to $8$). The outer loop modifies $\xi$ in CCOM. The inner loop finds the minimizer of $F(\mathbf{\Pi})$, whose main idea resembles descending algorithms:
In line $5$, we obtain descending direction by minimizing $\textbf{tr}(\nabla_{\mathbf{\Pi_k}} F(\mathbf{\Pi_k})^T\mathbf{X^{\bot}})$, dangling the highest probability to find a descending direction characterized by $\textbf{tr}(\nabla_{\mathbf{\Pi_k}} F(\mathbf{\Pi_k})^T\mathbf{X^{\bot}})<0$.
In line $6$ we search for step length $\gamma_k$ contributing most to lowering $F(\mathbf{\Pi})$ on this descending direction. Line $7$ is the update of estimation.

\begin{algorithm}[h]
\caption{Convex-concave Based De-anonymization Algorithm (CBDA)}
\begin{algorithmic}[1]
\REQUIRE ~Adjacent matrices $\mathbf{A}$ and $\mathbf{B}$;
Community assignment matrix $\mathbf{M}$; \\
\quad Weight controlling parameter $\mu$;
Adjustable parameters $\delta$, $\Delta\xi$. \\
\ENSURE  ~Estimated permutation matrix $\mathbf{\tilde{\Pi}}$.

\STATE Form the objective function $F_0(\mathbf{\Pi})$ and $F(\mathbf{\Pi})$.
\STATE $\xi\leftarrow0$, $k\leftarrow1$. Initialize $\mathbf{\Pi_1}$. Set $\xi_m$, the upper limit of $\xi$.
\WHILE {$\xi<\xi_{m}$ and $\mathbf{\Pi_k} \notin \Omega_0$}
\WHILE {$k = 1$ or $|F(\mathbf{\Pi_{k+1}})-F(\mathbf{\Pi_k})|\geq \delta$
}
\STATE $\mathbf{X^{\bot}}\leftarrow \arg\min_{\mathbf{X^{\bot}}}\textbf{tr}(\nabla_{\mathbf{\Pi_k}} F(\mathbf{\Pi_k})^T\mathbf{X^{\bot}})$, where $\mathbf{X^{\bot}}\in \Omega$.
\\
//Finding the optimal descent direction
\STATE $\gamma_k\leftarrow \arg\min_{\gamma}F(\mathbf{\Pi_k}+\gamma(\mathbf{X^{\bot}}-\mathbf{\Pi_k}))$, where $\gamma_k \in [0,1]$.
//Finding the optimal step size
\STATE $\mathbf{\Pi_{k+1}}\leftarrow \mathbf{\Pi_{k}}+\gamma_k(\mathbf{X^{\bot}}-\mathbf{\Pi_k})$, $k\leftarrow k+1$.
//Estimation Update
\ENDWHILE
\STATE $\xi \leftarrow \xi+\Delta\xi$.
\vspace{-1mm}
\ENDWHILE
\vspace{-1mm}
\end{algorithmic}
\label{alg:1}
\end{algorithm}

\subsubsection{Time Complexity and Convergence Analysis}

\

\textbf{Time Complexity:} The inner loop is similar to the Frank-Wolfe algorithm, with $O(n^6)$ in a round (since the input is an $n \times n$ matrix). If the maximum number of inner loops as $T$,
thus the whole algorithm has a complexity of $O\big(\frac{n^6T\xi}{\Delta\xi}\big)$
. As far as we know, a dearth of algorithmic analysis of seedless de-anonymization exists except for \cite{Fu:arxiv,Fu:GC}, with their proposed algorithm sharing identical complexity of $O(n^6)$ with ours.

\textbf{Convergence:} There are two loops in CBDA and we provide convergence analysis on them respectively. Before that, we first clarify that:

\begin{itemize}
\item We set
    $\mathbf{\mathbf{\Pi_{k}}}$ as the estimation after $k$ rounds in the inner loop, thus $\mathbf{\mathbf{\Pi_{k+1}}}$ is the estimation after $k+1$ rounds in the inner loop and $\mathbf{\mathbf{\Pi_{k+1}}}=\mathbf{\Pi_{k}}+\gamma_k(\mathbf{X_{\bot}}-\mathbf{\Pi_k})$.
\item We set $F_{\xi}(\mathbf{\Pi})=F_0(\mathbf{\Pi})+\xi(n-||\mathbf{\Pi}||_F^2)$ and $\mathbf{\Pi^{\xi}}$ as the minimizer of $F_{\xi}(\mathbf{\Pi})$.
    Thus $F_{\xi+\Delta \xi}(\mathbf{\Pi})=F_0(\mathbf{\Pi})+(\xi+\Delta\xi)(n-||\mathbf{\Pi}||_F^2)$
    and $\mathbf{\Pi^{\xi+\Delta \xi}}$ is the minimizer of $F_{\xi+\Delta \xi}(\mathbf{\Pi})$.

\end{itemize}

Then we propose Lemma \ref{lemma7} to discuss the convergence of CBDA.

\begin{lemma}
\label{lemma7}
CBDA converges and the final output is a permutation matrix in the original feasible region $\Omega_0$.
\end{lemma}

\begin{proof}
As stated above, showing the convergence of CBDA is equivalent to showing the convergence of both inner and outer loops.

\textbf{1. Inner Loop:}
We focus on $F_{\xi}(\mathbf{\Pi_{k+1}})$ and $F_{\xi}(\mathbf{\Pi_{k+1}})$. Since $\mathbf{\Pi_{k+1}}=\mathbf{\Pi_{k}}+\gamma_k(\mathbf{X_{\bot}}-\mathbf{\Pi_k})$, according to Taylor's Theorem,
\begin{equation}
\label{e65}
\begin{aligned}
F_{\xi}(\mathbf{\Pi_{k+1}})&=F_{\xi}(\mathbf{\Pi_{k}}+\gamma_k(\mathbf{X_{\bot}}-\mathbf{\Pi_k}))
\\&
=F_{\xi}(\mathbf{\Pi_{k}})+\gamma_k \textbf{tr}(\nabla F_{\xi}^T(\mathbf{\Pi_k})(\mathbf{X^{\bot}}-\mathbf{\Pi_k}))+\gamma_k\mathbf{R_k}
\\&\leq F_{\xi}(\mathbf{\Pi_{k}})+\gamma_k \textbf{tr}(\nabla F_{\xi}^T(\mathbf{\Pi_k})(\mathbf{\Pi^{\xi}}-\mathbf{\Pi_k}))+\gamma_k\mathbf{R_k},
\end{aligned}
\end{equation}
where $\gamma_k\mathbf{R_k}$ is the remainder of this Taylor series, and this form makes sense since the remainder must contain a multiplicative factor of $\gamma_k$. The last inequality holds since $\mathbf{X}^{\bot}$ is the minimizer of $\textbf{tr}(\nabla F_{\xi}^T(\mathbf{\Pi_k})(\mathbf{\Pi^{\xi}}-\mathbf{\Pi_k}))$.

In terms of $F_{\xi}(\mathbf{\Pi^{\xi}})$, we have
\begin{equation}
\label{e66}
\begin{aligned}
&F_{\xi}(\mathbf{\Pi^{\xi}})=F_{\xi}(\mathbf{\Pi_k}+\mathbf{\Pi^{\xi}}-\mathbf{\Pi_k})
\\&=F_{\xi}(\mathbf{\Pi_k})+\textbf{tr}(\nabla F_{\xi}^T(\mathbf{\Pi_k})(\mathbf{\Pi^{\xi}}-\mathbf{\Pi_k}))+\mathbf{R_k^{\prime}},
\end{aligned}
\end{equation}
where $\mathbf{R_k^{\prime}}$ is the remainder of this Taylor series.

Combining Eqn. (\ref{e65}) and (\ref{e66}), we can obtain
\begin{equation}
\label{i67}
F_{\xi}(\mathbf{\Pi_{k+1}})\leq
F_{\xi}(\mathbf{\Pi_{k}})+\gamma_k (F_{\xi}(\mathbf{\Pi^{\xi}})-F_{\xi}(\mathbf{\Pi_{k}}))+\gamma_k(\mathbf{R_k}-\mathbf{R_k^{\prime}}).
\end{equation}

Denote $\Delta \mathbf{R_k}= \mathbf{R_k-R_k^{\prime}}$ and by simple transformation of Inequality (\ref{i67}), we obtain
\begin{equation}
\label{i68}
F_{\xi}(\mathbf{\Pi_{k+1}})-F_{\xi}(\mathbf{\Pi^{\xi}})\leq
(1-\gamma_k)(F_{\xi}(\mathbf{\Pi_{k}})-F_{\xi}(\mathbf{\Pi^{\xi}}))+\gamma_k\Delta \mathbf{R_k}.
\end{equation}

Note that Inequality (\ref{i68}) builds up the relationship between $F_{\xi}(\mathbf{\Pi_{k+1}})$ and $F_{\xi}(\mathbf{\Pi_{k}})$, and we obtain
\begin{equation}
\label{i69}
\begin{aligned}
&F_{\xi}(\mathbf{\Pi_{k+1}})-F_{\xi}(\mathbf{\Pi^{\xi}}) \\&\leq \prod_{i=1}^k (1-\gamma_i) (F_{\xi}(\mathbf{\Pi_{1}})-F_{\xi}(\mathbf{\Pi^{\xi}}))
+
\sum_{i=1}^k \gamma_i \prod_{j=1}^{k-i} (1-\gamma_j)\Delta\mathbf{R}_i.
\end{aligned}
\end{equation}

For $F_{\xi}(\mathbf{\Pi_{1}})-F_{\xi}(\mathbf{\Pi^{\xi}})$, note that $\mathbf{\Pi_{1}}=\mathbf{\Pi^{\xi-\Delta\xi}}$, then
\begin{equation}
\label{i70}
\begin{aligned}
F_{\xi}(\mathbf{\Pi^{\xi}})&=F_0(\mathbf{\Pi^{\xi}})+\xi(n-||\mathbf{\Pi^{\xi}}||_F^2)
\\&=F_0(\mathbf{\Pi^{\xi}})+(\xi-\Delta\xi)(n-||\mathbf{\Pi^{\xi}}||_F^2)-\Delta\xi(n-||\mathbf{\Pi^{\xi}}||_F^2)
\\&\geq F_0(\mathbf{\Pi^{\xi-\Delta\xi}})+(\xi-\Delta\xi)(n-||\mathbf{\Pi^{\xi-\Delta\xi}}||_F^2)\\&\quad-\Delta\xi(n-||\mathbf{\Pi^{\xi}}||_F^2)
\\&=F_0(\mathbf{\Pi^{\xi-\Delta\xi}})+\xi(n-||\mathbf{\Pi^{\xi-\Delta\xi}}||_F^2)
\\& \quad +
\Delta\xi(||\mathbf{\Pi^{\xi}}||_F^2-||\mathbf{\Pi^{\xi-\Delta\xi}}||_F^2)
\\&=F_{\xi}(\mathbf{\Pi^{\xi-\Delta\xi}})+\Delta\xi(||\mathbf{\Pi^{\xi}}||_F^2-||\mathbf{\Pi^{\xi-\Delta\xi}}||_F^2)
.
\end{aligned}
\end{equation}

Hence
\begin{equation}
\label{e71}
F_{\xi}(\mathbf{\Pi^{\xi-\Delta\xi}})-F_{\xi}(\mathbf{\Pi^{\xi}})\leq\Delta\xi(||\mathbf{\Pi^{\xi-\Delta\xi}}||_F^2-||\mathbf{\Pi^{\xi}}||_F^2)
.
\end{equation}

Therefore by combining Inequalities (\ref{e71}) and (\ref{i69}), we can obtain if $\Delta\xi$ is small enough, or if $k\rightarrow \infty$, then the term $\prod_{i=1}^k (1-\gamma_i) (F_{\xi}(\mathbf{\Pi_{1}})-F_{\xi}(\mathbf{\Pi^{\xi}}))$ in last expression of Inequality (\ref{i69}) goes to $0$.

For the second term $\sum_{i=1}^k \gamma_i \prod_{j=1}^{k-i} (1-\gamma_j)\Delta\mathbf{R}_i$, we note that when $k\rightarrow \infty$, then $\forall \epsilon>0, \exists K>0, \delta_1>0$, when $i>K$, $\gamma_i \prod_{j=1}^{k-i} (1-\gamma_j)< \gamma_i < \frac{\epsilon}{2^{\delta_1 i}}$, and meanwhile when $i\leq K$, $\gamma_k (1-\gamma_j)<\prod_{j=1}^{k-i} (1-\gamma_j)<\frac{\epsilon}{2^{\delta_2 i}}$. Setting $\delta^*=\min\{\delta_1, \delta_2\}$, then we can upper bound the sum $\sum_{i=1}^k \gamma_i \prod_{j=1}^{k-i} (1-\gamma_j)\Delta\mathbf{R}_i \leq \sum_{i=1}^{\infty} \frac{\epsilon}{2^{\delta i}}=0$. Therefore we prove that the inner loop converges.

\textbf{2. Outer Loop:}
Note that from Eqn. (\ref{e71}), we know $(||\mathbf{\Pi^{\xi-\Delta\xi}}||_F^2-||\mathbf{\Pi^{\xi}}||_F^2)$ is nonnegative since $\Delta\xi>0$ and $\mathbf{\Pi^{\xi}}$ is the minimizer of $F_{\xi}(\mathbf{\Pi})$. Thus $||\mathbf{\Pi^{\xi}}||_F^2 \leq ||\mathbf{\Pi^{\xi-\Delta\xi}}||_F^2
$. Note that for all the $\mathbf{\Pi}\in \Omega$, the maximum value of $||\mathbf{\Pi}||_F^2$ is $n$, and the maximizer is in $\Omega_0$. Therefore $||\mathbf{\Pi}||_F^2-n\leq 0$. From Inequality (\ref{i70}), we find that
\begin{equation}
\small
\begin{aligned}
&F_{\xi}(\mathbf{\Pi^{\xi}})\\&\geq F_0(\mathbf{\Pi^{\xi-\Delta\xi}})+(\xi-\Delta\xi)(n-||\mathbf{\Pi^{\xi-\Delta\xi}}||_F^2)-\Delta\xi(n-||\mathbf{\Pi^{\xi}}||_F^2)\\&=F_{\xi-\Delta\xi}(\mathbf{\Pi^{\xi-\Delta\xi}})-\Delta \xi \textbf{tr}(||\mathbf{\Pi^{\xi}}||_F^2-n).\nonumber
\end{aligned}
\end{equation}

Therefore
\begin{equation}
\begin{aligned}
&|F_{\xi}(\mathbf{\Pi^{\xi}})-F_{\xi-\Delta\xi}(\mathbf{\Pi^{\xi-\Delta\xi}})|\\&\leq \Delta \xi|||\mathbf{(\Pi^{\xi})}||_F^2-n|
\leq \Delta \xi |||\mathbf{\Pi^{\xi-\Delta\xi}}||_F^2-n|
\\&\leq \Delta \xi |||\mathbf{\Pi^{\xi_0}}||_F^2-n|
\leq \Delta \xi (n-1),\nonumber
\end{aligned}
\end{equation}
where the third inequality holds since $\mathbf{\Pi^{\xi_0}}$ is the minimizer of $F_{\lambda_{min}}(\mathbf{\Pi})$, i.e., the convex relaxation of $F_0(\mathbf{\Pi})$, and the fourth inequality holds since $\min_{\mathbf{\Pi}\in \Omega} ||\mathbf{\Pi}||_F^2=1$ and $\mathbf{\Pi}=\textbf{1}_{n\times n}./n$ is the minimizer. Therefore, the analysis tells us if $\Delta\xi=o\big(\frac{1}{n}\big)$, then we can ensure that the outer loop converges.

Combining the convergence analysis of both inner and outer loops above, we complete the proof of the convergence of CBDA.
\end{proof}

Lemma \ref{lemma7} shows that CBDA can exactly find $\mathbf{\tilde{\Pi}}$, the minimizer of the objective function $F_0(\mathbf{\Pi})$, meanwhile ensuring that CBDA can perfectly solve WEMP, which vanishes the relative NME under mild conditions (Recall Theorem \ref{th3}). Therefore CBDA is an algorithmic approach for seedless de-anonymization with high feasibility and good performance, especially for networks with larger size.

\begin{figure*}[!tb]
	\centering
	\subfigure[N=500, a=3]{
		\begin{minipage}[]{0.235\linewidth}
			\centering
			\vspace{-3mm}
			\includegraphics[width=1.0\linewidth]{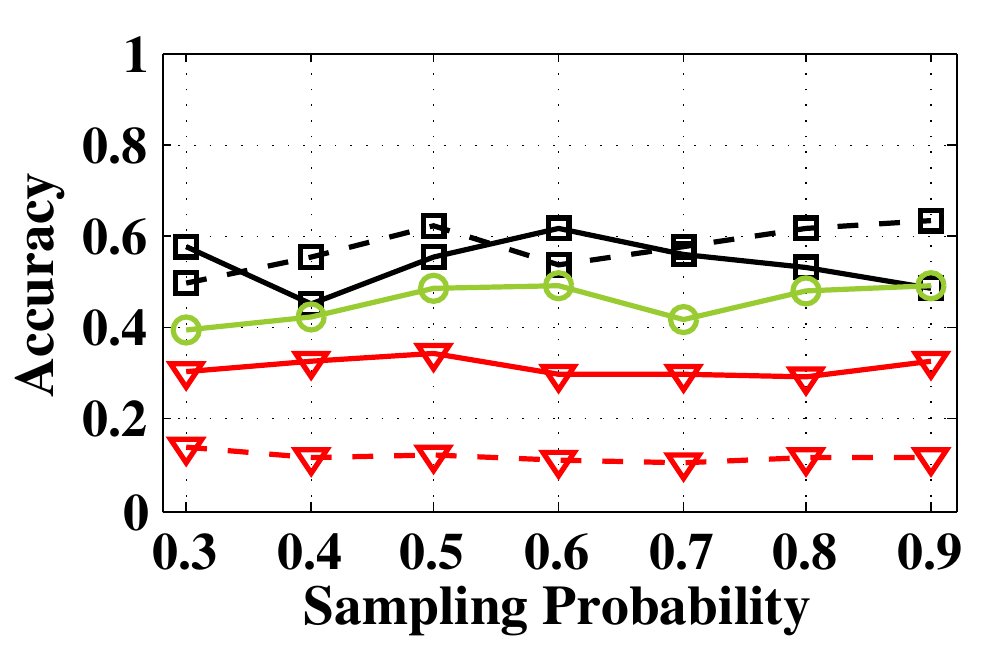}
			\vspace{-3mm}
		\end{minipage}%
		
	}
	\subfigure[N=1000, a=3]{
		\begin{minipage}[]{0.235\linewidth}
			\centering
			\vspace{-3mm}
			\includegraphics[width=1.0\linewidth]{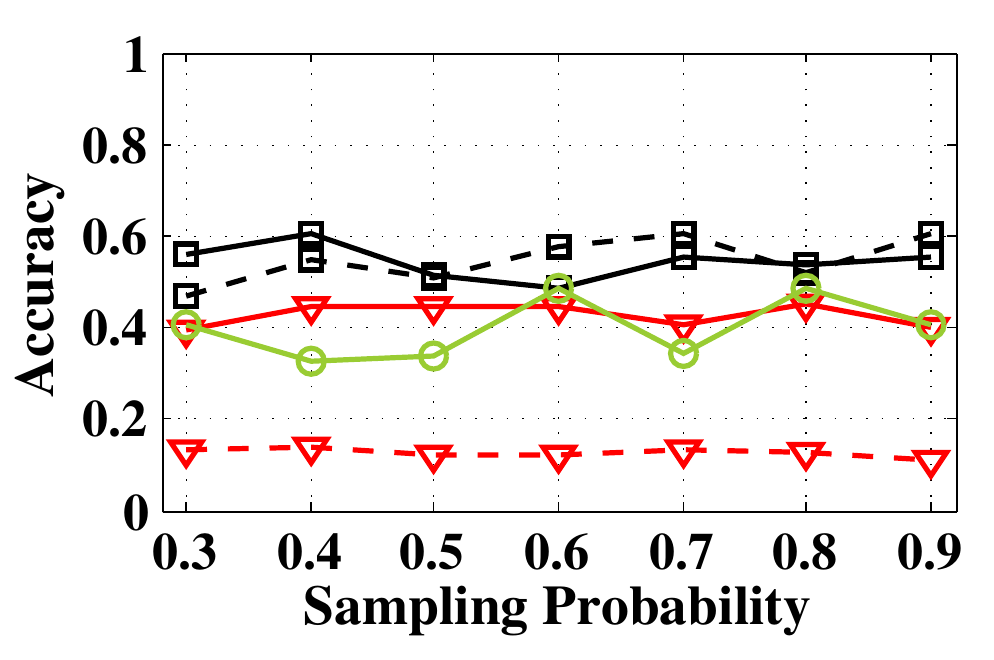}
			\vspace{-3mm}
		\end{minipage}%
		
	}
	\subfigure[N=1500, a=3]{
		\begin{minipage}[]{0.235\linewidth}
			\centering
			\vspace{-3mm}
			\includegraphics[width=1.0\linewidth]{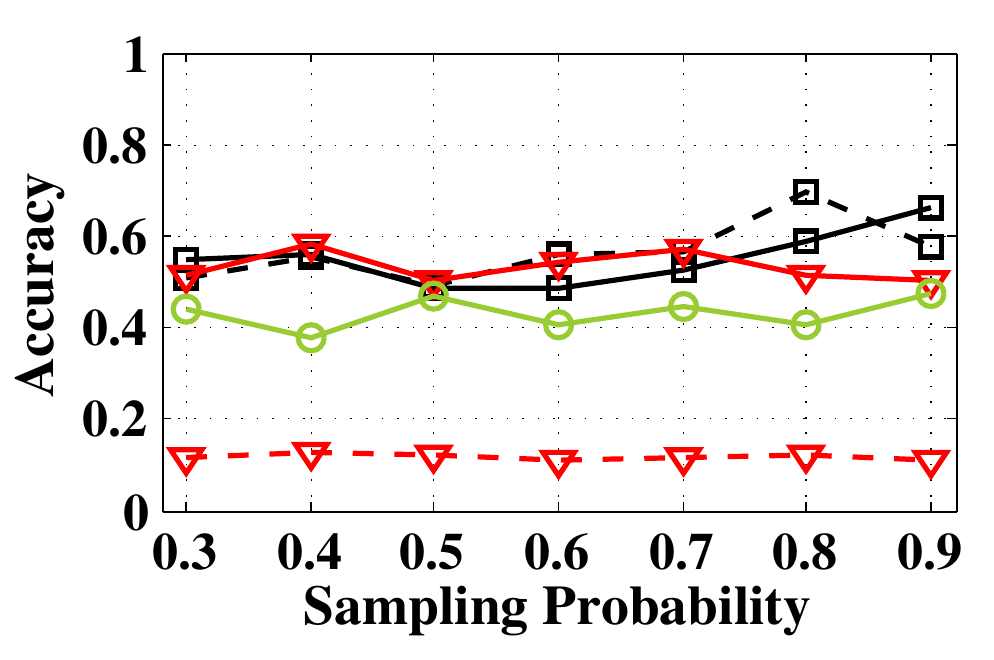}
			\vspace{-3mm}
		\end{minipage}%
		
	}
	\subfigure[N=2000, a=3]{
		\begin{minipage}[]{0.235\linewidth}
			\centering
			\vspace{-3mm}
			\includegraphics[width=1.0\linewidth]{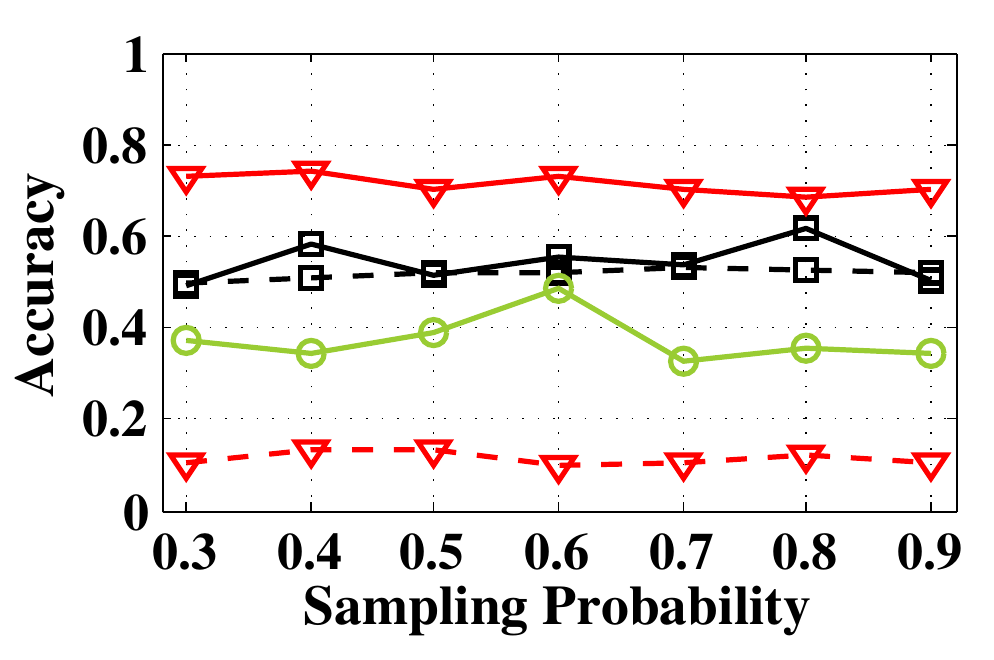}
			\vspace{-3mm}
		\end{minipage}%
		
	}

	\subfigure[N=500, a=5]{
		\begin{minipage}[]{0.235\linewidth}
			\centering
			\vspace{-3mm}
			\includegraphics[width=1.0\linewidth]{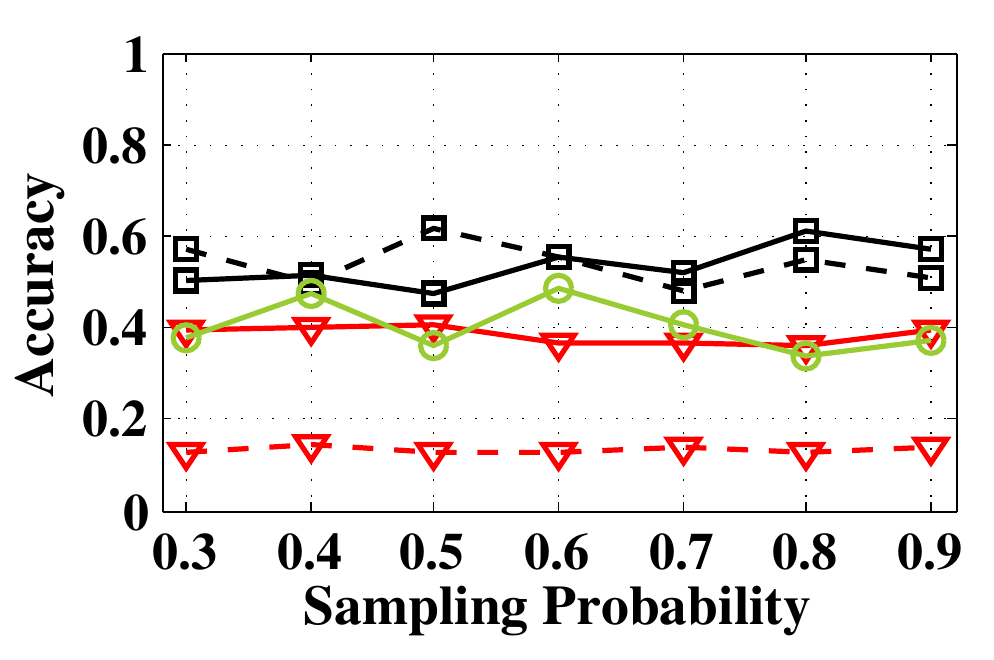}
			\vspace{-3mm}
		\end{minipage}%
		
	}
	\subfigure[N=1000, a=5]{
		\begin{minipage}[]{0.235\linewidth}
			\centering
			\vspace{-3mm}
			\includegraphics[width=1.0\linewidth]{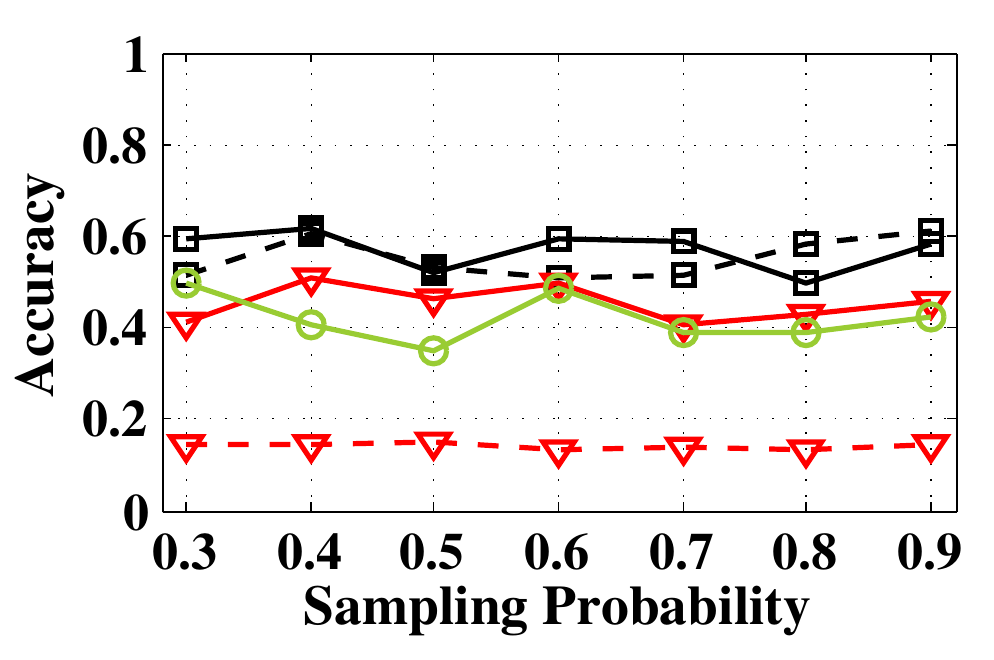}
			\vspace{-3mm}
		\end{minipage}%
		
	}
	\subfigure[N=1500, a=5]{
		\begin{minipage}[]{0.235\linewidth}
			\centering
			\vspace{-3mm}
			\includegraphics[width=1.0\linewidth]{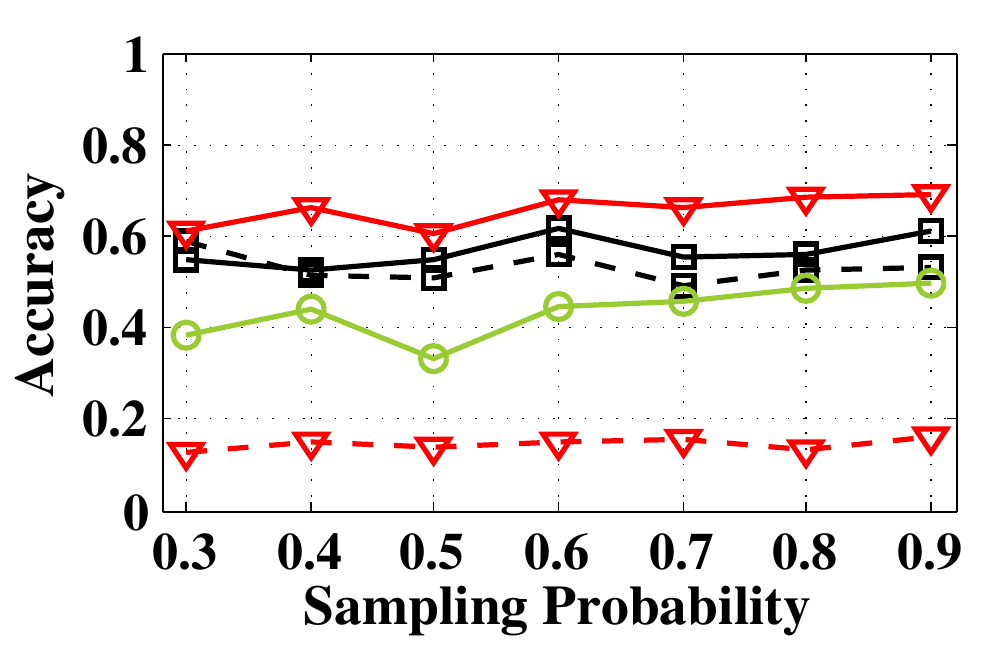}
			\vspace{-3mm}
		\end{minipage}%
		
	}
	\subfigure[N=2000, a=5]{
		\begin{minipage}[]{0.235\linewidth}
			\centering
			\vspace{-3mm}
			\includegraphics[width=1.0\linewidth]{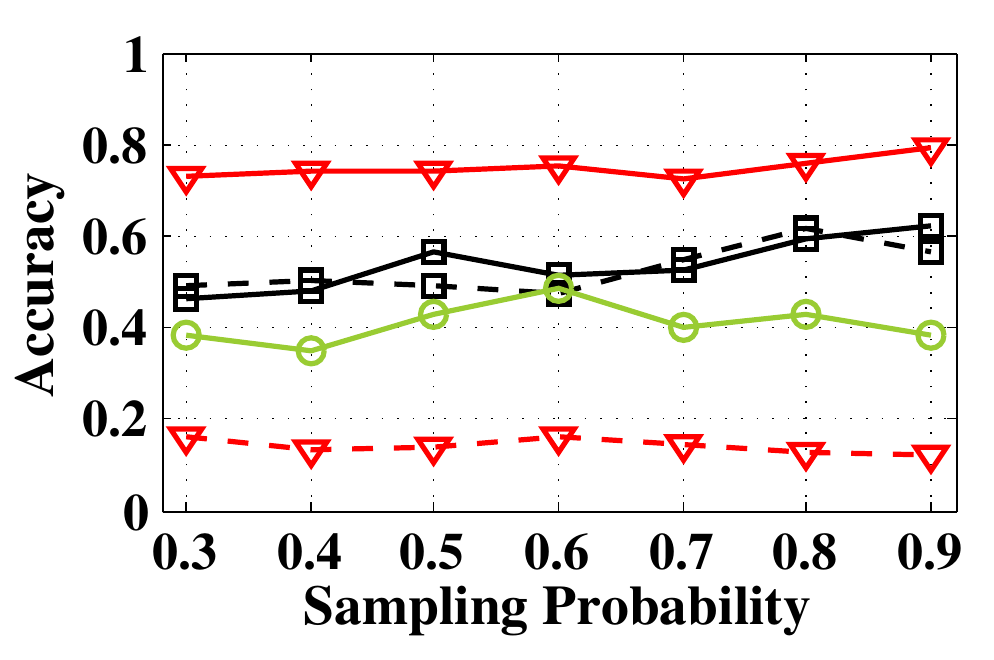}
			\vspace{-3mm}
		\end{minipage}%
		
	}

	\subfigure[N=500, a=7]{
		\begin{minipage}[]{0.235\linewidth}
			\centering
			\vspace{-3mm}
			\includegraphics[width=1.0\linewidth]{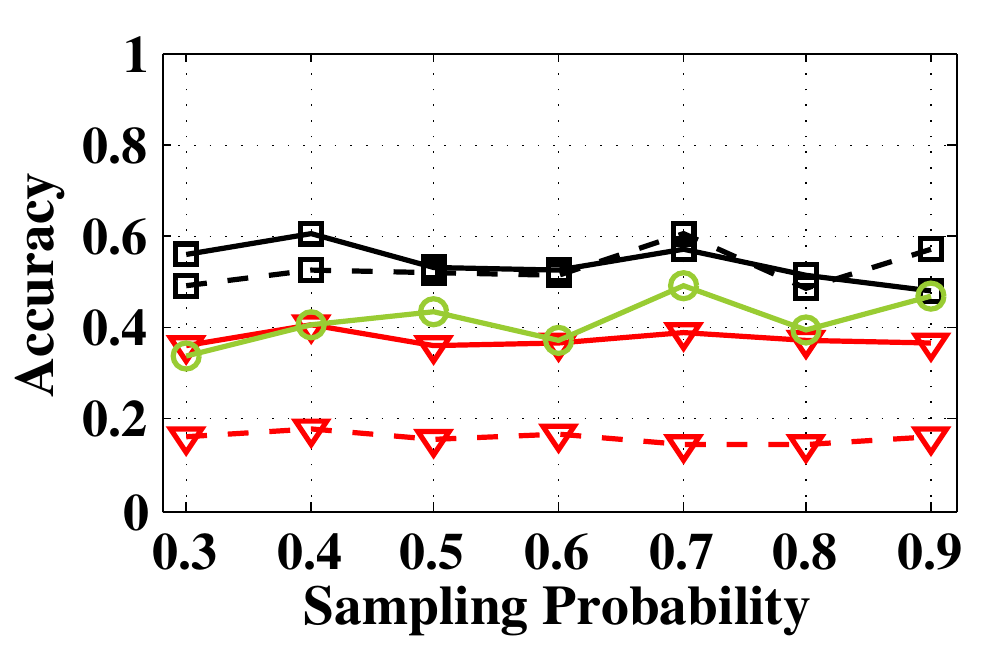}
			\vspace{-3mm}
		\end{minipage}%
		
	}
	\subfigure[N=1000, a=7]{
		\begin{minipage}[]{0.235\linewidth}
			\centering
			\vspace{-3mm}
			\includegraphics[width=1.0\linewidth]{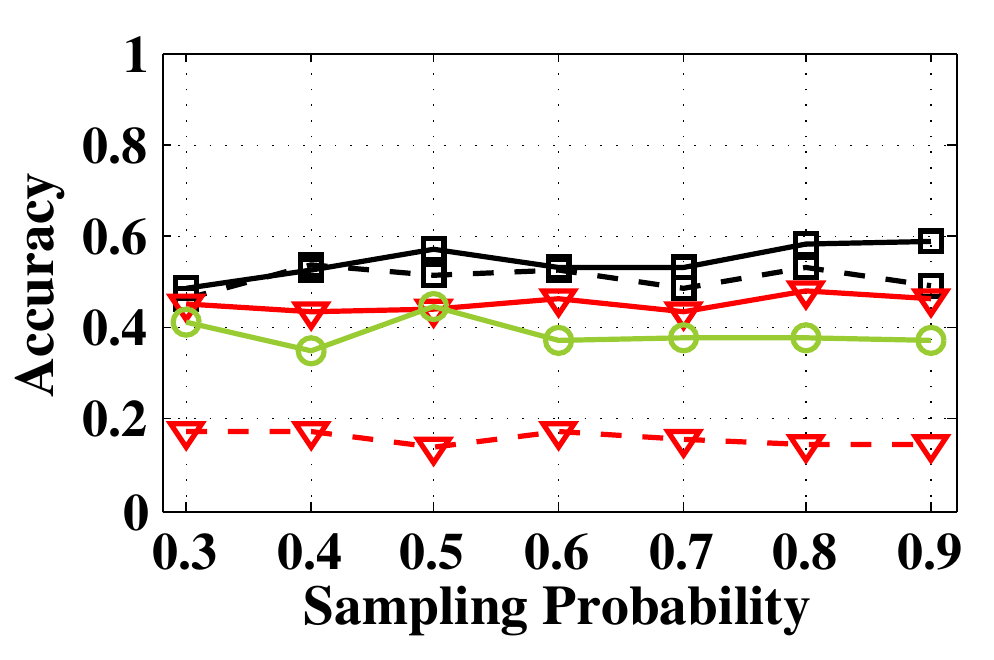}
			\vspace{-3mm}
		\end{minipage}%
		
	}
	\subfigure[N=1500, a=7]{
		\begin{minipage}[]{0.235\linewidth}
			\centering
			\vspace{-3mm}
			\includegraphics[width=1.0\linewidth]{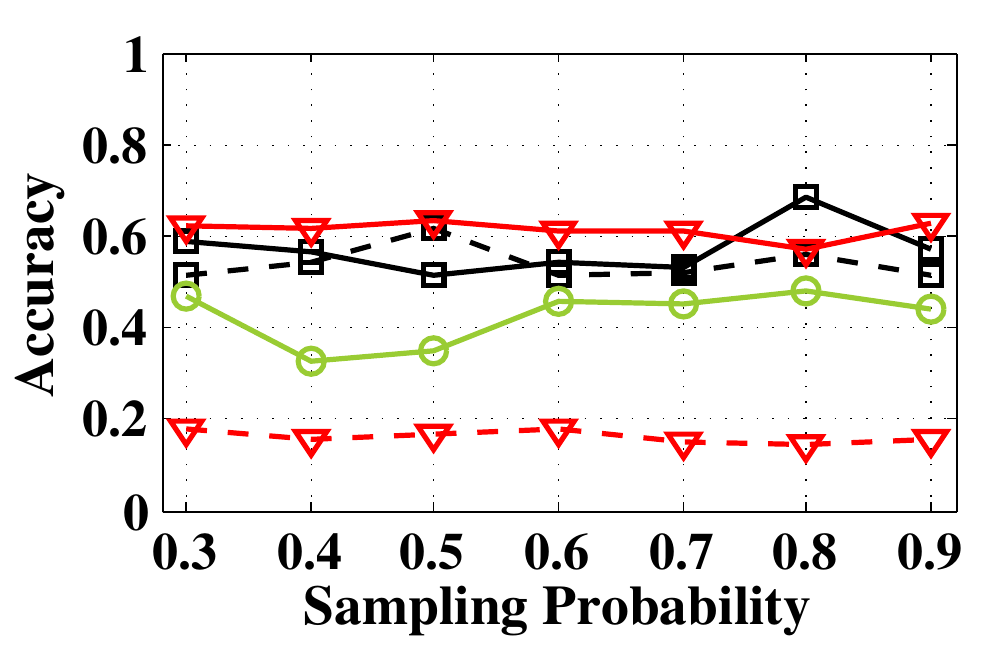}
			\vspace{-3mm}
		\end{minipage}%
		
	}
	\subfigure[N=2000, a=7]{
		\begin{minipage}[]{0.235\linewidth}
			\centering
			\vspace{-3mm}
			\includegraphics[width=1.0\linewidth]{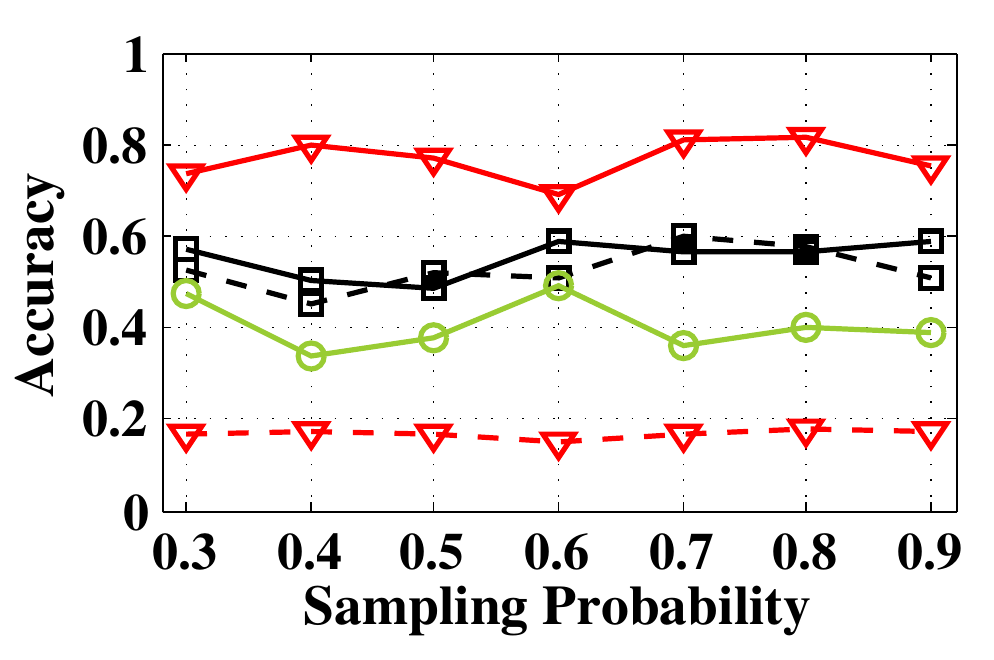}
			\vspace{-3mm}
		\end{minipage}%
		
	}
	\subfigure[N=500, a=9]{
		\begin{minipage}[]{0.235\linewidth}
			\centering
			\vspace{-3mm}
			\includegraphics[width=1.0\linewidth]{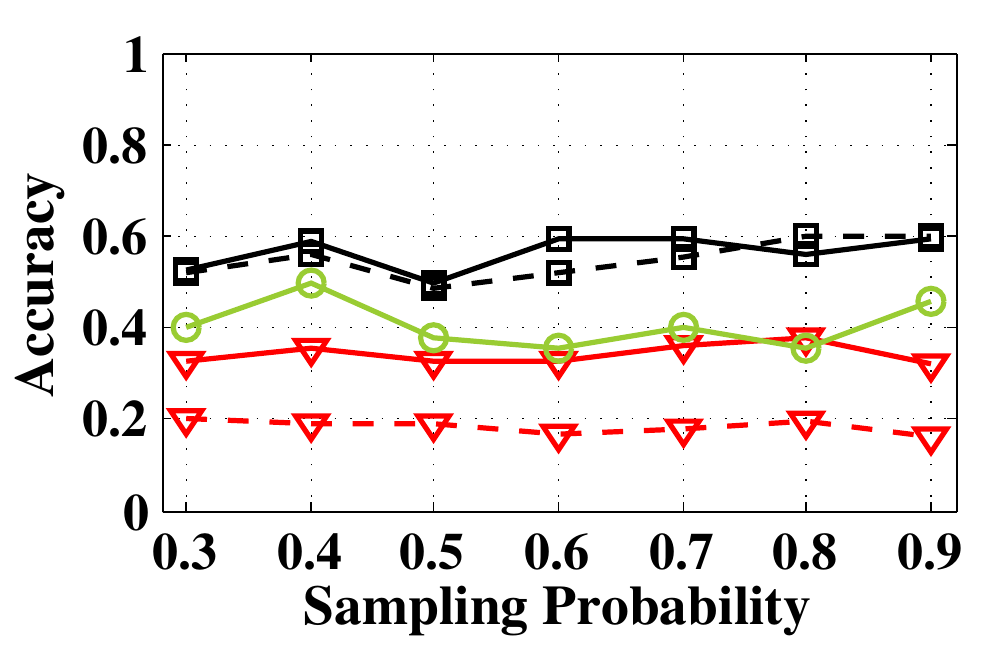}
			\vspace{-3mm}
		\end{minipage}%
		
	}
	\subfigure[N=1000, a=9]{
		\begin{minipage}[]{0.235\linewidth}
			\centering
			\vspace{-3mm}
			\includegraphics[width=1.0\linewidth]{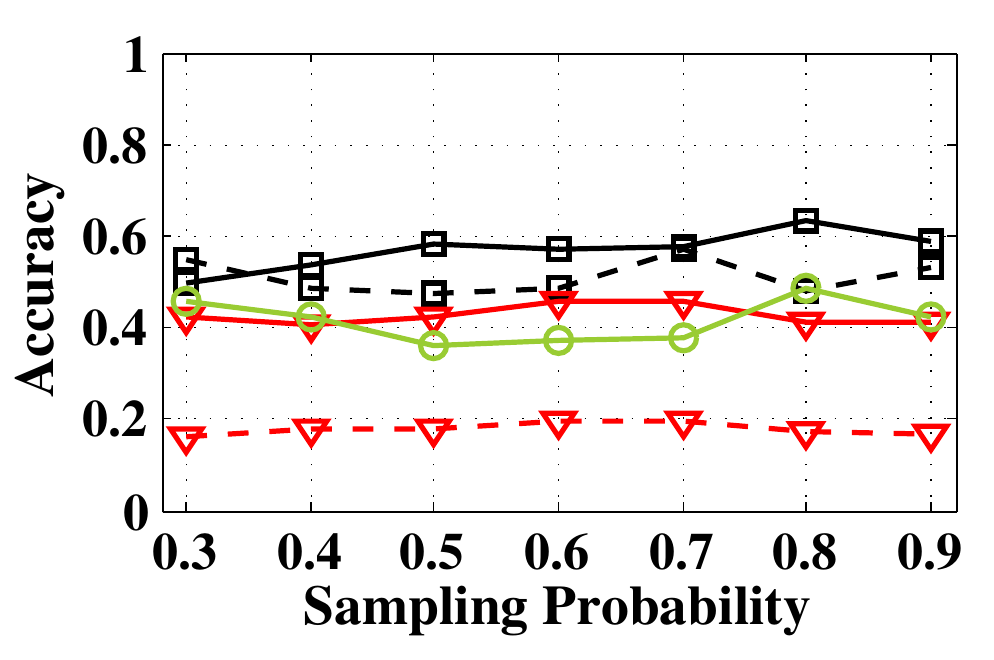}
			\vspace{-3mm}
		\end{minipage}%
		
	}
	\subfigure[N=1500, a=9]{
		\begin{minipage}[]{0.235\linewidth}
			\centering
			\vspace{-3mm}
			\includegraphics[width=1.0\linewidth]{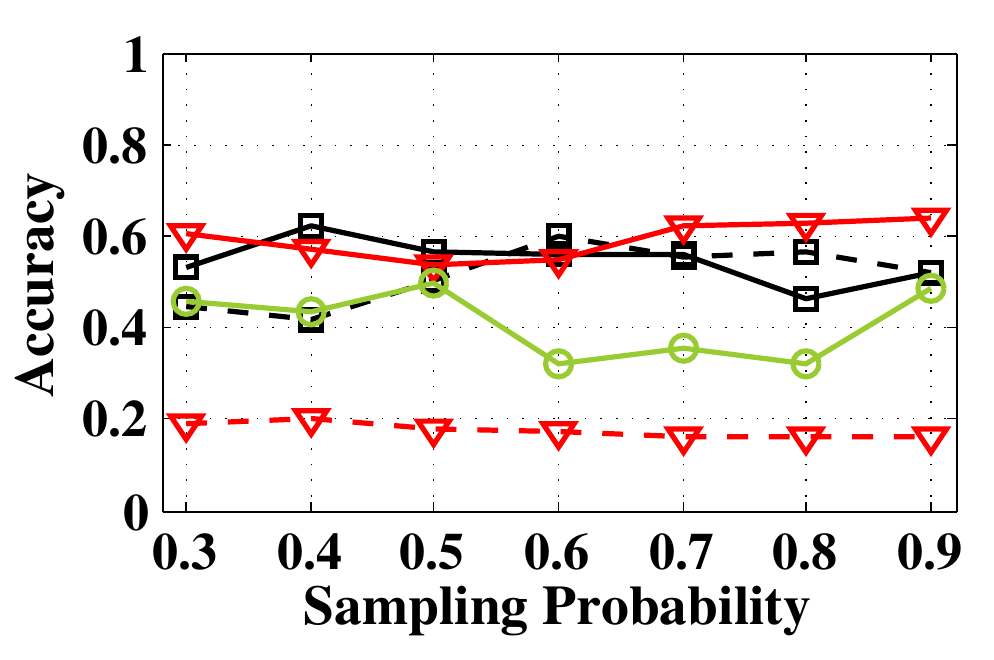}
			\vspace{-3mm}
		\end{minipage}%
		
	}
	\subfigure[N=2000, a=9]{
		\begin{minipage}[]{0.235\linewidth}
			\centering
			\vspace{-3mm}
			\includegraphics[width=1.0\linewidth]{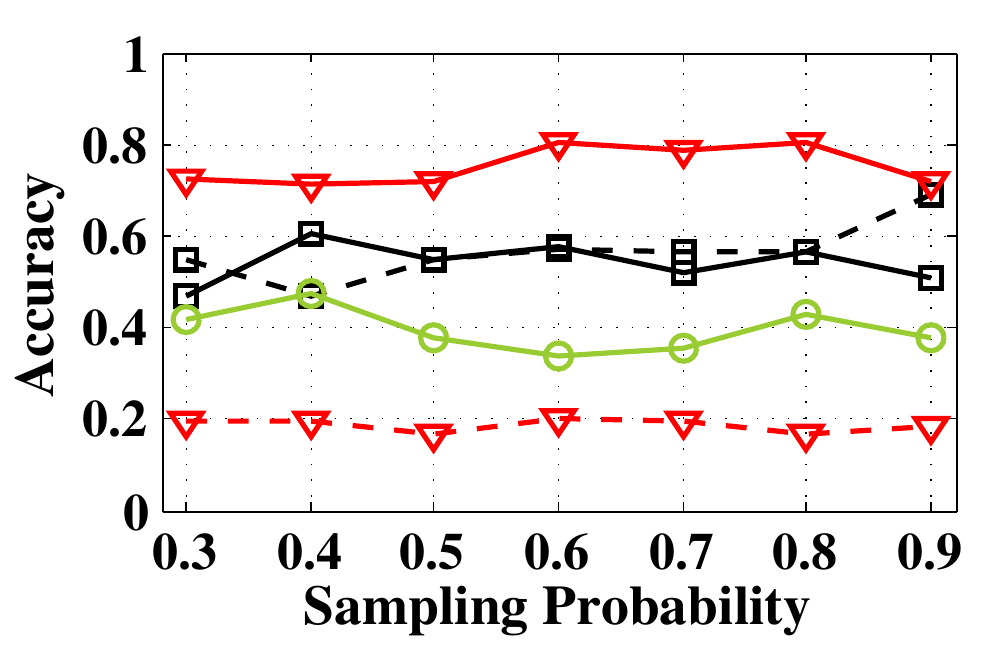}
			\vspace{-3mm}
		\end{minipage}%
	}

		\includegraphics[width=0.95\textwidth]{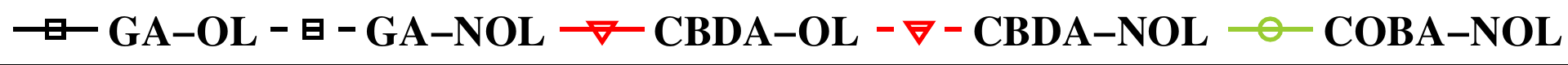}
		\caption{Experiments on Synthetic Networks with $\eta=0.05$.}
		\label{fig:syn5}
\end{figure*}

\section{Experimental Aspect of Social Network De-anonymization Problem}\label{Experimental}

In this section, we utilize three datasets: synthetic networks, sampled real social networks and true cross-domain networks, to conduct the experimental validation in terms of our analytical results and the performance of our proposed algorithm CBDA. Before we start, we need to clarify that our theoretical results are based on asymptotical analysis when the size of the network goes to infinity, thus it is hard to validate them under finite computability. However, we can also observe some expected phenomenons under networks with finite size. In our experiments, the number of nodes in cross-domain co-author networks is $3176$, larger than previous work in \cite{Fu:arxiv,Fu:GC} which is $2093$. The performance validation of algorithms for seedless de-anonymization on large-scale real social networks, adopted by studies on seeded de-anonymization, as far as we know, is still an open problem.

\subsection{Experiment Setup}

Before presenting our experimental results, we first introduce the basic experimental settings.

\subsubsection{Main Parameters} We list our adjustable parameters involved in our experiments in Table \ref{table:notation3}. Three parameters are in need of further explanations:

(i) ${a}$. This is a parameter in the overlapping stochastic block model (OSBM) which determines the $p_{\boldsymbol{C_iC_j}}$, the probability of edge existence between nodes $i$ and $j$ in underlying graph. Specifically, $p_{\boldsymbol{C_iC_j}}$ can be expressed as\footnote{Note that this expression is equivalent to that in \cite{latouche2011overlapping}, though their forms are different.}
\begin{equation}
\label{exp39}
p_{\boldsymbol{C_iC_j}}=\frac{1}{1+ae^{-x}},
\end{equation}
where $x$ is the number of communities that both nodes $i$ and $j$ belong to. Note that $p_{\boldsymbol{C_iC_j}}$ increases as $x$ rises, which corresponds to the real case that nodes with more overlapping communities are more possibly related. Meanwhile, if $a$ becomes larger (smaller), then $p_{\boldsymbol{C_iC_j}}$ is smaller (larger) so that the graph becomes sparser (denser).

(ii) ${\eta}$. This is the community ratio. It means the ratio between the number of communities and nodes. This ratio reflects the fact that when the size of network becomes larger, the number of communities also increases. In performance validation of CBDA we set $\eta=0.05$ or $0.1$, while when studying the influence of $\eta$ on de-anonymization accuracy, it will be endowed with more values.

(iii) $OL/NOL$. $OL$ means that communities are overlapping while $NOL$ means not. This makes for illustrating the impact of the overlapping property of communities on the mapping accuracy.

				\begin{table}[!tb]
					\setlength{\extrarowheight}{3pt}
					
						\renewcommand\arraystretch{0.82}
						\caption{\bf Main Experimental Parameters}
					
						\centering
						\label{table:notation3}
						\resizebox{1.0\columnwidth}{!}{

							\begin{tabular}{c|c|c}
                                 \textbf{Notation} & \textbf{Definition} & \textbf{Range} \\ \hline
                                $N$ & Number of Nodes & \{500, 1000, 1500, 2000\} \\
								$s$ & Sampling Probability ($s_1=s_2=s$) & 0.3-0.9 \\ 
								$a$ & OSBM Parameter & \{3, 5, 7, 9\} \\ 
								$\eta$ & Community Ratio & \{0.05, 0.1\} \\ 
                                $OL/NOL$ & Overlapping or Non-Overlapping & \{OL, NOL\} \\
                                \hline						
							\end{tabular}}		
							
					\end{table}

\begin{figure*}[htbp]
	\centering
	\subfigure[N=500, a=3]{
		\begin{minipage}[]{0.235\linewidth}
			\centering
			\vspace{-3mm}
			\includegraphics[width=1.0\linewidth]{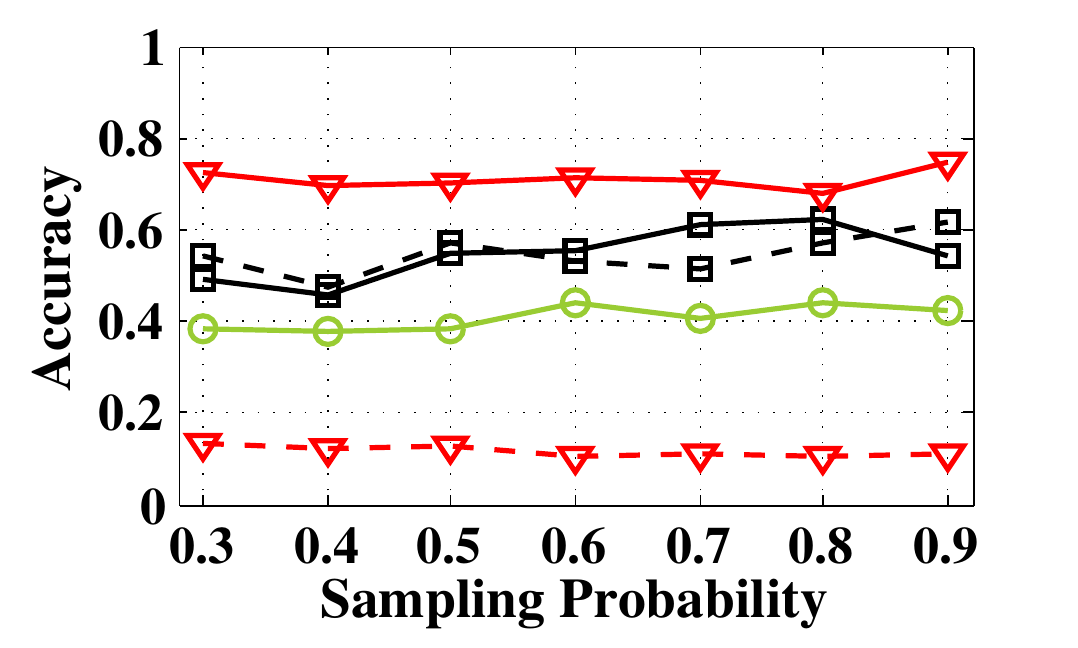}
			\vspace{-3mm}
		\end{minipage}%
	}
	\subfigure[N=1000, a=3]{
		\begin{minipage}[]{0.235\linewidth}
			\centering
			\vspace{-3mm}
			\includegraphics[width=1.0\linewidth]{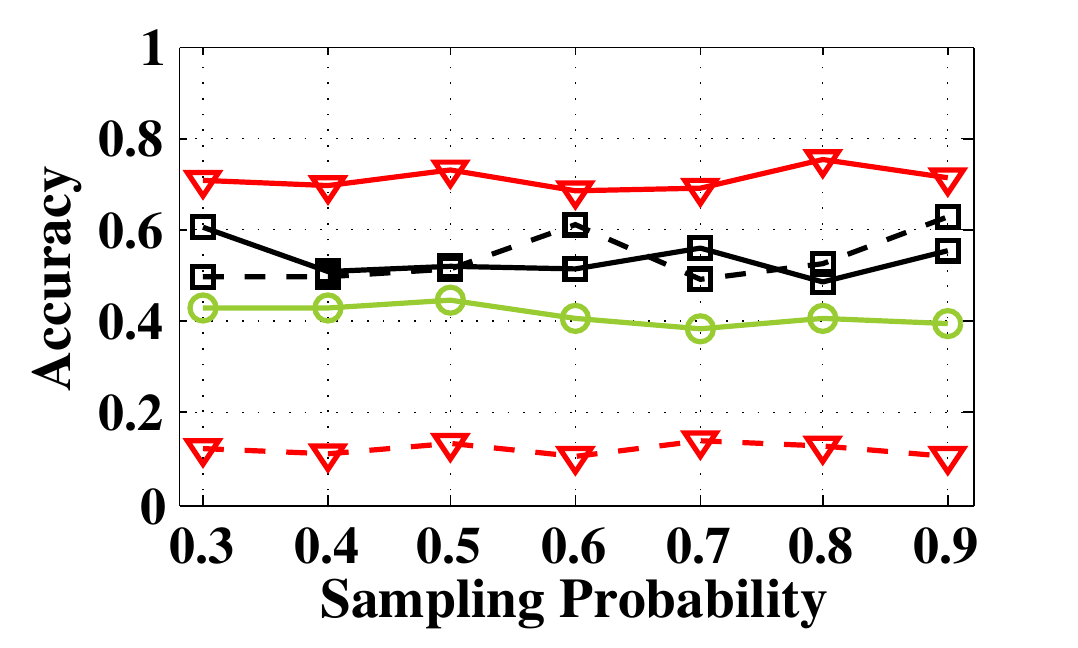}
			\vspace{-3mm}
		\end{minipage}%
		
	}
	\subfigure[N=1500, a=3]{
		\begin{minipage}[]{0.235\linewidth}
			\centering
			\vspace{-3mm}
			\includegraphics[width=1.0\linewidth]{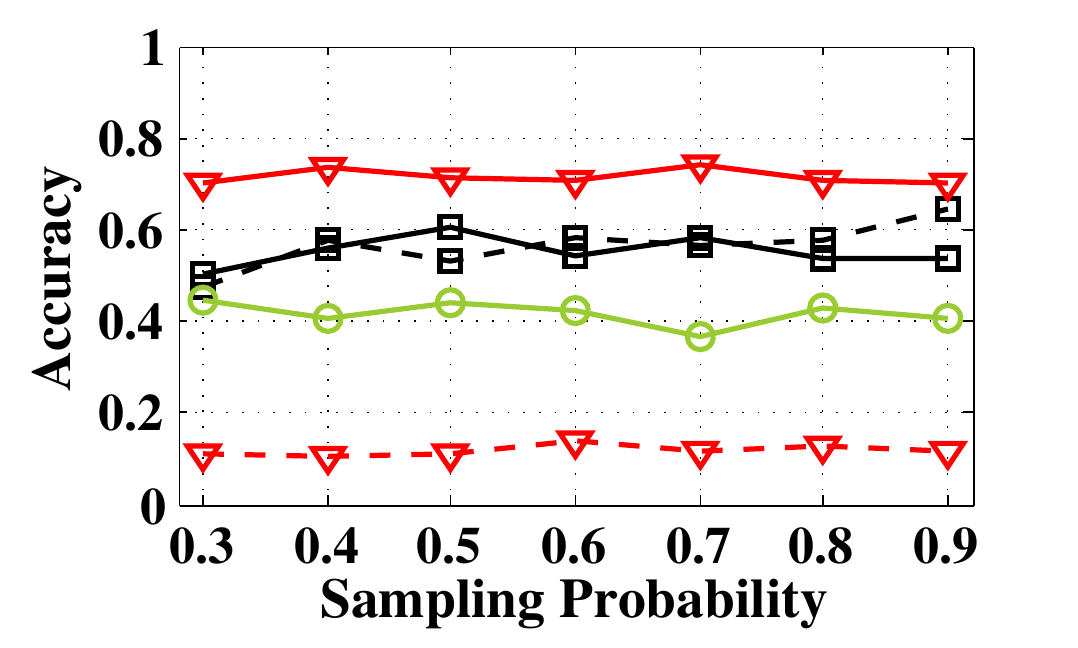}
			\vspace{-3mm}
		\end{minipage}%
		
	}
	\subfigure[N=2000, a=3]{
		\begin{minipage}[]{0.235\linewidth}
			\centering
			\vspace{-3mm}
			\includegraphics[width=1.0\linewidth]{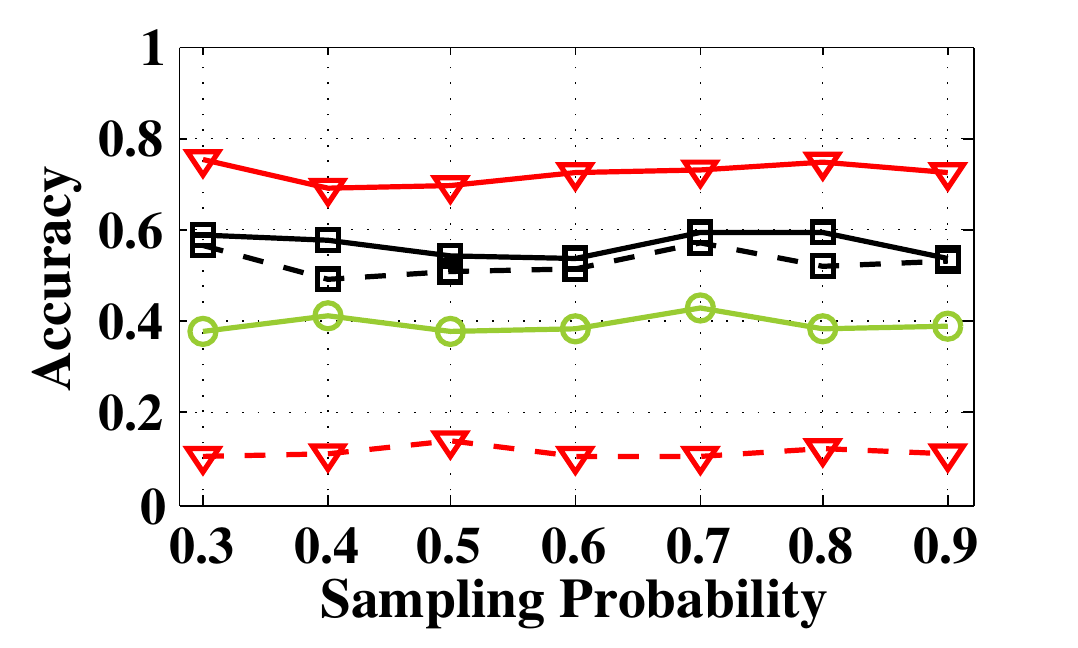}
			\vspace{-3mm}
		\end{minipage}%
		
	}

	\subfigure[N=500, a=5]{
		\begin{minipage}[]{0.235\linewidth}
			\centering
			\vspace{-3mm}
			\includegraphics[width=1.0\linewidth]{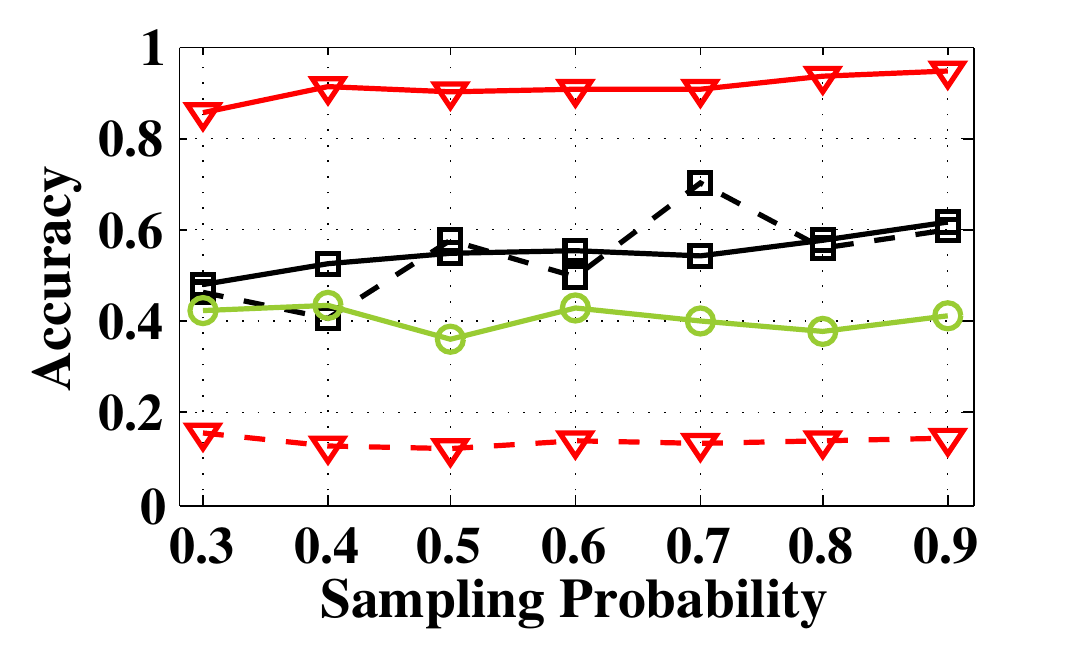}
			\vspace{-3mm}
		\end{minipage}%
		
	}
	\subfigure[N=1000, a=5]{
		\begin{minipage}[]{0.235\linewidth}
			\centering
			\vspace{-3mm}
			\includegraphics[width=1.0\linewidth]{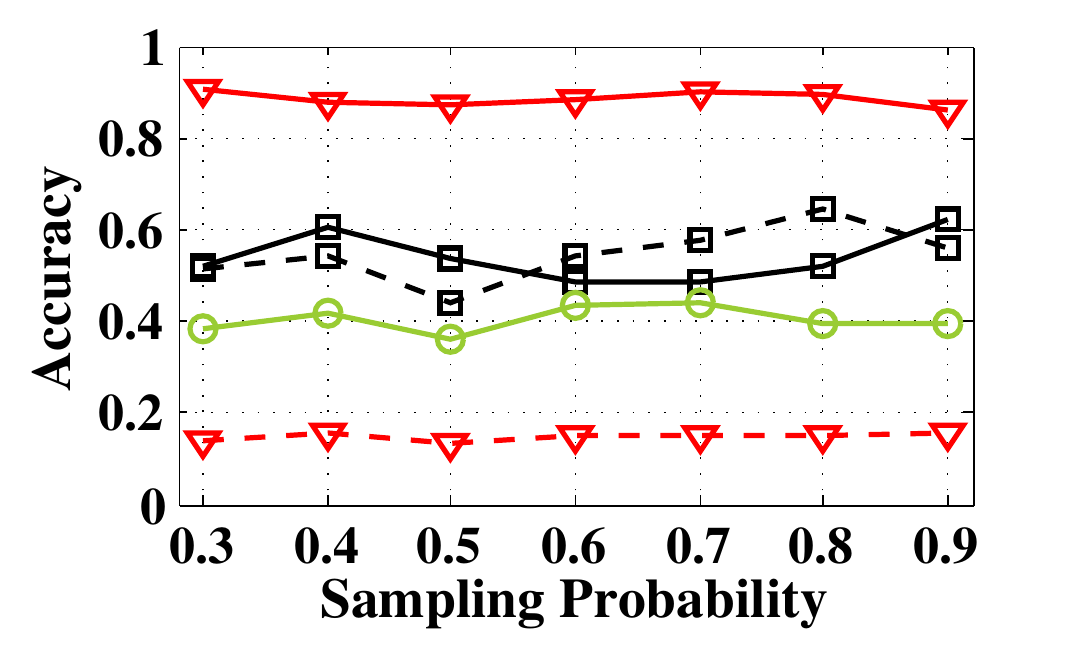}
			\vspace{-3mm}
		\end{minipage}%
		
	}
	\subfigure[N=1500, a=5]{
		\begin{minipage}[]{0.235\linewidth}
			\centering
			\vspace{-3mm}
			\includegraphics[width=1.0\linewidth]{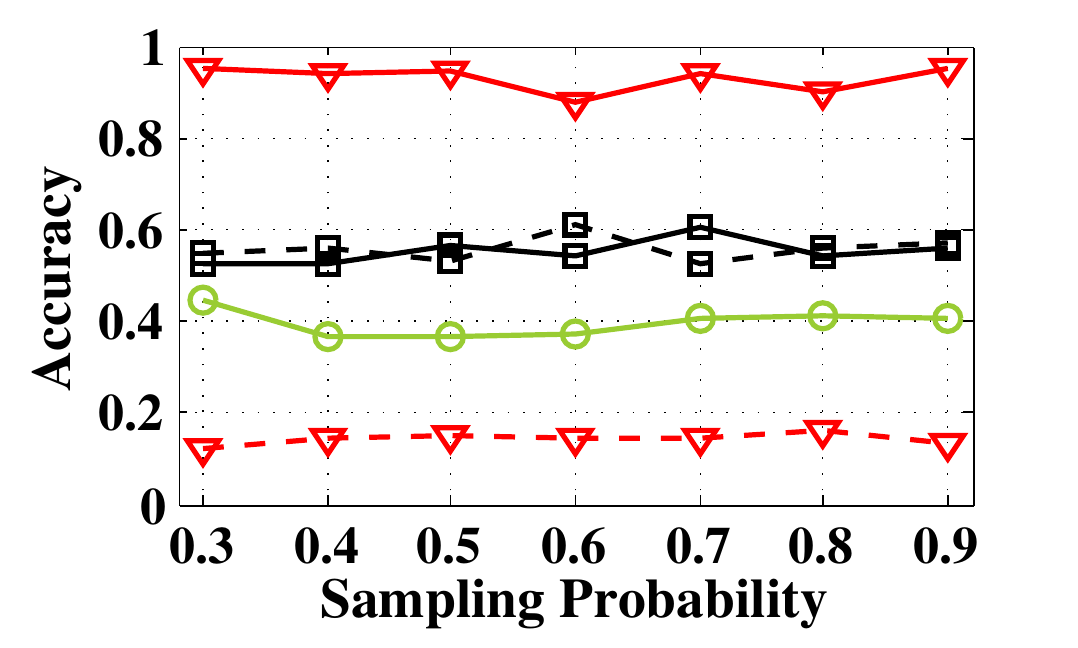}
			\vspace{-3mm}
		\end{minipage}%
		
	}
	\subfigure[N=2000, a=5]{
		\begin{minipage}[]{0.235\linewidth}
			\centering
			\vspace{-3mm}
			\includegraphics[width=1.0\linewidth]{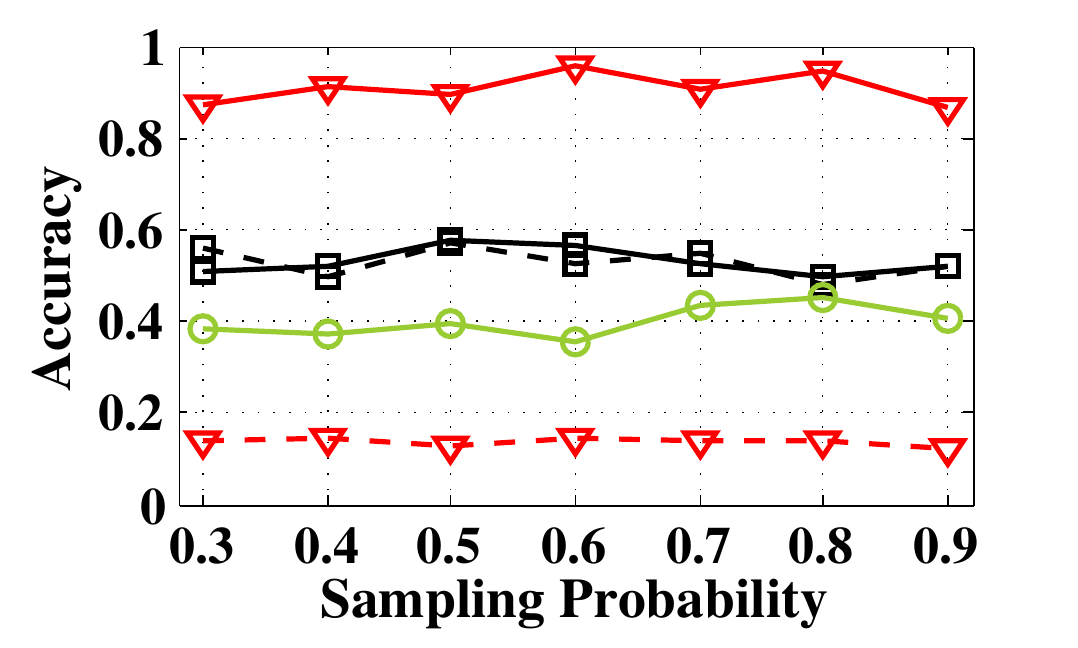}
			\vspace{-3mm}
		\end{minipage}%
		
	}

	\subfigure[N=500, a=7]{
		\begin{minipage}[]{0.235\linewidth}
			\centering
			\vspace{-3mm}
			\includegraphics[width=1.0\linewidth]{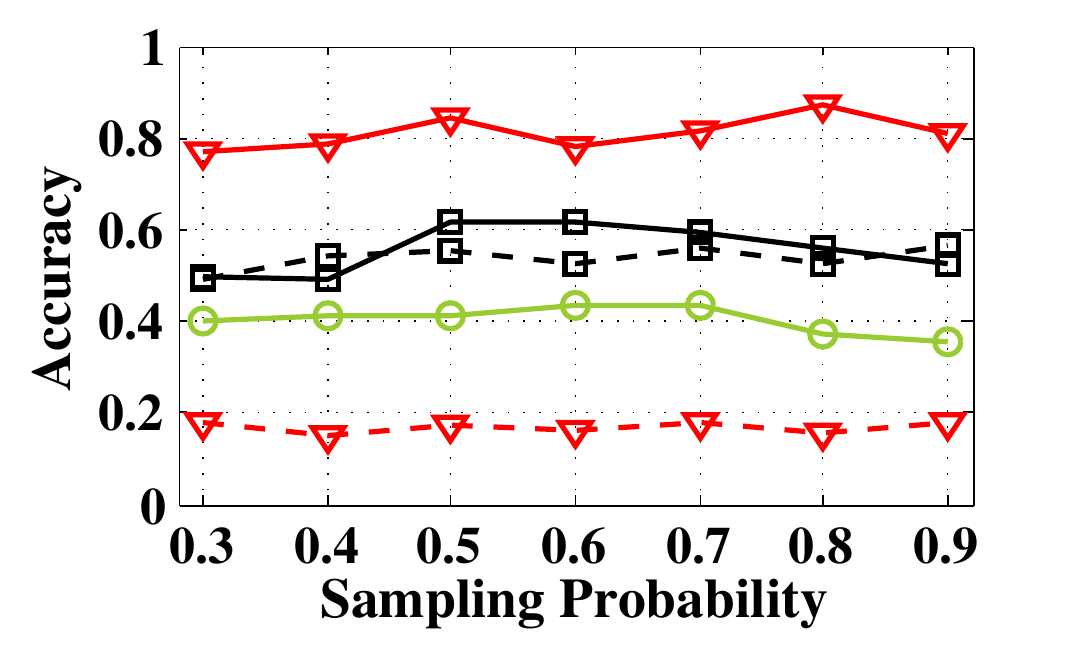}
			\vspace{-3mm}
		\end{minipage}%
		
	}
	\subfigure[N=1000, a=7]{
		\begin{minipage}[]{0.235\linewidth}
			\centering
			\vspace{-3mm}
			\includegraphics[width=1.0\linewidth]{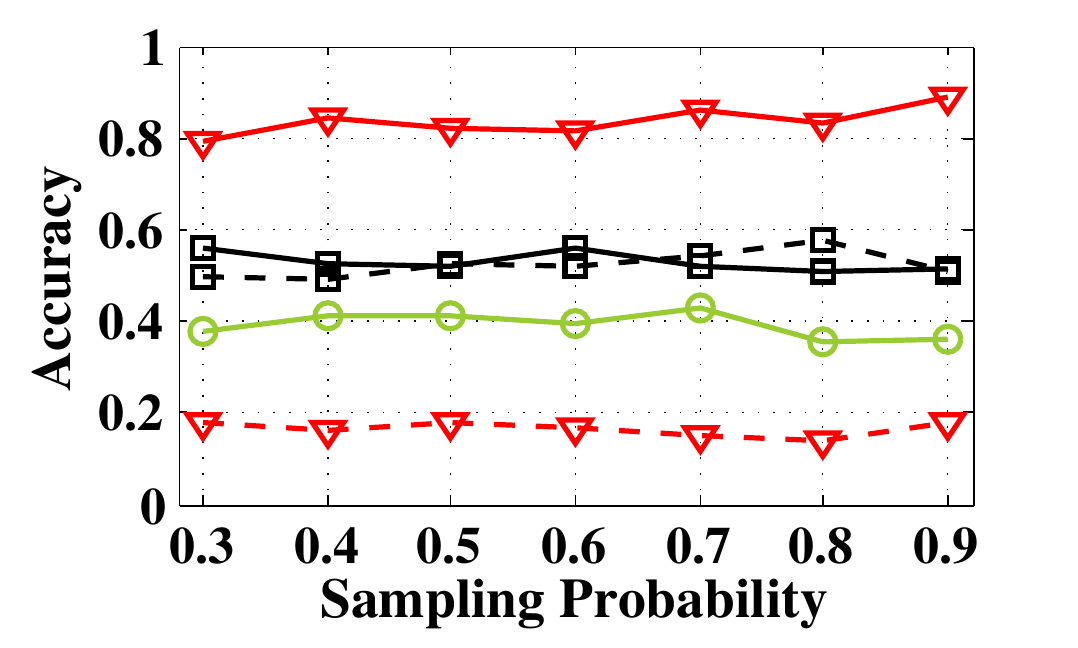}
			\vspace{-3mm}
		\end{minipage}%
		
	}
	\subfigure[N=1500, a=7]{
		\begin{minipage}[]{0.235\linewidth}
			\centering
			\vspace{-3mm}
			\includegraphics[width=1.0\linewidth]{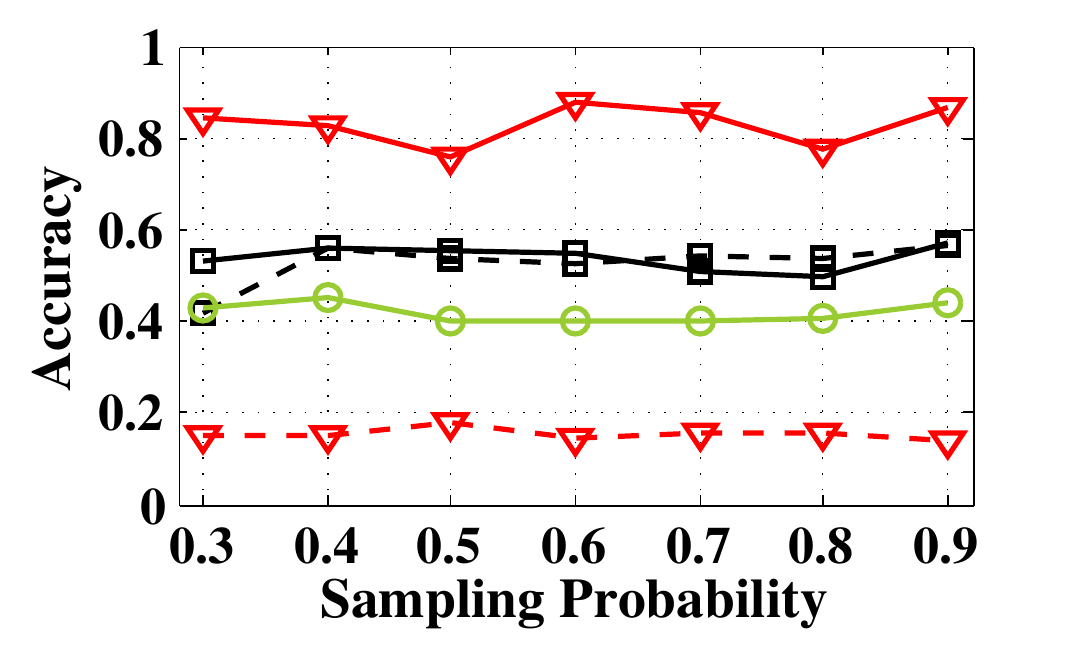}
			\vspace{-3mm}
		\end{minipage}%
		
	}
	\subfigure[N=2000, a=7]{
		\begin{minipage}[]{0.235\linewidth}
			\centering
			\vspace{-3mm}
			\includegraphics[width=1.0\linewidth]{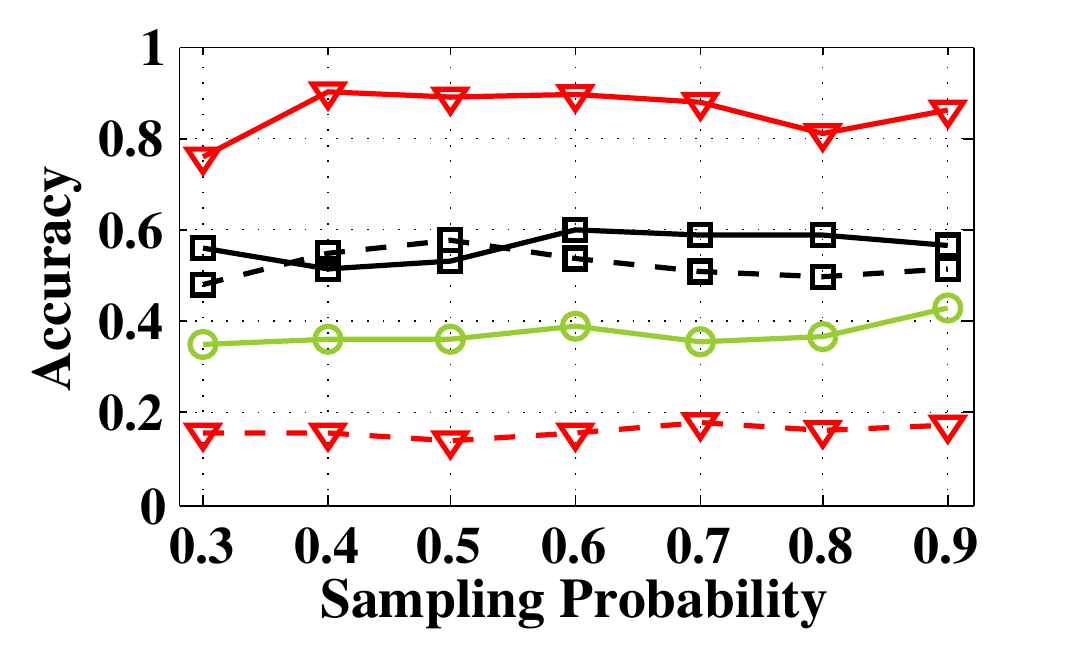}
			\vspace{-3mm}
		\end{minipage}%
		
	}
	\subfigure[N=500, a=9]{
		\begin{minipage}[]{0.235\linewidth}
			\centering
			\vspace{-3mm}
			\includegraphics[width=1.0\linewidth]{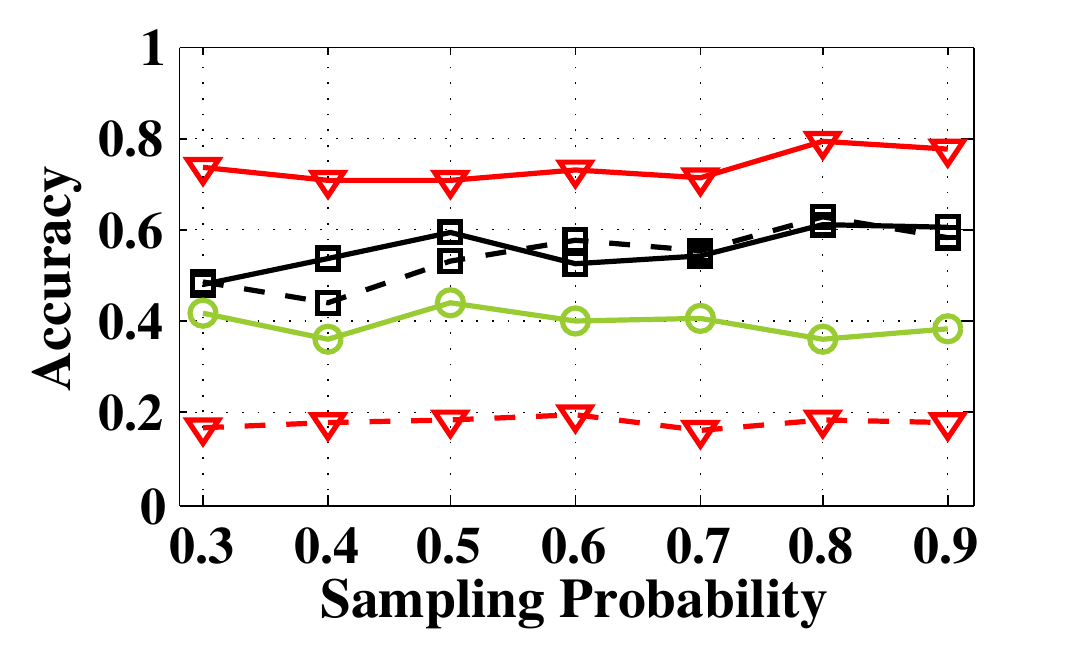}
			\vspace{-3mm}
		\end{minipage}%
		
	}
	\subfigure[N=1000, a=9]{
		\begin{minipage}[]{0.235\linewidth}
			\centering
			\vspace{-3mm}
			\includegraphics[width=1.0\linewidth]{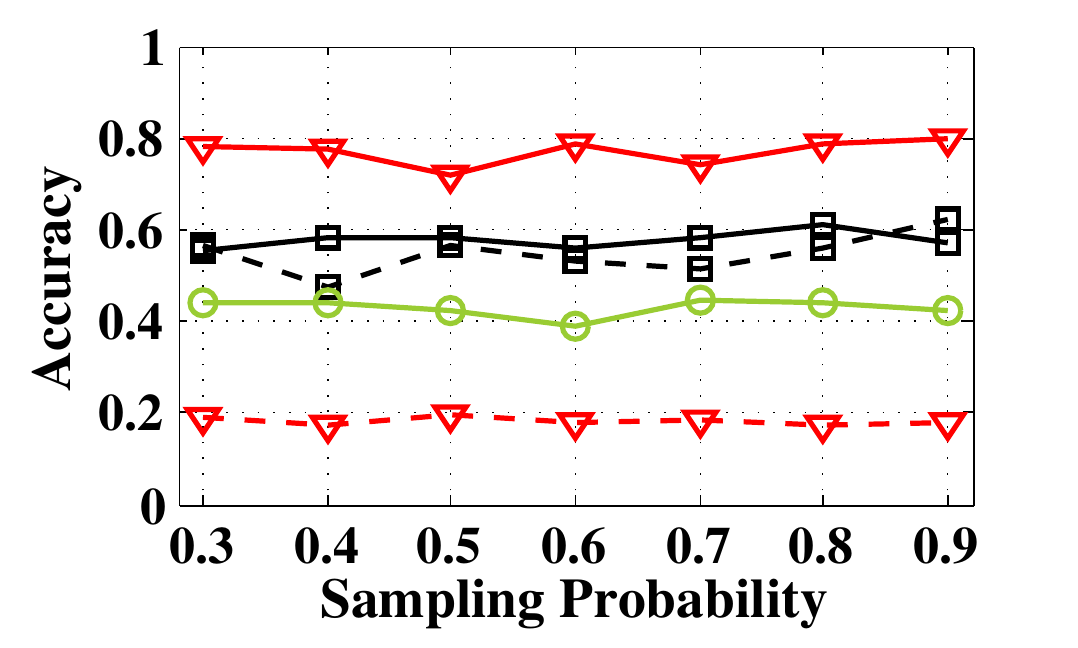}
			\vspace{-3mm}
		\end{minipage}%
		
	}
	\subfigure[N=1500, a=9]{
		\begin{minipage}[]{0.235\linewidth}
			\centering
			\vspace{-3mm}
			\includegraphics[width=1.0\linewidth]{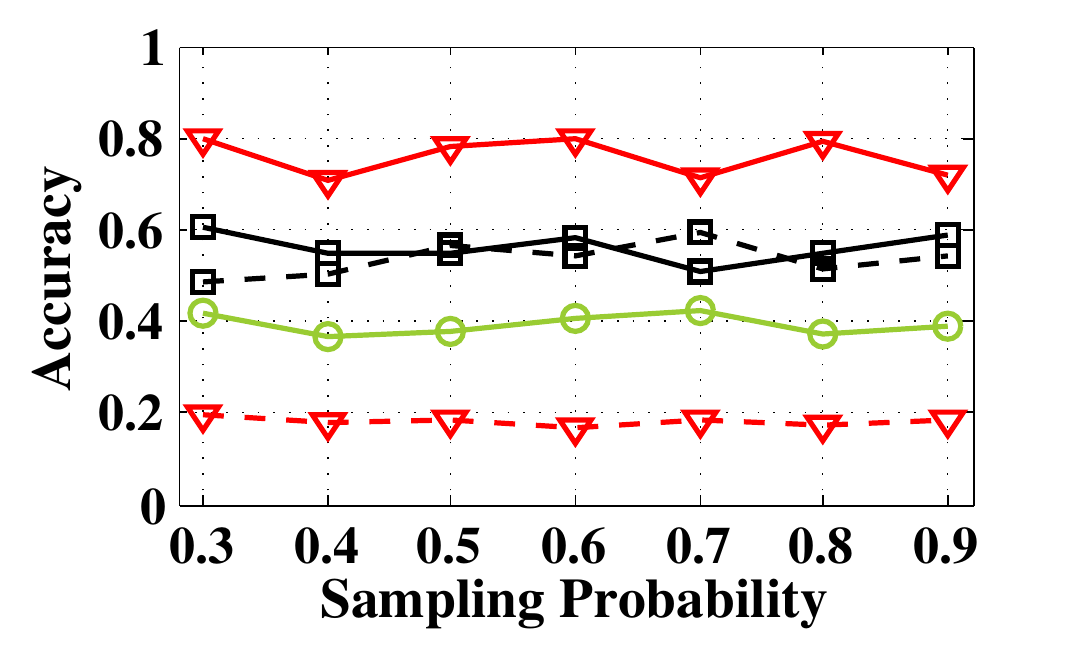}
			\vspace{-3mm}
		\end{minipage}%
		
	}
	\subfigure[N=2000, a=9]{
		\begin{minipage}[]{0.235\linewidth}
			\centering
			\vspace{-3mm}
			\includegraphics[width=1.0\linewidth]{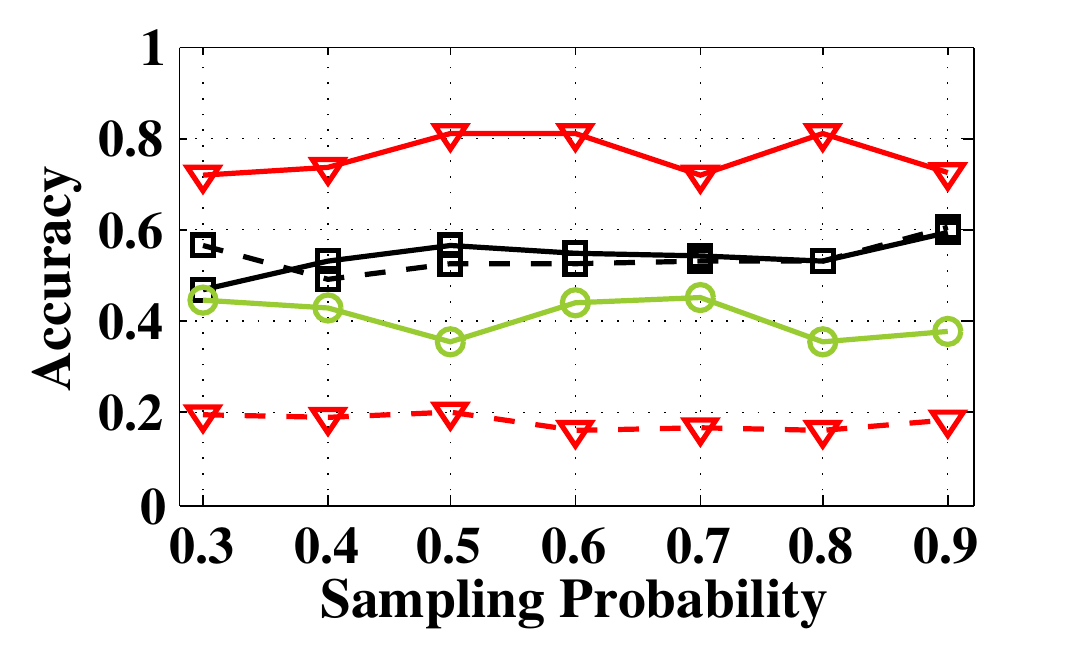}
			\vspace{-3mm}
		\end{minipage}%
		
	}
 \includegraphics[width=0.95\textwidth]{legend.pdf}
		\caption{Experiments on Synthetic Network with $\eta=0.1$.}
		\label{fig:syn1}
\end{figure*}

\subsubsection{Experimental Datasets} We discuss three adopted datasets in an order from model-based to real social networks.

\underline{1. Synthetic Networks:} When we generate synthetic networks, there are two main steps: (i) randomly setting the community representation of every node and (ii) judging whether an edge exists between any two nodes. For step (i), since the nodes and communities in our model are both independently distributed, step (i) can be viewed as a Bernoulli trial for every node: Setting the probability that node $i$ belongs to any one community as $p_{c_i}$, then the probability that node $i$ belongs to $k$ different communities is $p_{c_{ik}}=C_{m}^kp_{c_i}^k(1-p_{c_i})^{m-k}$. In our experiment we set the same $p_{c_i}$ for all nodes as we view them equally. For step (ii),  we can set the probability of edge existence between any two nodes based on Eqn. (\ref{exp39})\footnote{Unlike existing work \cite{cite:seedless} which determines the edge existence in the graph based on different distributions like Poisson, power law or exponential expected degree distributions, we strictly follow the OSBM and alter the edge distribution by modifying the parameter $a$.}, with $x$ determined by the community representation matrix (Recall Definition \ref{def11}). In experiments on this dataset, we adjust the parameters based on Table \ref{table:notation3} to validate the performance of our algorithm under different network settings.

\underline{2. Sampled Real Social Networks:} In sampled real social networks, the underlying social network $G$ is extracted from the dataset LiveJournal \cite{cite:livejournal} without changes. The published and auxiliary networks ($G_1$ and $G_2$) are artificially sampled from $G$ with the  same probability $s$. To compare with the results in synthetic networks, we keep the settings of $N$, $s$ and $\eta$ as Table \ref{table:notation3}. However, since $G$ is not generated from the OSBM, the OSBM parameter $a$ does not exist anymore. In experiments on this dataset, we adjust $N$, $s$ and $\eta$ to characterize different possible situations based on real underlying networks.

\underline{3. Cross-Domain Co-author Networks:} The co-author networks are from the Microsoft Academic Graph (MAG) \cite{cite:MAG}. We extract $4$ networks belonging to different sub-areas in the field of computer science, with the same group of authors, each of whom has a unique $8$-bit hexadecimal ID enabling us to construct the true mapping between two networks as the one mapping nodes with same ID. Each network can be viewed as $G_1$ or $G_2$, thus there are $C_4^2=6$ combinations. (Table \ref{table:notation3}) Note that we can assign $w_{ij}$ on all these $3$ datasets since the prior knowledge is just $M$, which can be generated or known from the real networks. In experiments on this dataset, the results accurately reflect the practical situations.
\begin{figure*}[!tb]
	\centering
	\subfigure[N=500, C=0.05]{
		\begin{minipage}[]{0.235\linewidth}
			\centering
			\vspace{-3mm}
			\includegraphics[width=1.0\linewidth]{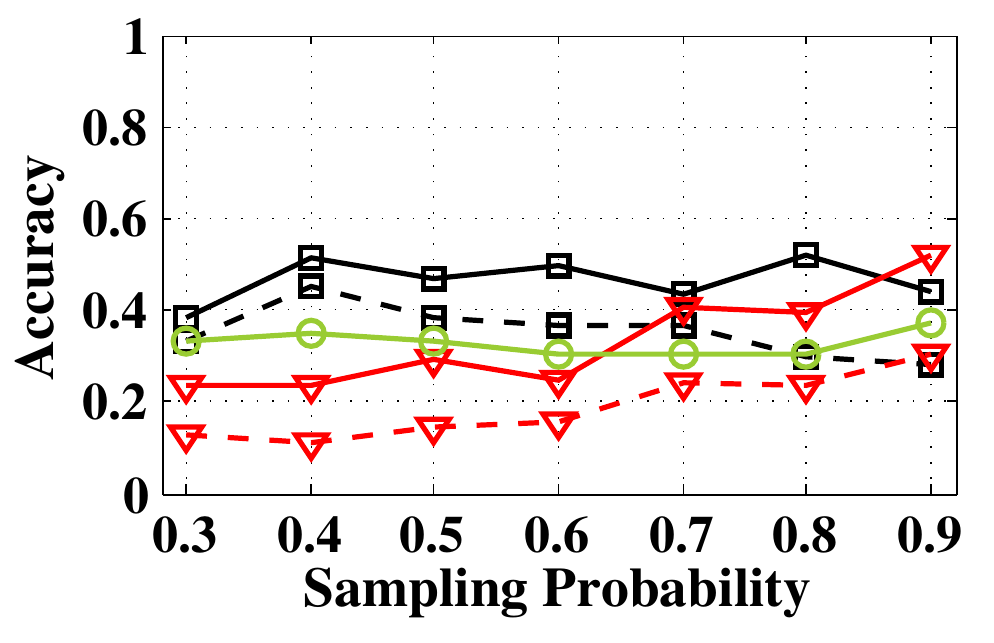}
			\vspace{-3mm}
		\end{minipage}%
		
	}
	\subfigure[N=1000, C=0.05]{
		\begin{minipage}[]{0.235\linewidth}
			\centering
			\vspace{-3mm}
			\includegraphics[width=1.0\linewidth]{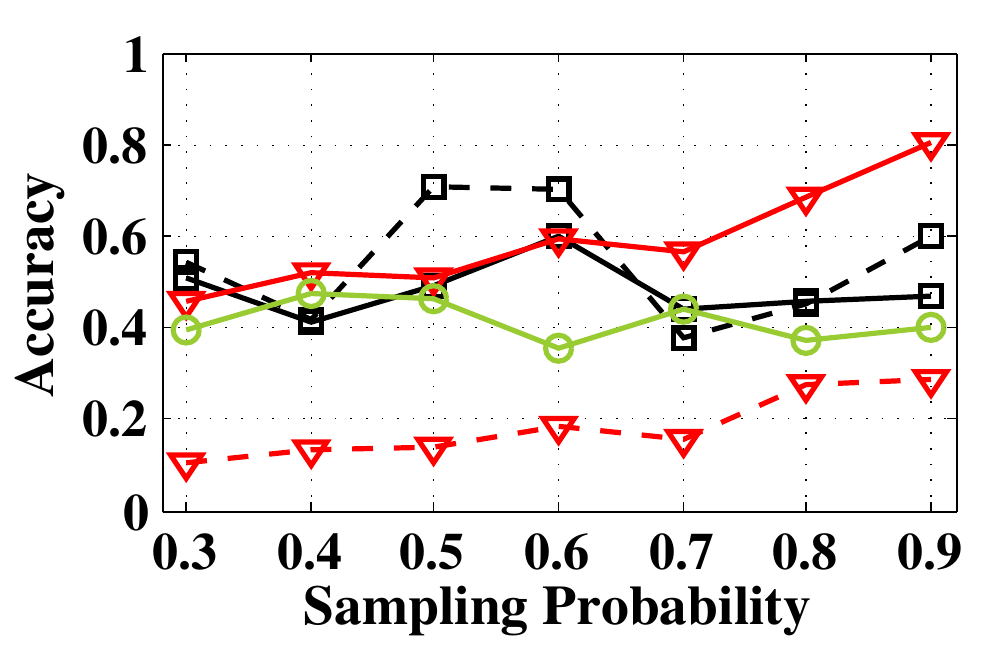}
			\vspace{-3mm}
		\end{minipage}%
		
	}
	\subfigure[N=1500, C=0.05]{
		\begin{minipage}[]{0.235\linewidth}
			\centering
			\vspace{-3mm}
			\includegraphics[width=1.0\linewidth]{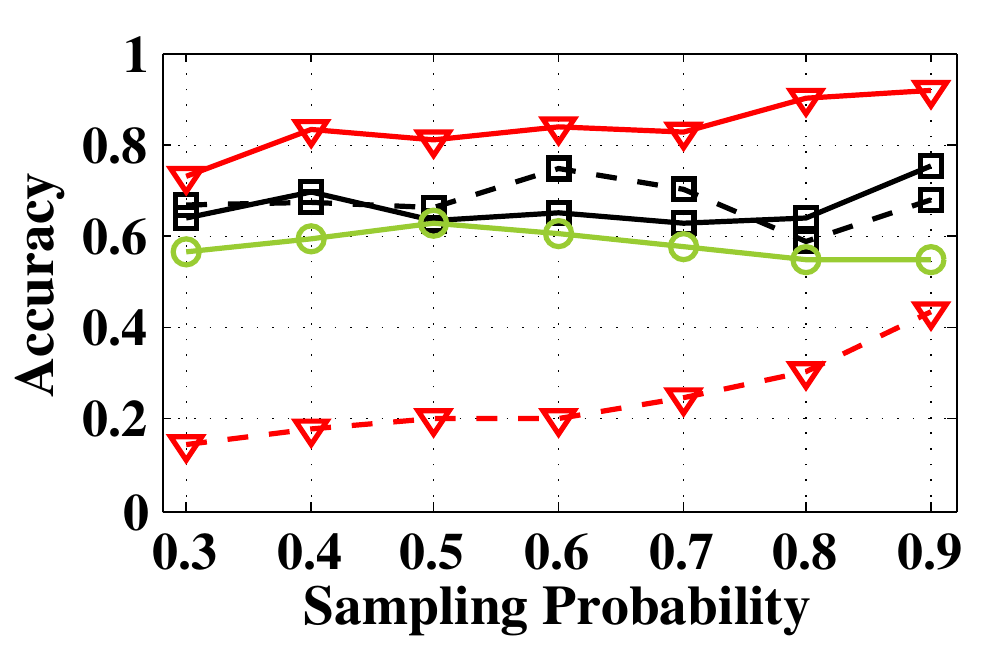}
			\vspace{-3mm}
		\end{minipage}%
		
	}
	\subfigure[N=2000, C=0.05]{
		\begin{minipage}[]{0.235\linewidth}
			\centering
			\vspace{-3mm}
			\includegraphics[width=1.0\linewidth]{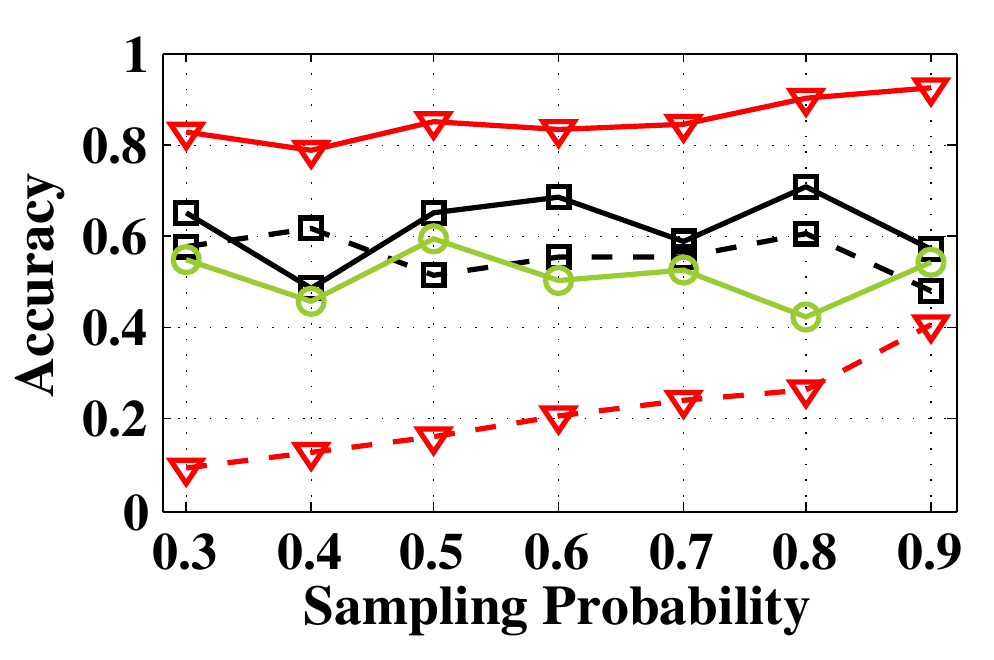}
			\vspace{-3mm}
		\end{minipage}%
		
	}
	\subfigure[N=500, C=0.1]{
		\begin{minipage}[]{0.235\linewidth}
			\centering
			\vspace{-3mm}
			\includegraphics[width=1.0\linewidth]{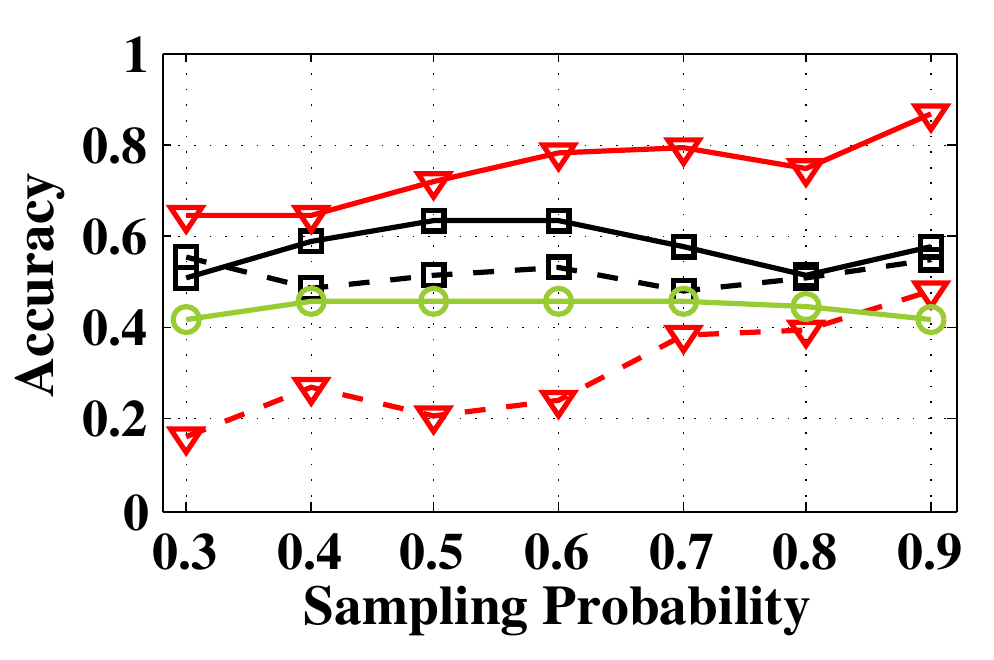}
			\vspace{-3mm}
		\end{minipage}%
		
	}
	\subfigure[N=1000, C=0.1]{
		\begin{minipage}[]{0.235\linewidth}
			\centering
			\vspace{-3mm}
			\includegraphics[width=1.0\linewidth]{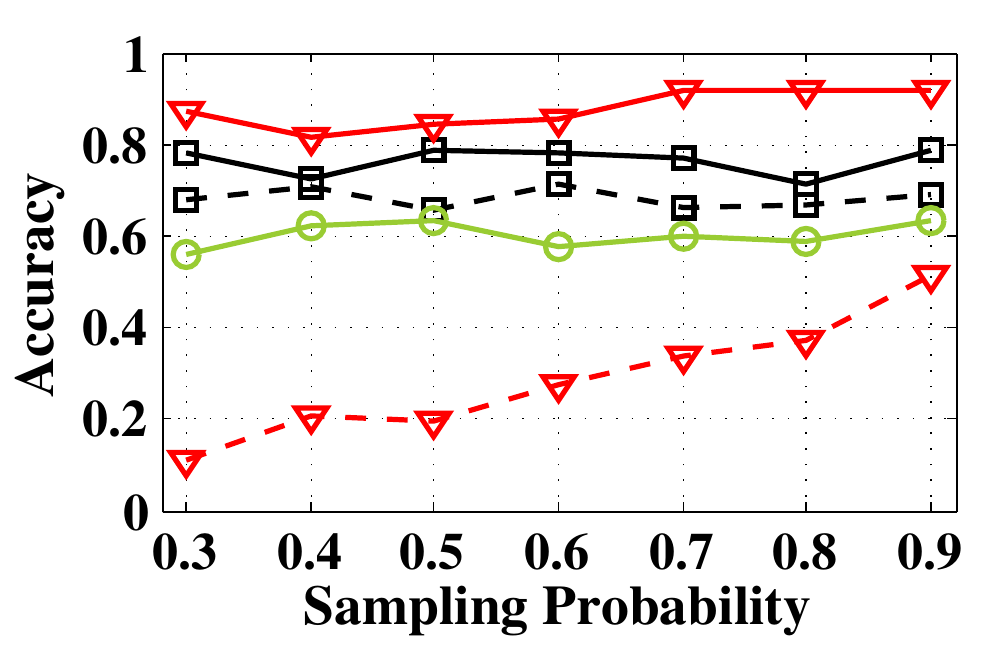}
			\vspace{-3mm}
		\end{minipage}%
		
	}
	\subfigure[N=1500, C=0.1]{
		\begin{minipage}[]{0.235\linewidth}
			\centering
			\vspace{-3mm}
			\includegraphics[width=1.0\linewidth]{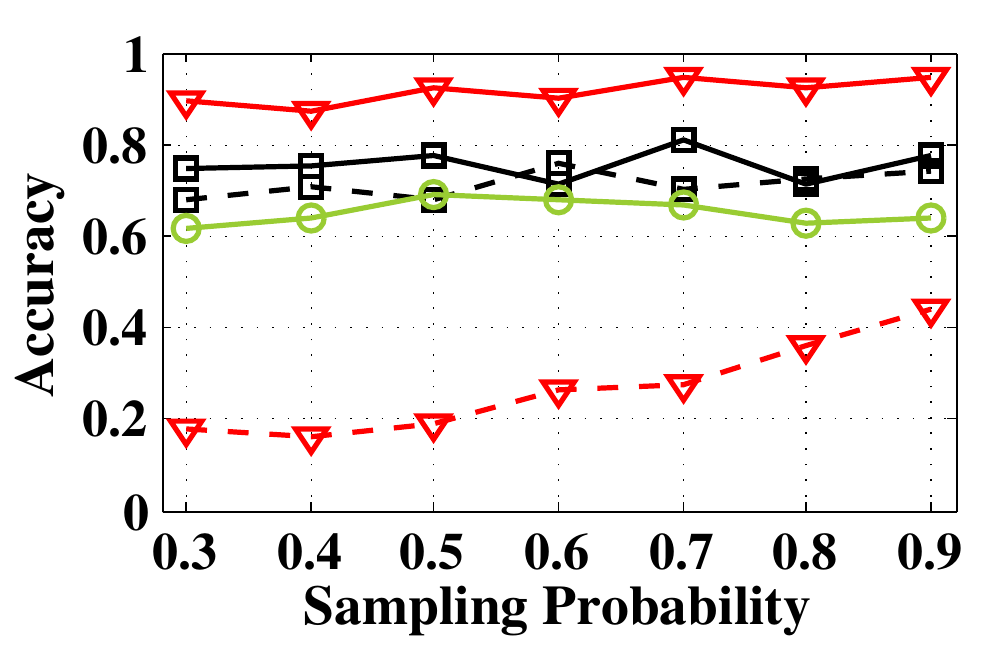}
			\vspace{-3mm}
		\end{minipage}%
		
	}
	\subfigure[N=2000, C=0.1]{
		\begin{minipage}[]{0.235\linewidth}
			\centering
			\vspace{-3mm}
			\includegraphics[width=1.0\linewidth]{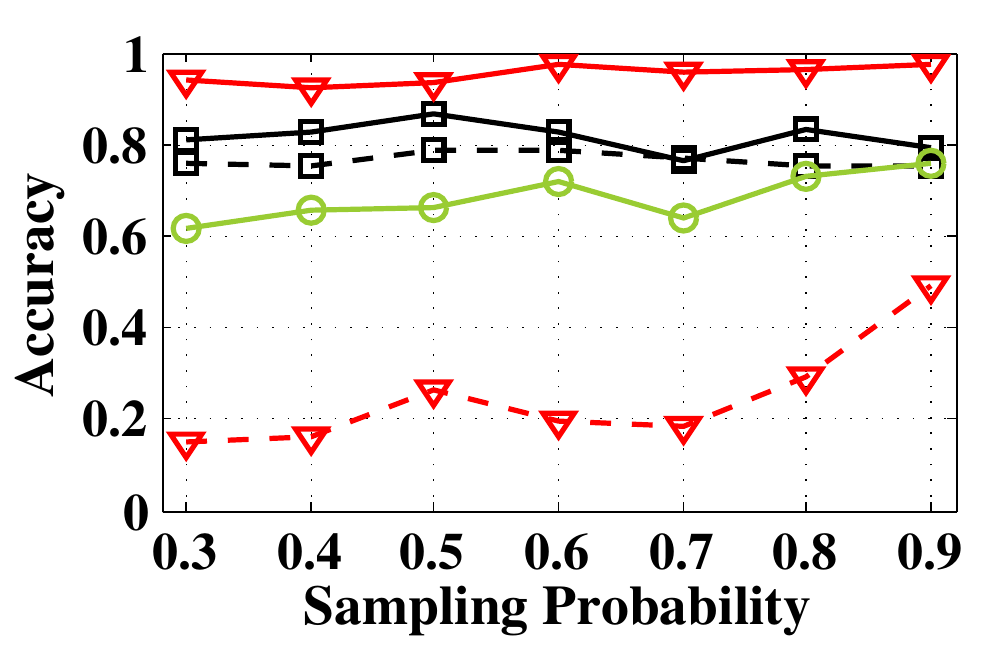}
			\vspace{-3mm}
		\end{minipage}%
		
	}
 \includegraphics[width=0.95\textwidth]{legend.pdf}
		\caption{Experiments on Sampled Real Social Networks.}
		\label{fig:soc}
\end{figure*}

	\begin{table}[!tb]
					\setlength{\extrarowheight}{5pt}
					\begin{scriptsize}
						\renewcommand\arraystretch{0.82}
						\caption{\bf Datasets in Basic Experiments}
						\centering
						\label{table:notation2}
						\resizebox{1.0\columnwidth}{!}{
							
							\begin{tabular}{l|l|l|l}\hline
								\textbf{Dataset} & Synthetic & Sampled Real Social & Cross-Domain Co-author \\ \hline
                                \textbf{Source} & OSBM & LiveJournal \cite{cite:SNAP} & MAG \cite{cite:MAG} \\ \hline
								\textbf{Num. of Nodes} & $500\sim2000$ & $500\sim2000$ & $3176$ \\ \hline
								\textbf{Num. of Communities} & $25\sim1000$ & $25\sim1000$ & $89$ \\
							\end{tabular}}		
							\vspace{-4.6mm}
						\end{scriptsize}
					\end{table}

\subsubsection{Algorithms for Comparison and Performance Metric}
Note that the main point of our experiments is to show the influence of overlapping communities on the accuracy, and our algorithm can effectively harness this overlapping property. Therefore, We exclude algorithms for seeded de-anonymization and select algorithms suitable for seedless cases related to our main point: showing the impact of overlapping communities on reducing NME, though other algorithms might outperform ours. We select two algorithms for comparison: (i) the Genetic Algorithm (GA), an epitome of heuristic algorithms, however due to its instability\footnote{The instability of GA will be shown in experimental results.}, we run $10$ times and average these results as the accuracy of GA in every experiment; (ii) the Convex Optimization-Based Algorithm (COBA) in \cite{Fu:arxiv,Fu:GC}, assigning a node to a unique community, which primarily suits non-overlapping cases. The performance metric is \emph{accuracy}, the proportion of correctly mapped nodes.

\begin{figure}[!tb]
	\centering
	\subfigure[Overlapping]{
		\begin{minipage}[]{0.46\linewidth}
			\centering
			\vspace{-3mm}
			\includegraphics[width=1.0\linewidth]{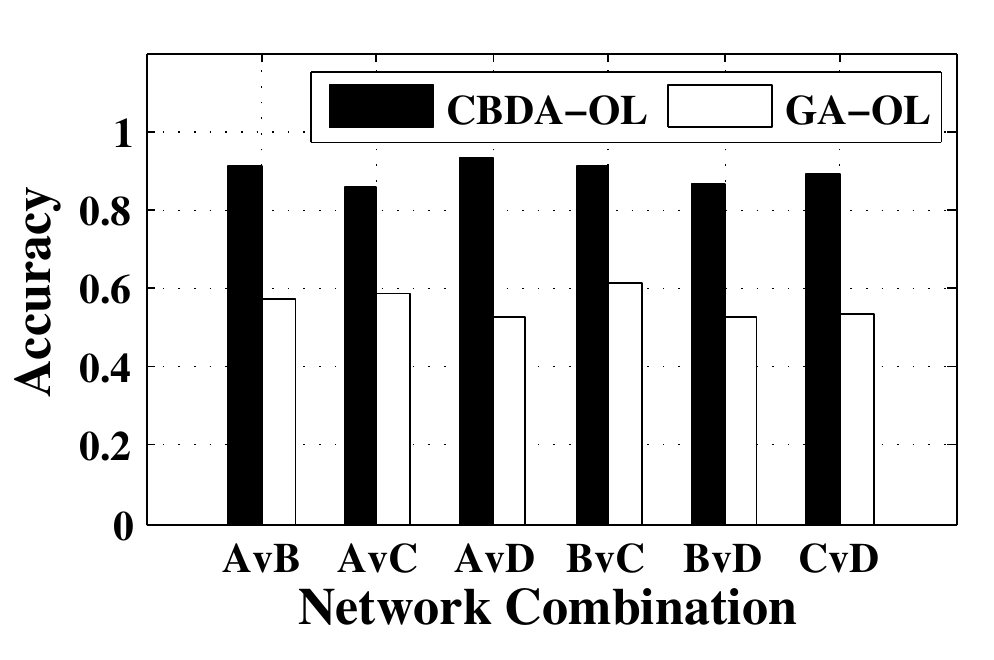}
			\vspace{-3mm}
		\end{minipage}%
		
	}
	\subfigure[Non-Overlapping]{
		\begin{minipage}[]{0.46\linewidth}
			\centering
			\vspace{-3mm}
			\includegraphics[width=1.0\linewidth]{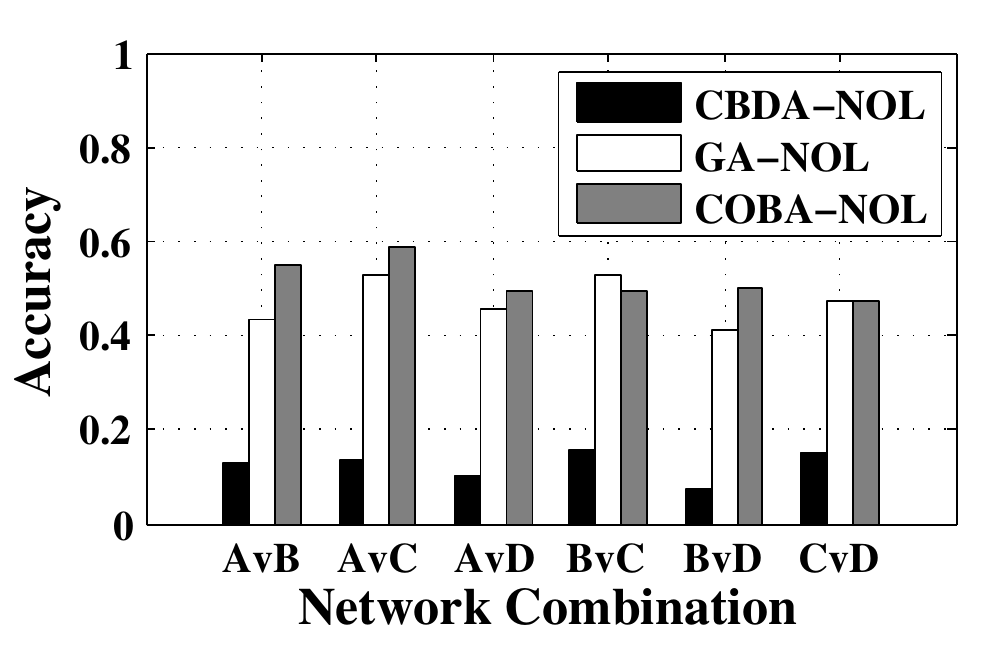}
			\vspace{-3mm}
		\end{minipage}%
		
	}
\caption{Experiments on Cross-Domain Co-author Networks.}
	\label{fig:coauthor}
\end{figure}

\begin{figure}[htbp]
     	\centering
		\includegraphics[width=0.46\textwidth]{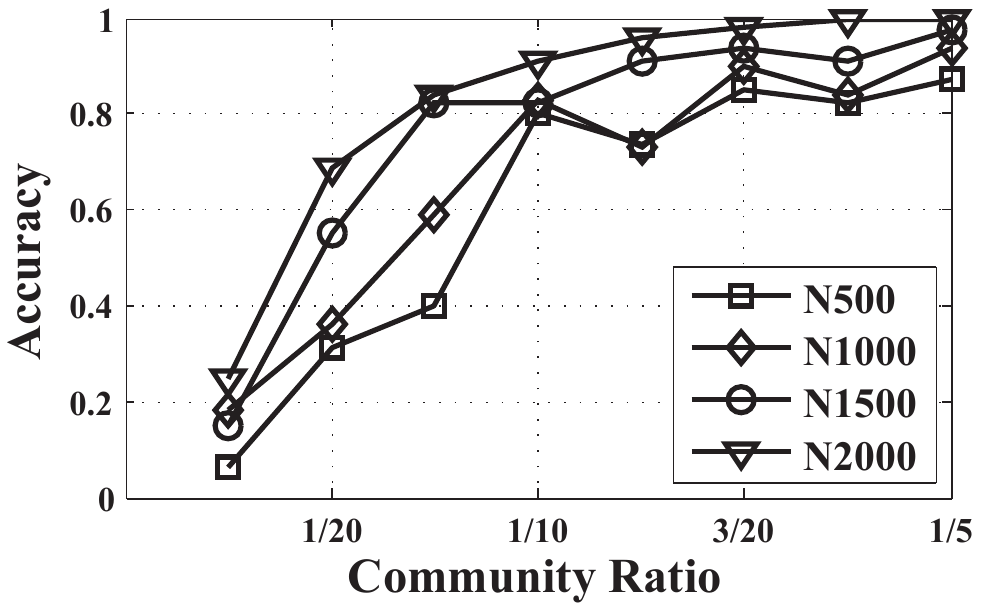}
		\caption{The Influence of Community Ratio on Accuracy.}
		\label{Fig:CommunityDiff}
\end{figure}

\begin{figure}[htbp]
     	\centering
		\includegraphics[width=0.46\textwidth]{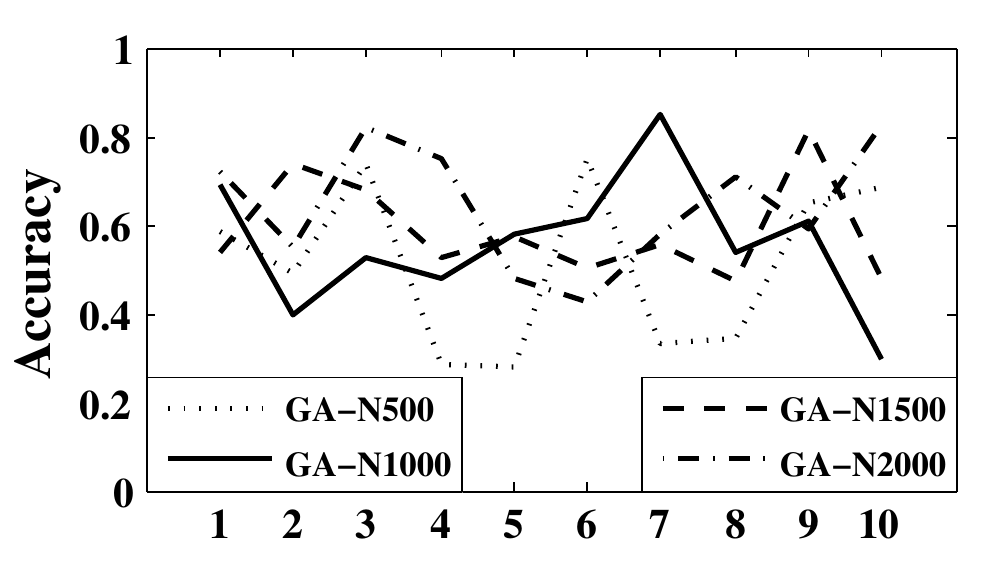}
		\caption{The Instability of Genetic Algorithm.}
		\label{Fig:GAinstable}
\end{figure}

\subsubsection{Supplementary Experiments}
To make our experimental validation more comprehensive and convincing, we supplement three experiments: (i) We study the effect of different community ratios ($\eta$) on the accuracy based on sampled real social networks. We modify $\eta$ from $0.025$ to $0.2$ with interval $0.025$; (ii) We study whether the weight matrix $\mathbf{W}$ in our cost function makes for the higher accuracy, compared with the cost function without appending $\mathbf{W}$ in existing work \cite{cite:seedless} . Appending $\mathbf{W}$ means adding the community information in the cost function. (iii) We study the instability of genetic algorithm (GA) and reveals the reason why GA lacks practical usage even if it achieves acceptable average accuracy in our main experiments.

\begin{figure*}[!htbp]
	\centering
	\subfigure[N=500, C=0.05]{
		\begin{minipage}[]{0.235\linewidth}
			\centering
			\vspace{-3mm}
			\includegraphics[width=1.0\linewidth]{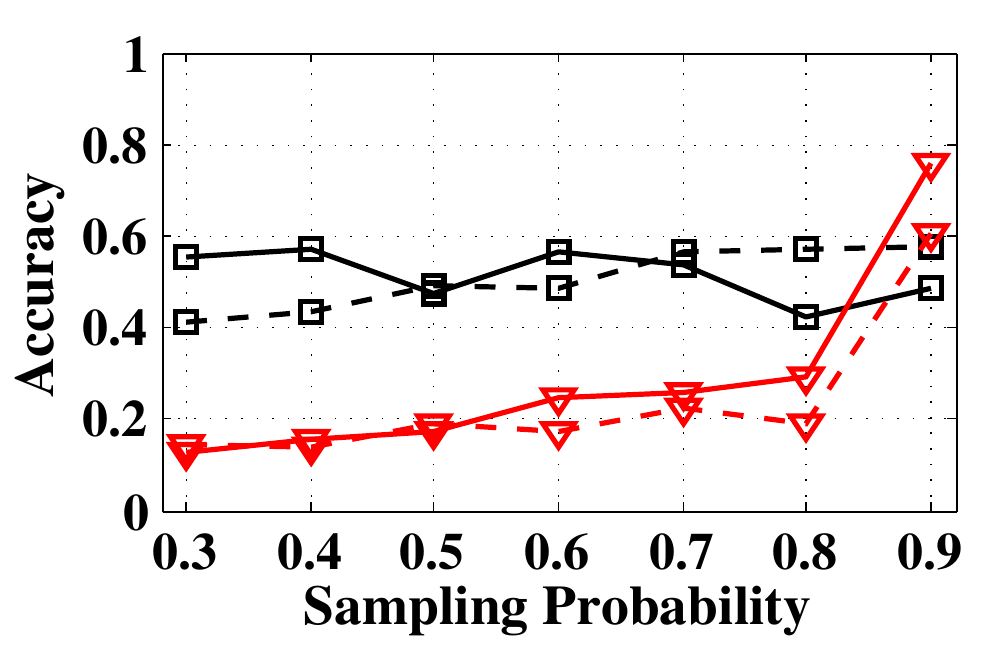}
			\vspace{-3mm}
		\end{minipage}%
		
	}
	\subfigure[N=1000, C=0.05]{
		\begin{minipage}[]{0.235\linewidth}
			\centering
			\vspace{-3mm}
			\includegraphics[width=1.0\linewidth]{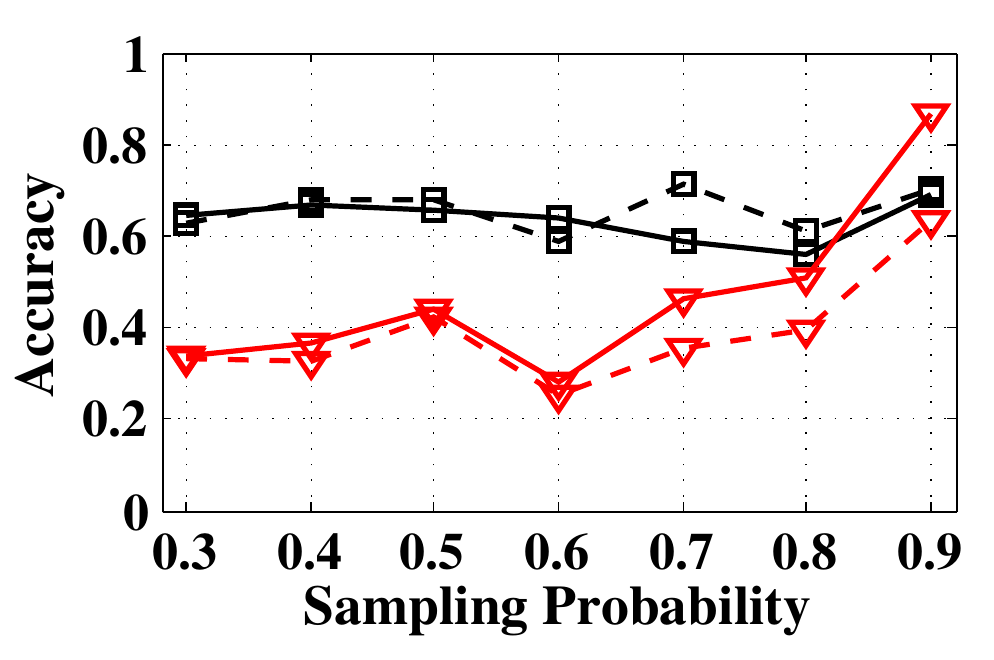}
			\vspace{-3mm}
		\end{minipage}%
		
	}
	\subfigure[N=1500, C=0.05]{
		\begin{minipage}[]{0.235\linewidth}
			\centering
			\vspace{-3mm}
			\includegraphics[width=1.0\linewidth]{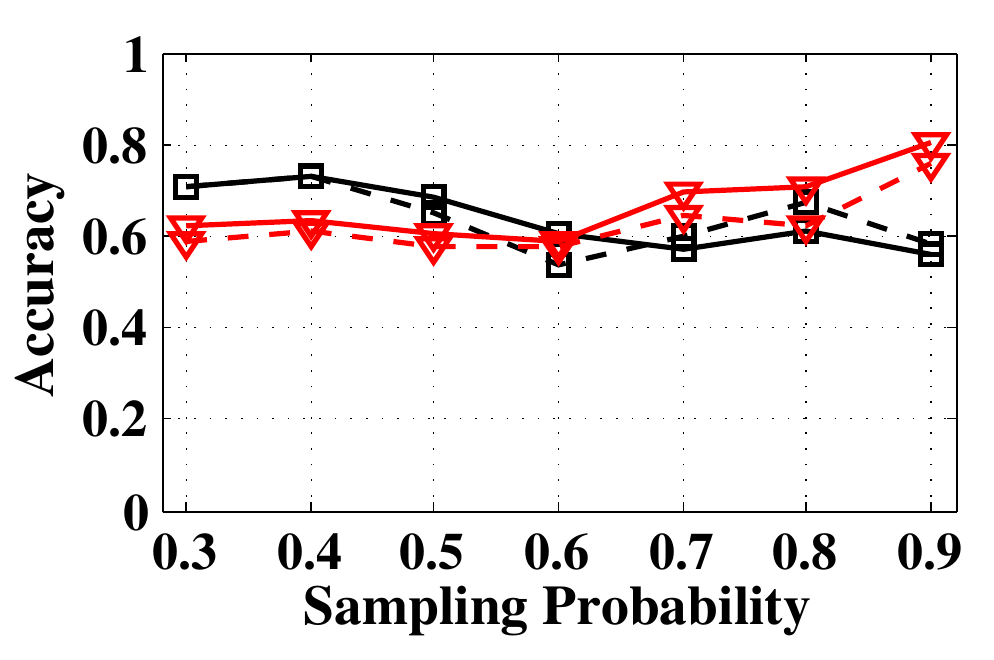}
			\vspace{-3mm}
		\end{minipage}%
		
	}
	\subfigure[N=2000, C=0.05]{
		\begin{minipage}[]{0.235\linewidth}
			\centering
			\vspace{-3mm}
			\includegraphics[width=1.0\linewidth]{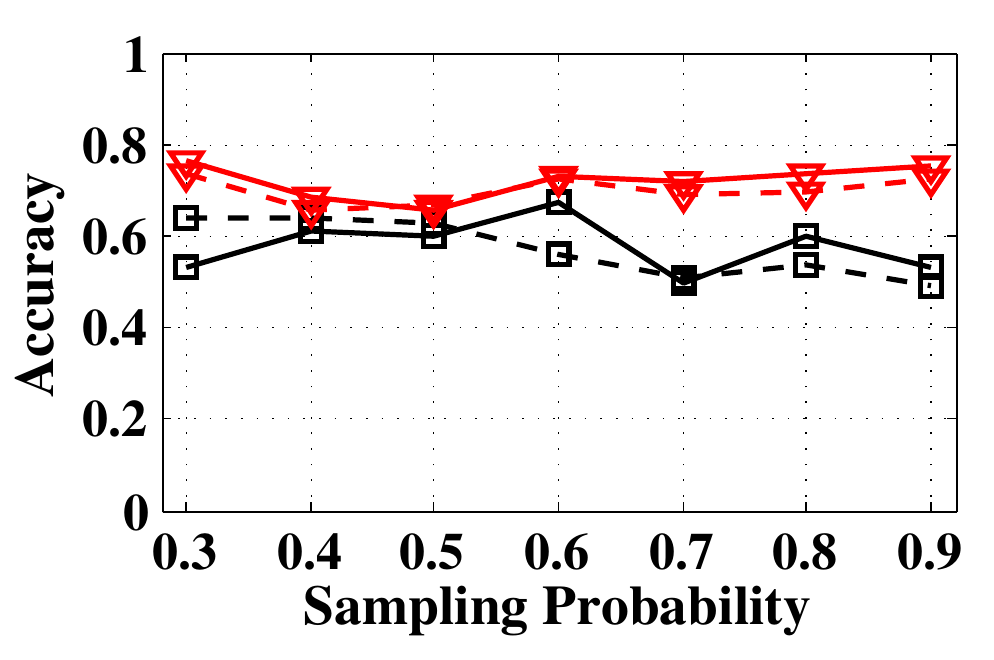}
			\vspace{-3mm}
		\end{minipage}%
		
	}
	\subfigure[N=500, C=0.1]{
		\begin{minipage}[]{0.235\linewidth}
			\centering
			\vspace{-3mm}
			\includegraphics[width=1.0\linewidth]{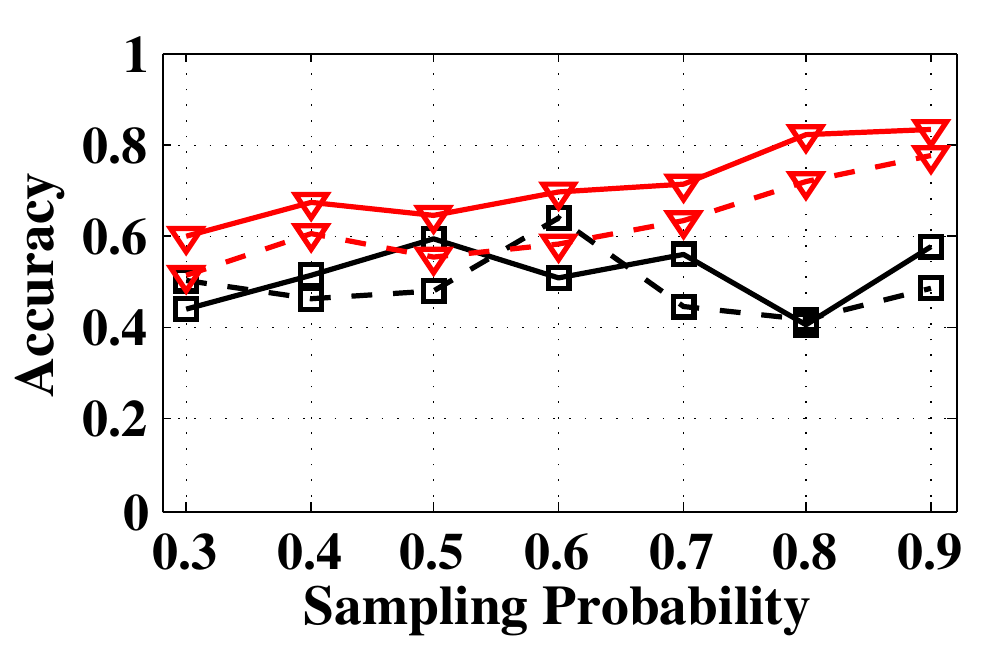}
			\vspace{-3mm}
		\end{minipage}%
		
	}
	\subfigure[N=1000, C=0.1]{
		\begin{minipage}[]{0.235\linewidth}
			\centering
			\vspace{-3mm}
			\includegraphics[width=1.0\linewidth]{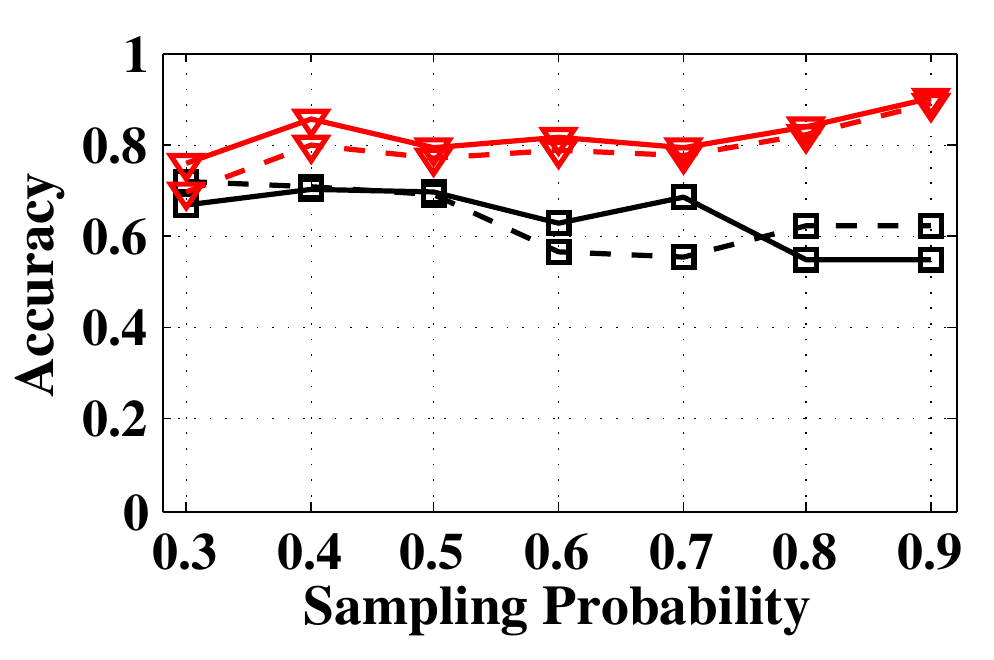}
			\vspace{-3mm}
		\end{minipage}%
		
	}
	\subfigure[N=1500, C=0.1]{
		\begin{minipage}[]{0.235\linewidth}
			\centering
			\vspace{-3mm}
			\includegraphics[width=1.0\linewidth]{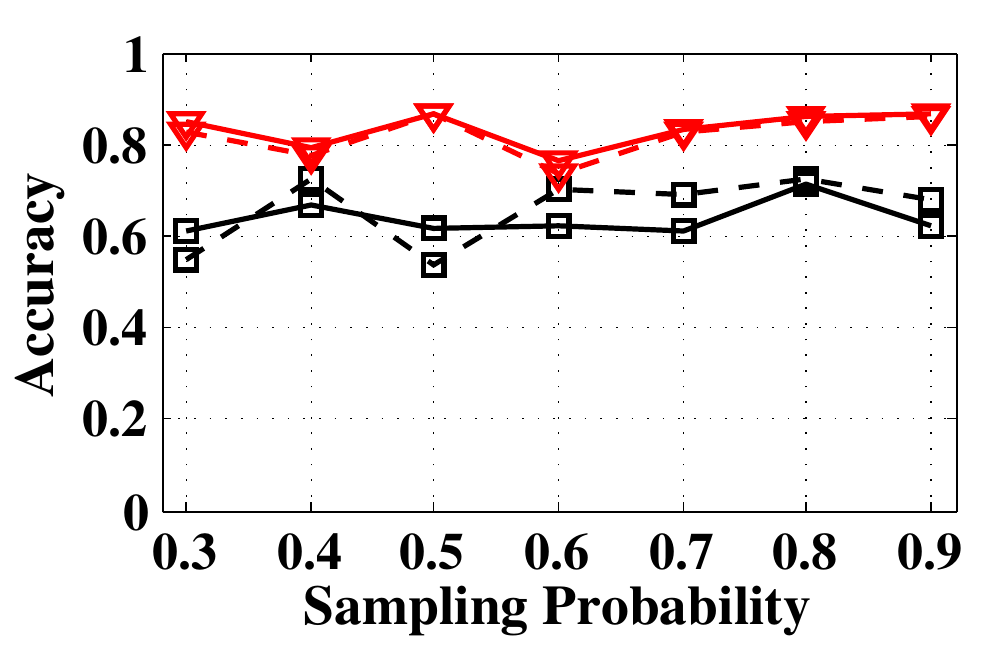}
			\vspace{-3mm}
		\end{minipage}%
		
	}
	\subfigure[N=2000, C=0.1]{
		\begin{minipage}[]{0.235\linewidth}
			\centering
			\vspace{-3mm}
			\includegraphics[width=1.0\linewidth]{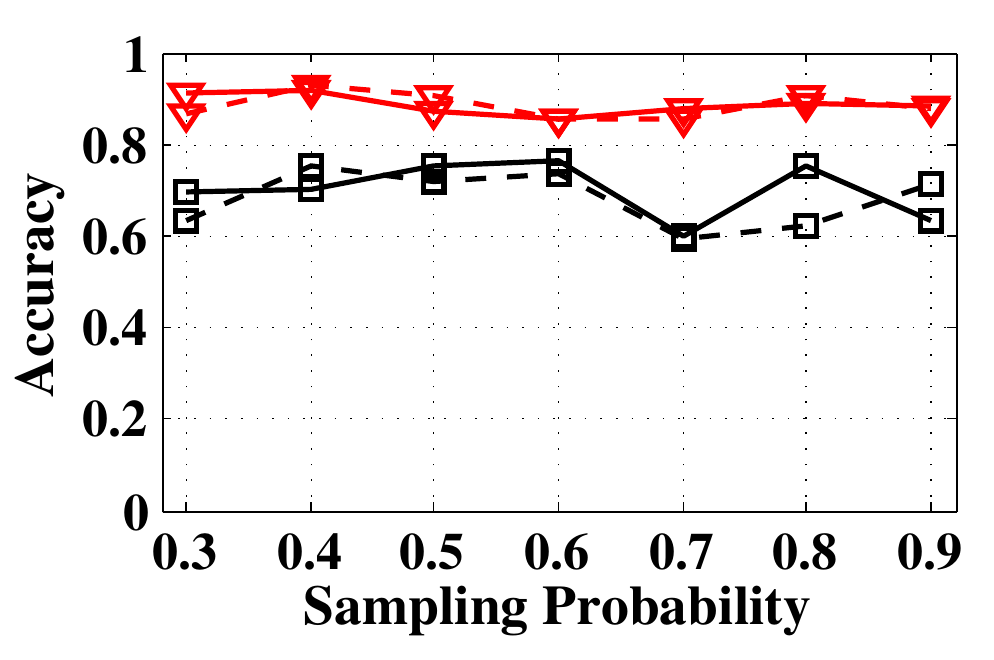}
			\vspace{-3mm}
		\end{minipage}%
		
	}
		\includegraphics[width=0.9\textwidth]{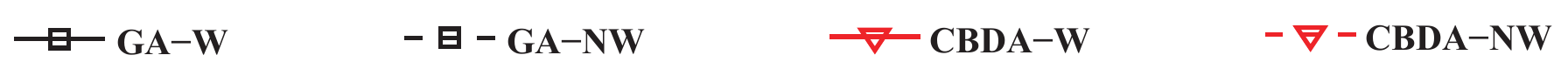}
		\caption{Experiments on Weighted and Non-weighted Cost Function.}
\label{fig:wnw}
\end{figure*}

\subsection{Experiment Results}
\subsubsection{Synthetic Networks}

Fig. \ref{fig:syn5} and \ref{fig:syn1} illustrate our experimental results on synthetic networks, where community ratio $\eta=0.05$ in Fig. \ref{fig:syn5} and $\eta=0.1$ in Fig. \ref{fig:syn1}.
Firstly looking at Fig. \ref{fig:syn5}, with lower community ratio, we observe that: (i) The \emph{average} accuracy of genetic algorithm (GA) under different settings keeps at levels around $40\%-60\%$, which illustrates that based on OSBM, different sizes, densities and whether the communities overlap or not do not make a difference on the performance of GA averagely. This is because GA examines the edges one by one to make the cost function as small as possible, like a greedy algorithm which searches for the local optimum, therefore GA is not seriously affected by the global setting of the networks. (ii) The accuracy of COBA also keeps at a stable level in different situations. However, COBA can only cope with non-overlapping situations, and generally its performance is inferior to GA when communities are not overlapped, which is in line with the results in \cite{Fu:arxiv,Fu:GC}. (iii) The accuracy of CBDA, our algorithm, keeps stable under one specific situation but varies a lot in different networks when the communities overlap each other. This variation is mainly caused by the value of $N$. When the network size $N$ becomes larger, the accuracy of CBDA rises up as well. Specifically, when $N$ goes from $500$ to $2000$, the accuracy rises from approximately $40\%$ to $80\%$. This striking phenomenon demonstrates that our CBDA is suitable for larger size of networks under networks with relatively sparse communities, which corresponds to our Theorem \ref{th3} that as the size of networks becomes larger, the relative NME becomes smaller\footnote{Here when $N$ is larger, the NME is smaller, thus the relative NME becomes smaller as well.}. On the other hand, however, when dealing with non-overlapping situations, our CBDA works stably but not as efficiently as GA or COBA, with the accuracy only around $20\%$.

Now we focus on Fig. \ref{fig:syn1} and compare it with Fig. \ref{fig:syn5}. Fig. \ref{fig:syn1} shows the results under higher community ratio, i.e., denser communities. We can discover that the performance of GA follows that in Fig. \ref{fig:syn5}, which makes sense since, as mentioned above, the performance of GA is not at the mercy of global information like community density. When communities are non-overlapping, the COBA and our CBDA keep similar trends as they do in lower $\eta$, showing that the community density under non-overlapping situations does not affect the performance of all these algorithms. However, what is noticeable is that our CBDA always performs better than other algorithms when the communities are overlapping each other. Moreover, compared with Fig. \ref{fig:syn5} in which $\eta$ is low, the community parameter $a$ is dominant in the accuracy of CBDA when the $\eta$ is high, and when $a=5$ the accuracy can keep stable at around $90\%$. This vivid comparison tells us that our CBDA is very suitable for high accuracy de-anonymization when the community density is large. Moreover, when the community density is large, the performance of CBDA is mainly decided by the edge density ($a$), positively correlated to community density; when the community density is small, then the performance of CBDA is mainly decided by the size of the networks ($N$). This shows that the community ratio (density) determines the dominant factor ($a$ or $N$) in de-anonymization accuracy in networks with overlapping communities.


\subsubsection{Sampled Real Social Networks}
In sampled real social networks, we utilize the real underlying network, thus no modifications on $a$ exist. The results are in Fig. \ref{fig:soc}. We can observe: (i) GA performs better in larger networks and under denser communities, either overlapping or non-overlapping; (ii) The performance of COBA is also enhanced when the size of networks become larger and the community becomes denser; (iii) The performance of CBDA under non-overlapping situations does not outperform other algorithms, but a rising tendency exists as the sampling probability $s$ becomes larger; (iv) The performance of CBDA under overlapping situations still performs well under denser communities and larger network size, with the highest point $95\%$ and the highest average level around $90\%$ when $N=2000$ and $\eta=0.1$, the largest size and densest communities in Table \ref{table:notation2}.

Synthesizing the above four observations, we can learn that the OSBM does not reflect the real social networks very precisely, since the performance of all three algorithms under non-overlapping or overlapping communities differs in two datasets. Moreover, with the same experimental setting, we discover that the performance of our CBDA is better in sampled real social networks than in OSBM-based synthetic networks, which further undergirds the high performance of our algorithm in practical use. Additionally, the results in Fig. \ref{fig:soc} also meet Theorem \ref{th3} that as the network size becomes larger, the relative NME is much smaller and close to $0$, indicating that Theorem \ref{th3} also works in real social networks.

\subsubsection{Cross-Domain Co-author Networks}
In cross-domain co-author networks, we pick up four networks with the same set of $3176$ users.
Fig. \ref{fig:coauthor} illustrates our results. We find that in non-overlapping situation, the results correspond to those in previous datasets that our CBDA does not perform well, while GA and COBA work well. On the other hand, in overlapping situation, we find our CBDA reaches accuracy around $90\%$, outstripping GA whose accuracy is averagely $60\%$. This phenomenon places the significance of our CBDA in a higher level in de-anonymization with overlapping communities since it characterizes the real case totally. Moreover, due to the fact that overlapping situations are much more broadly in real large social networks than non-overlapping situations, our CBDA has wider usage than GA and COBA.

\subsubsection{The Effect of Community Density}
After presenting the results of three basic datasets, we further study the effect of community density on accuracy with more details by using our CBDA. Note that the community ratio $\eta$ directly controls the community density, thus we apply the sampled real social networks under which we can adjust the community ratio $\eta$. We modify $\eta$ from $0.025$ to $0.2$, with interval $0.025$. The results are shown in Fig. \ref{Fig:CommunityDiff}. We can observe that in most cases our CBDA performs better when the network size is larger, which again echoes the conclusion in Theorem \ref{th3}. Moreover, with the larger community ratio, the accuracy of CBDA rises up, showing that CBDA is suitable for social networks with highly overlapping communities. If we observe more carefully, the huge difference of accuracy occurs between $\eta=0.025$ and $\eta=0.075$, and when $\eta\geq0.01$, the accuracy of CBDA under all the network sizes involved keeps at high levels, around $80\%$ or higher. The results further illustrate that the higher community ratio $\eta$, the better de-anonymizing result will be.

\subsubsection{The Instability of Genetic Algorithm}
Now we discuss the weakness of GA in detail. Due to the fact that GA is a heuristic algorithm searching for a local minimum, we will obtain different results when trailing GA multiple times. Fig. \ref{Fig:GAinstable} illustrates the results running GA for $10$ times under real social networks with different sizes. Note that the performance of GA fluctuates violently, for example it swings from $30\%$ to $84\%$ when $N=1000$ and from $42\%$ to $80\%$ when $N=2000$. Therefore, although in average case GA keeps stable at around $40\%$ to $60\%$, users who adopt GA cannot determine whether the solution GA outputs this time is of good or bad quality. This instability in output quality inhibits the usage of GA in practical situations.

\subsubsection{The Effect of Weight Matrix $\mathbf{W}$}
In addition to previous experiments, we intend to supplement a study on the effect of weight matrix $\mathbf{W}$. The purpose of this study is to show that whether minimizing the cost function with $\mathbf{W}$ is of higher accuracy than minimizing the cost function without $\mathbf{W}$, proposed in \cite{cite:seedless}. Embedding $\mathbf{W}$ in the cost function means that We do this experiment under real sampled social networks. Fig. \ref{fig:wnw} illustrates the results. We can observe: (i) The performance of GA does not depend on whether the cost function is appended with $\mathbf{W}$. The curves under weighted and non-weighted cost functions interleave each other. This phenomenon, we suggest, is attributed to the instability of GA. (ii) The performance of our CBDA under weighted cost function is higher than that under non-weighted cost function in almost all the situations. One exception exists when $N=2000$ and $\eta=0.1$. In this situation two curves are almost overlapping each other, which tells us that in larger networks, embedding the community information in the cost function is less significant compared with the increasing network size. In smaller network size ($N\leq 1500$), however, the embedding of community information performs visible increment in accuracy.

\section{Conclusion}\label{sec:conclusion}
We tackle seedless de-anonymization under a more practical social network model parameterized by \emph{overlapping communities} than existing work. By MMSE, we derive a well-justified cost function minimizing the expected number of mismatched users. While showing the NP-hardness of minimizing MMSE, we validly
transform it into WEMP which resolves the tension between optimality and complexity:
(i) WEMP asymptotically returns a negligible mapping error under mild conditions facilitated by higher overlapping strength; (ii) WEMP can be algorithmically solved via CBDA, which exactly finds the optimum of WEMP. Extensive experiments further confirm the effectiveness of CBDA under overlapping communities.

\end{document}